\documentclass [12pt,reqno]{amsart}

\usepackage{savesym}
\usepackage{amsmath}
\usepackage{bbm}
\savesymbol{curves}
\usepackage{curves}
\restoresymbol{TXF}{curves}

\usepackage{mathrsfs}
\usepackage{amsfonts}
\usepackage{txfonts}

\usepackage{verbatim}
\usepackage{graphicx}
\usepackage[all,curve]{xy}
\input xy

\usepackage{indentfirst,latexsym,bm}
\usepackage{newlfont}
\usepackage[top=1in, bottom=1in, left=1.1in, right=1.1in]{geometry}
\usepackage{longtable}
\usepackage{graphics}

\newcounter{pp}
\newtheorem*{BulletTheorem*}{$\bullet$ Theorem}
\newtheorem*{Problem*}{{\bfseries Problem}}

\newtheorem{rem}{{\bfseries Remark}}[section]

\numberwithin{equation}{section}
\newtheorem{cor}{{\bfseries Corollary}}[section]
\newtheorem*{Proposition*}{{\bfseries Proposition}}
\allowdisplaybreaks
\newtheorem{theorem}{Theorem}[section]
\newtheorem{proposition}{Proposition}[section]
\newtheorem{lemma}{Lemma}[section]
\newtheorem{definition}{Definition}[section]

\theoremstyle{definition}

\newtheorem*{Conjecture}{Conjecture}

\begin{document}
\title[Classification of singularities]{4d $\mathcal{N}=2$ SCFT and singularity theory Part II: Complete intersection}

\author{Bingyi \ Chen}
\address{Department of Mathematical Sciences,
Tsinghua University,
Beijing 100084
P. R. China.}
\author{Dan \ Xie}
\address{Center of Mathematical Sciences and Applications\\
Jefferson Physical Laboratory, Harvard University, Cambridge, 02138, USA}
\author{Shing-Tung\ Yau}
\address{Department of Mathematics \\
 Center of Mathematical Sciences and Applications\\
Jefferson Physical Laboratory, Harvard University, Cambridge, 02138, USA}

\author{Stephen S.-T.\ Yau}
\address{Department of Mathematical Sciences,
Tsinghua University,
Beijing 100084
P. R. China.}
%\email{yau@uic.edu}
\author{Huaiqing \ Zuo}
\address{Yau Mathematical Sciences Center, Tsinghua University, Beijing, 100084, P. R. China.}
%\email{hqzuo@math.tsinghua.edu.cn}
\thanks{ The work of S.T Yau is supported by  NSF grant  DMS-1159412, NSF grant PHY-
0937443, and NSF grant DMS-0804454. The work of Stephen Yau and Huaiqing Zuo is supported by NSFC (grant nos. 11401335, 11531007) and  Tsinghua University Initiative Scientific Research Program. 
The work of DX is supported by Center for Mathematical Sciences and Applications at Harvard University, and in part by the Fundamental Laws Initiative of
the Center for the Fundamental Laws of Nature, Harvard University.}

\begin{abstract}
We classify  three dimensional isolated weighted homogeneous rational complete intersection singularities, which define many new  four dimensional $\mathcal{N}=2$ superconformal field theories. 
We also determine the mini-versal deformation of these singularities, and therefore solve the Coulomb branch spectrum and Seiberg-Witten solution. 
\end{abstract}

\maketitle

\section{Introduction}
This is the second  of a series of papers in which we try to classify four dimensional $\mathcal{N}=2$ superconformal field theories (SCFTs) using classification of singularity. 
This program has several interesting features: 
\begin{itemize}
\item The classification of field theory is reduced to the classification of singularities, which in many cases are much simpler than the classification using  field theory 
tools.
\item Many highly non-trivial physical questions such as Coulomb branch spectrum and the Seiberg-Witten solution \cite{SW1,SW2} can be easily found by studying the mini-versal 
deformation of the singularity. 
\end{itemize}

In \cite{XY}, we conjecture that any three dimensional rational Gorenstein graded isolated singularity should define a $\mathcal{N}=2$ SCFT. 
A complete list of hypersurface singularities was obtained in \cite{YY, YY1}, and this immediately gives us a large number of new
four dimensional $\mathcal{N}=2$ SCFTs. 

The natural next step  is to classify three dimensional rational weighted homogeneous isolated complete intersection singularities (ICIS). To our surprise, the 
space of such singularities is also very rich, and we succeed in giving a complete classification. Let's summarize our major findings:
\begin{itemize}
\item The number of polynomials defining ICIS is two, i.e. the singularity is defined as $f_1=f_2=0$. 
\item We find a total of 303 class of singularities, and some of them consist only finite number of models, but we do get 
many infinite sequences. 
\end{itemize}
Our classification gives many new interesting 4d $\mathcal{N}=2$ SCFTs. 
Some of these singularities describe the familiar gauge theory, i.e. the singularity $(f_1, f_2)=(z_1^2+z_2^2+z_3^2+z_4^2+z_5^{2N},z_1^2+2z_2^2+3z_3^2+4z_4^2+5z_5^{2N})$ 
describes the affine $D_5$ quiver gauge theory with SU type gauge group, see figure. \ref{D5}. The major purpose of this paper is to describe the classification, and more detailed study of the 
corresponding SCFTs will appear in a different publication. 
\begin{figure}[h]
  \centering
  \includegraphics[width=2in]{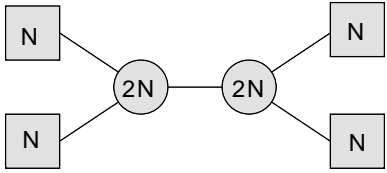}
  \caption
   {4d $\mathcal{N}=2$ SCFT described by the singularity $(f_1, f_2)=(z_1^2+z_2^2+z_3^2+z_4^2+z_5^{2N},z_1^2+2z_2^2+3z_3^2+4z_4^2+5z_5^{2N})$. Here the gauge group is $SU(2N)$.}
   \label{D5}
\end{figure}

This paper is organized as follows: 
Section two  reviews the connection between the physics of 4d $\mathcal{N}=2$ SCFT and the property of ICIS; Section III reviews some preliminary 
facts about  the general property of singularities; Section IV proves that the maximal embedding dimension of 3d rational ICIS is five, and section V classifies the general weighted homogeneous ICIS and 
we also compute the Milnor number and the monomial basis of the mini-versal deformations. 

\section{ICIS and  4d $\mathcal{N}=2$ SCFT}
Four dimensional $\mathcal{N}=2$ SCFT has an important $SU(2)_R\times U(1)_R$ R symmetry. These theories have 
 half BPS scalar operators  ${\cal E}_{r,(0,0)}$ which can get expectation value and parameterize the Coulomb branch, here $r$ is the $U(1)_R$ charge and this operator is a singlet under $SU(2)_R$ symmetry. 
The scaling dimension of  operator ${\cal E}_{r,(0,0)}$ is $\Delta=r$, and if  $1<r<2$, one can turn on the relevant deformation 
\begin{equation}
\delta S=\lambda \int d^4\theta {\cal E}_{r,(0,0)}+c.c.
\end{equation}
The scaling dimension of the coupling constant is determined by the relation $[\lambda]+[ {\cal E}_{r,(0,0)}]=2$. 
The coupling constants do not lift the Coulomb branch, but will change the infrared physics, so we need to include those coupling constants 
besides the expectation values of ${\cal E}_{r,(0,0)}$ to parameterize the Coulomb branch. We also have the dimension one mass parameters $m_i$ which 
can also change the IR physics, which should also be included as the parameters of Coulomb branch. To solve the Coulomb branch of a $\mathcal{N}=2$ SCFT, we need to achieve the following two goals:
\begin{itemize}
\item Determine the set of rational numbers which include the scaling dimension of the coupling constants $\lambda$, Coulomb branch operators  ${\cal E}_{r,(0,0)}$, and mass parameters $m_i$. 
\begin{equation}
(r_1,\ldots, 1,\ldots,1, \ldots r_{\mu}).
\end{equation}
Since the scaling dimension of the coupling constant is paired with that of Coulomb branch operator, this set is symmetric with respect to identity.
\item Once we find out the parameters on the Coulomb branch, we want to write down a Seiberg-Witten curve which describes the low energy effective 
theory on the Coulomb branch 
\begin{equation}
F(z_i, u)=0.
\end{equation}
Here $u$ includes all the parameters discussed above. 
\end{itemize}
These two questions are central in understanding Coulomb branch of a $\mathcal{N}=2$ SCFT, and in general are quite hard to answer.  

If our 4d SCFT is engineered using 3-fold singularity, the above two questions can be found from the miniversal deformation of the singularity\footnote{If the theory is engineered using M5 branes \cite{Ga,GMN, NX,CDY, Xie, WX}, these questions are solved by spectral curve of Hitchin system. If 
the theory is engineered using Kodaira singularity \cite{ALLM1,ALLM2,ALLM3}, these questions can also be solved by studying the deformations of the singularity.}. The formulae for the hypersurface case  has been given in \cite{SV, CNV}.
In the following, we will review the relevant formulas for ICIS case, which is first derived in \cite{XY}.

Consider a three dimensional ICIS defined by two polynomials $f=(f_1,f_2)$, where$f$ is the map $f:(C^5,0)\rightarrow (C^2,0)$ \footnote{ICIS defined by two polynomials is enough for our purpose, see section IV.}. We require the defining polynomials
to have a manifest $\mathbb{C}^*$ action, which is proportional to the $U(1)_R$ symmetry of the field theory. We normalize the $\mathbb{C}^*$ action so that the weights of the coordinates $(z_1,\ldots, z_5)$ are 
$(w_1,w_2,\ldots, w_5)$, and the degree of $f_1$ is one, the degree of $f_2$ is d:
\begin{equation}
f_1(\lambda^{w_i}z_i)=\lambda f_1(z_i),~~~f_2(\lambda^{w_i}z_i)=\lambda^d f_2(z_i).
\end{equation} 
This singularity has a distinguished $(3,0)$ form: 
\begin{equation}
\Omega={dz_1\wedge dz_2\wedge \ldots \wedge dz_5\over df_1\wedge df_2},
\end{equation}
which has charge $\sum w_i-1-d$ under the $\mathbb{C}^*$ action. To define a sensible 4d SCFT, we require this charge to be positive, which means that 
\begin{equation}
\sum w_i>1+d.
\end{equation}
We conjecture that this condition is necessary and sufficient to define a SCFT. Such singularity is called rational singularity, see section III for the definition.  
The SW solution is described by the mini-versal deformation of the singularity:
\begin{equation}
F(\lambda, z_i)=f(z_i)+\sum_{\alpha=1}^\mu \lambda_{\alpha} \phi_{\alpha},
\label{SW}
\end{equation}
here $\phi_{\alpha}$ is the monomial basis of the Jacobi module of $f$, and $\mu$ is the Milnor number. The coefficient $\lambda_{\alpha}$ is identified 
with the parameters on Coulomb branch. The scaling dimension of $\lambda_{\alpha}$ is determined by requiring $\Omega$ to have dimension one as its integration over
the middle homology cycle of Milnor fibration gives the mass of BPS particle, and we have 
\begin{align}
&\phi_{\alpha}=[\phi_{\alpha},0]:~~[\lambda_{\alpha}]={1-Q_{\alpha}\over \sum w_i-1-d}, \nonumber\\
&\phi_{\alpha}=[0,\phi_{\alpha}]:~~[\lambda_{\alpha}]={d-Q_{\alpha}\over \sum w_i-1-d}.
\label{spectrum}
\end{align}
Here $Q_{\alpha}$ is the $\mathbb{C}^*$ charge of the monomial $\phi_{\alpha}$. 
The spectrum is classified into  following categories:
\begin{itemize}
\item  Coulomb branch operator ${\cal E}_{r,(0,0)}$ if $[\lambda_{\alpha}]>1$.
\item  Mass parameters if $[\lambda_{\alpha}]=1$. 
\item  Coupling constants for relevant deformations if $0<[\lambda_{\alpha}]<1$.
\item  Exact marginal deformations if $[\lambda_{\alpha}]=0$. These deformations are related to the moduli of the singularity.
\item  Irrelevant deformations if $[\lambda_{\alpha}]<0$. 
\end{itemize}
The spectrum is paired and is symmetric with respect to one, which is in perfect agreement with the field theory expectation.

\textbf{Example}: Consider the singularity $f=(f_1,f_2)=(z_1^2+z_2^2+z_3^2+z_4^2+z_5^2, z_1^2+2z_2^2+3z_3^2+4z_4^2+5z_5^2)=0$.  The weights and degrees of two polynomials are 
$({1\over 2},{1\over2},{1\over2},{1\over2},{1\over2};1,1)$. 
The Jacobi module $J_f$
has the basis 
\begin{align}
&\phi_1=[0,z_5^2],~\phi_2=[0,z_4^2],~\phi_3=[0,z_5],~\phi_4=[0,z_4],~\phi_5=[0,z_3], \nonumber\\
&\phi_6=[0,z_2],~\phi_7=[0,z_1],~\phi_8=[0,1],~\phi_9=[1,0].
\end{align}
Using the formula (\ref{spectrum}), we find the scaling dimension of the coefficients $\lambda_i$ in mini-versal deformation:
\begin{equation}
[\lambda_1]=[\lambda_2]=0,~[\lambda_3]=[\lambda_4]=[\lambda_5]=[\lambda_6]=[\lambda_7]=1,~[\lambda_8]=[\lambda_9]=2; 
\end{equation}
So this theory has two exact marginal deformations, two Coulomb branch operators with dimension 2, and five mass parameters. 
The corresponding gauge theory is depicted in figure. \ref{D5} with $N=1$. 

In the following sections, we will classify all possible 3 dimensional weighted homogeneous ICIS, and describe the miniversal deformation. 
The Coulomb branch is then solved using formulas \ref{SW} and \ref{spectrum}.
\section{Preliminaries}

In this section, we recall some definitions and known results about the Gorenstein singularity. 

\begin{definition}
For a commutative Noetherian local ring $R$, the depth of $R$ (the maximum length of a regular sequence in the maximal ideal of $R$) is at most the Krull dimension of $R$. The ring $R$ is called Cohen-Macaulay if its depth is equal to its dimension. More generally, a commutative ring is called Cohen-Macaulay if it is Noetherian and all of its localizations at prime ideals are Cohen-Macaulay. In geometric terms, a scheme is called Cohen-Macaulay if it is locally Noetherian and its local ring at every point is Cohen-Macaulay.
\end{definition}

\begin{definition}
Let $(X, x)$ be an isolated singularity of dimension $n$. $(X, x)$ is said to be normal or Cohen-Macaulay if the local ring $\mathcal{O}_{X, x}$ has such a property.
\end{definition}

Let $(X, x)$ be an isolated singularity of dimension $n$. Then we have the following propositions.

\begin{proposition}(Corollary 3.10, \cite{Ha})
 $(X, x)$ is called Cohen-Macaulay iff $H^1_{\{x\}}(\mathcal{O}_{X})=\cdots=H^{n-1}_{\{x\}}(\mathcal{O}_{X})=0$. $(X, x)$ is normal iff $H^1_{\{x\}}(\mathcal{O}_{X})=0$.
 \end{proposition}
 
 \begin{proposition} Let $(X, x)$ be an isolated singularity of dimension $n$ and  $\pi: (\widetilde{X}, E)\rightarrow (X, x)$ be a resolution of $(X, x)$. Then
$H^i(\widetilde{X}, \mathcal{O})\cong H^i_{\infty}(\widetilde{X}, \mathcal{O}) \cong H^{i+1}_{\{x\}}(X, \mathcal{O}), 1\leq i\leq n-2.$
\end{proposition}
\begin{proof}

Following Laufer \cite{La1}, we consider the sheaf cohomology with support at infinity. The following sequence is exact:

$$0\rightarrow \Gamma(\widetilde{X}, \mathscr{O}) \rightarrow \Gamma_{\infty}(\widetilde{X}, \mathscr{O})\rightarrow H^1_{c} (\widetilde{X}, \mathscr{O})\rightarrow H^1(\widetilde{X}, \mathscr{O}) \rightarrow H^1_{\infty}(\widetilde{X}, \mathscr{O})\rightarrow H^2_{c} (\widetilde{X}, \mathscr{O})\rightarrow \cdots $$

By Serre duality,
$$ H^i_{c} (\widetilde{X}, \mathscr{O})\cong  H^{n-i} (\widetilde{X}, \mathscr{O}(K)).$$
by Grauert-Riemenschneider Vanishing Theorem, we have $$H^{n-i}(\widetilde{X}, \mathscr{O}(K))=0, \, \textup{for}\, \,  i\leq n-1.$$  It follows that 
\begin{align}\label{2.1}
H^i(\widetilde{X}, \mathcal{O})\cong H^i_{\infty}(\widetilde{X}, \mathcal{O}),\, 1\leq i\leq n-2.
\end{align}

On the other hand, we also  have the following exact sequence: 
 
\begin{align*}
0\rightarrow \Gamma(X, \mathscr{O}) \rightarrow \Gamma(X-\{x\}, \mathscr{O})\rightarrow H^1_{\{x\}} (X, \mathscr{O})\rightarrow H^1(X, \mathscr{O}) \rightarrow H^1(X-\{x\}, \mathscr{O})\\ \rightarrow H^2_{\{x\}} (\widetilde{X}, \mathscr{O})\rightarrow \cdots
\end{align*} 

Since $H^i(X, \mathscr{O})=0,\, i\geq 1$. Thus we have 

\begin{align}\label{2.2}
 H^i(X-\{x\}, \mathscr{O}) \cong H^{i+1}_{\{x\}}(X, \mathcal{O}),\, i\geq 1.
 \end{align}

 Take a 1-convex exhaustion function $\phi$ on $\widetilde{X}$ such that $\phi\geq 0$ on $\widetilde{X}$ and $\phi(y)=0$ if and only if $y\in E_i$ where $E_i$ is the irreducible component of $E$. Put $\widetilde{X}_r=\{y\in \widetilde{X}: \phi(y)\leq r\}$. Then by Laufer \cite{La1}, 
$$\varinjlim_{r} H^q(\widetilde{X}-\widetilde{X}_r, \mathcal{O})\cong H^q_{\infty}(\widetilde{X}, \mathcal{O}).$$
 On the other hand, by Andreoti and Grauert (Th\'{e}or\`{e}me 15 of \cite{An-Gr62}), $H^{q}(\widetilde{X}-E, \mathcal{O})$ is isomorphic to $H^{q}(\widetilde{X}-\widetilde{X}_r, \mathcal{O})$ for $q\leq n-2$ and $H^{n-1}(\widetilde{X}-E, \mathcal{O})\rightarrow H^{n-1}(\widetilde{X}-\widetilde{X}_r, \mathcal{O})$ is injective. Thus we have
\begin{align}\label{2.3}
 H^i_{\infty}(\widetilde{X}, \mathcal{O})\cong H^{i}(\widetilde{X}-E, \mathcal{O})\cong  H^i(X-\{x\}, \mathscr{O}),\, 1 \leq i\leq n-2.
 \end{align}
 
Combining with (\ref{2.1}),(\ref{2.2}) and (\ref{2.3}) we have 
$$H^i(\widetilde{X}, \mathcal{O})\cong H^i_{\infty}(\widetilde{X}, \mathcal{O}) \cong H^{i+1}_{\{x\}}(X, \mathcal{O}), 1\leq i\leq n-2.$$

This completes the proof.
\end{proof}

\begin{cor}
(1) $(X, x)$ is Cohen-Macaulay $\Rightarrow  H^i(\tilde{X}, \mathcal{O})=0, 1\leq i\leq n-2$. \\
(2) $(X, x)$ is normal and $ H^i(\tilde{X}, \mathcal{O})=0, 1\leq i\leq n-2 \Rightarrow (X, x)$ is  Cohen-Macaulay.
\end{cor}

\begin{definition}
A normal variety $X$ is called Gorenstein if it is Cohen-Macaulay and the sheaf $\omega_X:=\mathcal{O}(K_X)$ is locally free.
\end{definition}

\begin{definition}
A Gorenstein point $x \in X$ of an $n$-dimensional variety $X$ is rational (respectively minimally elliptic) if for a resolution $f : Y\rightarrow X$ we have
$f_{*}\omega_ Y= \omega_ X$(respectively $f_{*}\omega_ Y= m_x\omega_ X$, where $m_x$ is the ideal of $x$). (This is equivalent via duality to the cohomological assertion $R^{n-1}f_{*}\mathcal{O}_Y	=
0$ (respectively, is a 1-dimensional $\mathbb{C}$-vector space at $x$). 
\end{definition}

It is convenient to make intrinsic (and generalize slightly) the notion of “a general hyperplane section through $x$”:

\begin{definition}
Let $(\mathcal{O}_X, x, m_x) $ be the local ring of a point $x \in X$ of a $\mathbb{C}$- scheme, and let $V \subset m_x$ be a finite-dimensional $\mathbb{C}$-vector space which maps onto $m_x /m^2_x$ (equivalently, by Nakayama’s lemma, $V$ generates the $\mathcal{O}_{X, x}$-ideal $m_x$ ); by a general hyperplane section through $x$ is mean the sub-scheme $H \subset X_0$ defined in a suitable neighborhood $X_0$ of $x$ by the ideal  $\mathcal{O}_{X, v}$, where $v \in V$ is a sufficiently general element (that is, $v$ is a $\mathbb{C}$-point of a certain dense Zariski open $U\subset V$).
\end{definition}

\begin{theorem}(\cite{Mi}, Theorem 2.6)\label{Mth}
 If $x \in X$ is a rational Gorenstein point (dim $X=n\geq 3$. Then for a general hyperplane section $S$ through $x$, $x\in S$ is minimally elliptic or rational Gorenstein.
\end{theorem}
\begin{proof}
Suppose that $S$ runs through any linear system of sections $x\in S \subset X$ whose equations generate the maximal ideal $m_x$ of $\mathcal{O}_{X, x}$. Then a general element $S$ of this linear system is normal.

Let $f: Y\rightarrow X$ be any resolution of $X$ which dominates the blow-up of the maximal ideal $m_x$; by definition of the blow-up, the scheme-theoretic fiber over $x$ is an effective divisor $E$ such that $m_x\mathcal{O}_{Y}=\mathcal{O}_{Y}(-E).$ Hence $f^*S=T+E$, where $T$ runs through a free linear system on $Y$. By Bertini's theorem, $\phi=f\mid_T: T\rightarrow S$ is a resolution of $S$. Now we use the adjunction formula to compare $K_T$ and $\varphi^*K_S$.

In the diagram (Figure 1),
\begin{figure}[h!]
\vspace*{-12pt}
\begin{equation*}
\xymatrix{
Y\ar@{->}[d]^{f}
&\supset&
T \ar@{->}[d]^{\varphi}+E
\\
X&\supset& S 
}
\end{equation*}
\caption{}\label{fig1} 
\end{figure}
we have  $$K_Y=f^*K_X+\Delta,\, \textup{with}\, \Delta\geq 0$$ and 
$$T=f^*S-E,$$
so that 
$$K_Y+T=f^*(K_X+S)+\Delta-E$$ and 
$$K_T=(K_Y+T)\mid_T=\varphi^*K_S+(\Delta-E)\mid_T.$$
This just means that any $s\in \omega_s$ has at worst $(\Delta-E)\mid_T$ as pole on $T$. On the other hand, since the maximal ideal $m_{S, x}\subset \mathcal{O}_{S, x}$ is the restriction to $S$ of the maximal ideal $m_{X, x}\subset \mathcal{O}_{X}$, it follows that every element of $m_{S, x}$ vanishes along $E\cap T$. Hence every element of $m_{S, x} \omega_S$ is regular on $T$, that is 
$$ m_x\omega_S\subset \varphi_{*}w_T\subset \omega_S.$$ 
Thus $m_x\omega_S=\varphi_{*}w_T$ implies $x\in S$ is minimally elliptic and $\omega_S=\varphi_{*}w_T$ implies $x\in S$ is rational.
\end{proof}

\begin{theorem}\cite{La2} \label{Lth}
Let $x$ be a minimally elliptic singularity.  Let $\pi\colon M\to V$ be a resolution of a Stein neighborhood  $V$ of $x$ with $x$ as its only singular point.  Let $m$ be the maximal ideal in $\mathcal{O}_{V,x}$.  Let $Z$ be the fundamental cycle on $E=\pi^{-1}(x)$.
\begin{itemize}
\item[(1)]
If $Z^2\le -2$, then $\mathcal{O}(-Z)=m\mathcal{O}$ on $E$.
\item[(2)]
If $Z^2=-1$, and $\pi$ is the minimal resolution or the minimal resolution with non-singular $E_i$ and normal crossings, $\mathcal{O}(-Z)\big/m\mathcal{O}$ is the structure sheaf for an embedded point.
\item[(3)]
If $Z^2=-1$ or $-2$, then $x$ is a double point.
\item[(4)]
If $Z^2=-3$, then for all integers $n\ge 1$, $m^n\approx H^0\bigl(E,\mathcal{O}(-nZ)\bigr)$ and $\dim m^n/m^{n+1}=-nZ^2$.
\item[(5)]
If $-3\le Z^2\le -1$, then $x$ is a hypersurface singularity.
\item[(6)]
If $Z^2=-4$, then $x$ is a complete intersection and in fact a tangential complete intersection.
\item[(7)]
If $Z^2\le -5$, then $x$ is not a complete intersection.
\end{itemize}
\end{theorem}

\subsection{Deformation of singularities}

Let $(X_0, x_0)$ be an isolated singularity with dimension $n$, a deformation of $(X_0, x_0)$ will be simply a realization of $(X_0, x_0)$ as the fiber of a map-germ between complex manifolds whose dimensions differ by $n$. To be precise, it consists of holomorphic map-germ $f:(X, x)\rightarrow (S, o)$ between complex manifold germs with dim$(X, x)$-dim$(S, o)=n$ and an isomorphism $\iota:$ of $(X_0, x_0)$ onto the fiber $(X_{o}, x)$ of $f$. A morphism from a deformation $(\iota', f')$ to another $(\iota, f)$ is a pair of map-germs $(\tilde{g}, g)$ such that the diagram
  \begin{figure}[h!]
\vspace*{-10pt}
\begin{equation*}
\xymatrix{
(X', x')\ar@{->}[d]^{f'}
&\stackrel{\tilde{g}}{\longrightarrow}&
(X, x) \ar@{->}[d]^{f}
\\
(S', o')&\stackrel{g}{\longrightarrow}& (S, o)
}
\end{equation*}
%\caption{}\label{fig2} 
\end{figure}
\\
is Cartesian and $\tilde{g}\circ\iota'=\iota$. We say that a deformation $(\iota, f)$ of $(X_0, x_0)$ is versal if for any deformation $(\iota', f')$ of $(X_0, x_0)$ there exists a morphism $(\tilde{g}, g)$ from  $(\iota', f')$ to $(\iota, f)$. Notice that we do not require this morphism to be unique in any sense. If, however, the derivative of $g$ in $o'$, $\partial g(o'): T_{o'}(S')\rightarrow T_o(S)$ is unique, then we say that $(\iota, f)$  is miniversal.

\begin{proposition}(\cite{AGLV} (2.10))
Let  $f:(\mathbb{C}^{n+k},0)\rightarrow (\mathbb{C}^k,0)$ define an 
icis at the origin and has dimension $n$. A miniversal deformation of  $f=0$ can be taken in the form 
\begin{equation}
F(z,\lambda)=f(z)+\lambda_1e_1(z)+\ldots+\lambda_{\tau}e_\tau(z),
\end{equation}
where $e_i\in{\cal O}_{n+k}^k$ are the representative of a basis of the linear space:
\begin{equation}
T_f^1={\cal O}_{n+k}^k/\{I{\cal O}_{f^{n+k}}^k+{\cal O}_{n+k}\langle \partial f/\partial z_1,\ldots,\partial f/ \partial z_{n+k}\rangle\}.
\end{equation}
\end{proposition}
Here $\tau$ is the Tyurina number and is equal to the Milnor number $\mu$ if $f$ is weighted homogeneous.

\section {Homogeneous isolated complete intersection singularity}

In this section we shall prove the following conjecture for homogeneous isolated complete intersection singularity (ICIS) in Theorem \ref{th1} and three dimensional isolated complete intersection singularity in Theorem \ref{t1.2}. We shall also give a  classification of three-dimensional rational homogeneous  isolated complete intersection singularities in Theorems \ref{th2} and \ref{th3}. 

\begin{Conjecture}
Let $p$ be the dimension of rational isolated  complete intersection  singularity with $\mathbb{C}^*$-action. The the embedding dimension of the singularity  is at most $2p-1$. 
\end{Conjecture}

\begin{definition}
Let $(V, 0)\subset (\mathbb{C}^N, 0)$ be the analytic germ of an $n$-dimensional complex homogeneous  isolated complete intersection singularity. Let $\pi:(M, E)\rightarrow (V,0)$ be a resolution of singularity of dimension $n$ with exceptional set $E=\pi^{-1}(0)$.  The geometric genus $p_g$ of the singularity $(V,0)$ is the dimension of $H^{n-1}(M, \mathcal{O})$ and is independent of the resolution $M$.
\end{definition}

We have the following proposition.

\begin{proposition}\cite{KN}\label{prop1}
Let $(V, 0)=\{f_1=\cdots=f_r=0\}$ be a homogeneous isolated complete intersetion singularity of multidegree $(d_1,\cdots, d_r)$ and dimension $n$, that is deg$f_i=d_i$, then 
$$p_g=\sum_{\underline{k}\in K_{n, r}}\prod_{i=1}^r {d_i\choose k_i+1},$$
where $K_{n, r}:=\{\underline{k}=(k_1,\cdots, k_r): k_i\geq 0 \; \textup{for all}\; i, and\, \sum_i k_i=n\}.$
\end{proposition}

We prove that the above conjectue is true in homogeneous case.

\begin{theorem}\label{th1}
Let $(V, 0)=\{f_1=\cdots=f_r=0\}\subset(\mathbb{C}^N, 0)$ be a homogeneous rational isolated complete intersetion singularity of multidegree $(d_1,\cdots, d_r)$ and dimension $n$, then $r\leq n-1$ (i.e. $N\leq 2(N-r)-1$).
\end{theorem}

\begin{proof}
Since $(V, 0)$ is a homogeneous isolated rational complete intersetion singularity, so by Proposition \ref{prop1}, we have $$p_g=\sum_{\underline{k}\in K_{n, r}}\prod_{i=1}^r {d_i\choose k_i+1}=0.$$ Thus for any 
$\underline{k}\in K_{n, r}$, we have $\prod_{i=1}^r {d_i\choose k_i+1}=0$. If we assume the contrary, $r\geq n$,  without loss of generality, we consider thel $\underline{k}=(k_1,\cdots, k_r) $ with $k_1=k_2=\cdots=k_{n}=1, k_{n+1}=\cdots=k_{r}=0$.  Then for this choice $\underline{k}=(1,\cdots,1, 0,\cdots,0)$,  we have $\prod_{i=1}^r {d_i\choose k_i+1}\geq 1$ since $d_i\geq 2$. This contradicts with  $\prod_{i=1}^r {d_i\choose k_i+1}=0$. Therefore we have $r\leq n-1$.
\end{proof}
 We have the following two classification theorems for homogeneous case.
\begin{theorem}\label{th2}
Let $(V, 0)=\{f_1=\cdots=f_r=0\}\subset(\mathbb{C}^N, 0)$ be a three dimensional homogeneous rational isolated  complete intersection singularity of multidegree $(d_1,\cdots, d_r)$ which is not a hypersurface singularity, then $r=2, N=5$ and $d_1=d_2=2$.
\end{theorem}
\begin{proof}
If follows from Theorem \ref{th1} that $r\leq 2$. Since $(V, 0)$ is not a hypersurface singularity, so $r=2$. Let $\underline{k}=(k_1, k_2)\in K_{3, 2}$, we have $(k_1, k_2)=(0, 3), (3, 0), (1, 2)$ or $(2, 1)$.  For $(k_1, k_2)=(1, 2)$,   by Proposition \ref{prop1} $$p_g=\sum_{\underline{k}\in K_{3, 2}}\prod_{i=1}^2 {d_i\choose k_i+1}=0,$$  we have $$ {d_1\choose 2} {d_2\choose 3}=0$$ which implies $ d_2=2$ since $d_1$ and $d_2$ are at least 2.  Similarly,  $(k_1, k_2)=(2, 1)$ implies $d_1=2$. Thus $d_1=d_2=2$.
\end{proof}

\begin{theorem}\label{th3}
Let $(V, 0)=\{f=0\}\subset(\mathbb{C}^4, 0)$ be a three dimensional homogeneous rational isolated hypersurface singularity of degree $d$, then $d=2, 3$.
\end{theorem}

\begin{proof}
By Proposition \ref{prop1}, we have  $$p_g= {d\choose 4}=0,$$ so we  $ d=2$ or $3$.
\end{proof}

We can also prove the conjecture for $p=3$.

\begin{lemma}\label{l1.1} Let $(V,0)$ be a $n$-dimensional isolated singularity in $\mathbb{C}^N$. If $2n-N>0$, then $(V,0)$ cannot have two components of dimension $n$.
\end{lemma}

\begin{proof} If $V$ is union of two components $V_1$ and $V_2$, each of which is of dimension $n$, then $V_1$ and $V_2$ will intersect with at least dimension $2n-N>0$. This is a contradiction since $V$ is singular along the intersection.
\end{proof}

\begin{theorem} \label{t1.2} Let $(V,0)$ be a three dimensional rational isolated  complete intersection singularity. Then the embedding dimension of $(V, 0)$ is at most 5.
\end{theorem}

\begin{proof} Take a generic section $(H,0)$ of $(V,0)$. Then by Theorem \ref{Mth},  $(H,0)$ is either a 2-dimensional rational Gorenstein singularity or minimally elliptic singularity. It is well-known that 2-dimensional rational Gorenstein singularity must be rational double points. So the embedding dimension of $(H,0)$ is 3. On the other hand, by Theorem \ref{Lth} 
asserts that minimally elliptic complete intersection isolated singularity has embedding dimension at most 4. So the  embedding dimension of $(V,0)$ is at most 5.
\end{proof}
 An immediate corollary is as follows.
 
 \begin{cor}\label{c1.1}
The three dimensional weighted homogeneous rational isolated  complete intersection singularity  which is not hypersurface singularity is defined by two weighted homogeneous polynomials in 5 variables.
 \end{cor}

\begin{lemma}Both $f_1$ and $f_2$ in Corollary \ref{c1.1} are irreducible.
\end{lemma}

\begin{proof}Assume that $f_1=f_{11}f_{12}\cdots f_{1k_1}$ and $f_2=f_{21}f_{22}\cdots f_{2k_2}$,where $f_{1i}, 1\leq i\leq k_1$ and $f_{2j}, 1\leq j\leq k_2$ are irreducible. $V(f_{1i}, f_{2j})$ are $k_1k_2$ irreducible components of dimension 3. Since $n=3, N=5$, so $2n-N=1>0$. By Lemma \ref{l1.1}, we have $k_1k_2=1$, which implies $k_1=k_2=1$, thus both $f_1$ and $f_2$ are irreducible.
\end{proof}

 \section{Classification of three dimensional weighted homogeneous rational isolated  complete intersection singularity }

In this section, we shall give an complete classification of three-dimensional rational weighted homogeneous complete intersection singularities. We first recall some definitions and then we prove some properties which are used in the proof of classification theorem.

\begin{definition}
Let $w = (w_1, \cdots ,w_n; d)$ be an $(n+1)$-tuple of positive
rational numbers. A polynomial $f(z_1, \cdots ,z_n)$ is said to be a weighted homogeneous polynomial with weights $w$ if each monomial $\alpha z_1^{a_1}z_2^{a_2} \cdots z_n^{a_n}$ of $f$ satisfies $a_1w_1 +\cdots+ a_nw_n = d$.  And we say a pair of polynomials $(f_1,f_2)$ are weighted homogeneous of type $(w_1,\cdots,w_n;d_1,d_2)$ if $f_1$ is weighted homogeneous of type $(w_1,\cdots,w_n;d_1)$ and $f_2$ is weighted homogeneous of type $(w_1,\cdots,w_n;d_2)$.
\end{definition}

By Corollary \ref{c1.1}, in order to classify three dimensional weighted homogeneous rational isolated  complete intersection singularity, we only need to study weighted homogeneous polynomials in 5 variables.

\begin{theorem}\cite{Ma} \label{Ma}
Let $X=V(f_1,\cdots, f_k)$ be an weighted homogeneous  ICIS  of type $$(w_1,\cdots, w_n;d_1,\cdots,d_k).$$
Let $$A(N)=\{(a_1,\cdots,a_n)\in {\mathbb N}^n |a_i>0~\textup{and}~\sum\limits_{i = 1}^k {{a_i}{w_i} \le N} \}$$ and $$\ell(N)=\# A(N).$$
Then we have $$ {p_g}(X) = \ell({d_1}+\cdots+{d_k}) - \sum\limits_{i = 1}^k {\ell({d_1} + \cdots + {{\hat d}_i} + \cdots+ {d_k})}  + \cdots + {( - 1)^{k - 1}}\sum\limits_{i = 1}^k {\ell({d_i})}.$$
\end{theorem}

The following observations plays key role in our proof.
\begin{lemma}\label{l2}
Let $X=V(f_1,f_2)$ be a weighted homogeneous ICIS of type $(w_1,\cdots,w_5;d_1,d_2)$. Then $X$ is rational if and only if $$w_1+\cdots+w_5>d_1+d_2.$$
\end{lemma}
\begin{proof}
By Theorem \ref{Ma},we have $p_g(X)=\ell(d_1+d_2)-\ell(d_1)-\ell(d_2)$.
Thus $X$ is rational if and only if $\ell(d_1+d_2)-\ell(d_1)-\ell(d_2)=0.$  It is easy to see that $w_1+\cdots+w_5>d_1+d_2$ implies $\ell(d_1+d_2)=\ell(d_1)=\ell(d_2)=0$, it follows that $X$ is rational.
If $X$ is rational, and without lose of generality, we assume that $\ell(d_1)>0$, then $A(d_1)$ is not empty. Let $\textbf{a}_{\max}\in A(d_1)$ such that $\textbf{w}\textbf{a}_{\max}\ge \textbf{w} \textbf{a}$ for any $\textbf{a}\in A(d_1)$, where $\textbf{w} \textbf{a}=\sum\limits_{i = 1}^5 {w_ia_i}$ for $\textbf{a}=(a_1,\cdots, a_5)$. Then let $B= \{\textbf{a}_{\max}+\textbf{b} \mid \textbf{b} \in A(d_2)\}$ so we have  $d_1<\textbf{w} \textbf{c}\le  d_1+d_2$ for any $\textbf{c} \in B$ and $\#B=\ell(d_2)$. It is easy to seen that there exist $i \in \{1,2,3,4,5\}$ such that $w_i \le d_2$, because if not then $f_2=0$. Without lose of generality, we may assume that $w_1 \le d_2$. Let $\textbf{d}=\textbf{a}_{\max}+(1,0,0,0,0)$, then $d_1< \textbf{w} \textbf{d}\le  d_1+d_2$. Notice that $\textbf{d} \in A(d_1+d_2)\setminus (A(d_1)\cup B)$ and $A(d_1)\cap B=\emptyset$, thus we have $\ell(d_1+d_2)\ge \ell(d_1)+\ell(d_2)+1$. It follows that $p_g(X) \ge 1$. It contradicts with $X$ is rational, so we conclude that $\ell(d_1)=0$. Similarly we can prove that $\ell(d_2)=0$. So $p_g(X)=\ell(d_1+d_2)=0$ which implies $w_1+\cdots+w_5>d_1+d_2$.
\end{proof}

\begin{lemma}\label{l3}
Let $X=V(f_1,f_2)$ be a three dimensional weighted homogeneous ICIS of type $(w_1,\cdots, w_5;d_1,d_2)$. Then we have 
\\(1)~for any $i \in \{ 1,2,3,4,5 \}$, one of the following cases occurs:
  \\ ~ (1a)~$z_i^n$ appears in $f_1$ for some n,
  \\ ~ (1b)~$z_i^n$ appears in $f_2$ for some n,
  \\ ~ (1c)~there exist $j, k \in \{1,2,3,4,5\}\setminus \{i\} \; j \ne k$ such that $z_i^nz_j$ appears in $f_1$ for some $n$ and $z_i^mz_k$ appears in $f_2$ for some $m$.
\\
\\(2)~ for any $l=1,2$ and any $\{i,j\} \subset \{1,2,3,4,5\}$, one of the following cases occurs:
   \\~~~(2a)~$z_i^az_j^b$ appears in $f_l$, for some non-negative integers $a,b$,
   \\~~~(2b)~there exists $k \in \{1,2,3,4,5\} \setminus \{i,j\}$ such that $z_kz_i^az_j^b$ appears in $f_l$, for some non-negative integers $a,b$.
\\
\\(3)~ for any $\{i,j\} \subset \{1,2,3,4,5\}$, one of  the following cases occurs:
   \\~~~(3a)~$z_i^az_j^b$ appears in $f_1$  for some non-negative integer $a,b$,
   \\~~~(3b)~$z_i^az_j^b$ appears in $f_2$  for some non-negative integer $a,b$,
   \\~~~(3c)~there exist $\{p_1,p_2\},\{s_1,s_2\} \subset \{1,2,3,4,5\}\setminus \{i,j\}$ and $\{p_1,p_2\}\ne\{s_1,s_2\}$ such that $z_{p_1}z_i^{a_1} z_j^{b_1}, z_{p_2}z_i^{a_2}z_j^{b_2}$ appear in $f_1$ for some non-negative integers  $a_1,a_2,b_1,b_2$  and \\ $z_{s_1}z_i^{c_1}z_j^{d_1}, z_{s_2}z_i^{c_2} z_j^{d_2}$ appear in $f_2$ for some non-negative integers $c_1,c_2,d_1,d_2$.
\\
\\
(4)~ for any $l=1,2$ and any $\{i,j,k\}  \subset \{1,2,3,4,5\}$, let $\{p,s\}=\{1,2,3,4,5\}\setminus \{i,j,k\}$, then one of the following cases occurs:
   \\~~~(4a)~$z_i^az_j^bz_k^c$ appears in $f_l$, for some non-negative integers $a,b,c$,
   \\~~~(4b)~$z_pz_i^{a_1}z_j^{b_1} z_k^{c_1}$ and $z_s z_i^{a_2} z_j^{b_2} z_k^{c_2}$ appear in $f_l$, for some non-negative integers $a_1,b_1,c_1, a_2,b_2,c_2$.
\\
\\
(5)~for any $\{i,j,k\}\subset\{1,2,3,4,5\}$, there exists $l\in\{1,2\}$ such that $z_i^a z_j^b z_k^c$ appears in $f_l$ for some non-negative integers $a,b,c$.
\end{lemma}

\begin{proof}
(1)~Without lose of generality, we may assume that $i=1$. Assume on the contrary  that neither of $(1a), (1b)$ and $(1c)$  occurs. Then $z_1^n$ does not appear in $f_l$ for any $l=1,2$ and integer $n$, so we have $f_1=f_2=\frac{{\partial {f_1}}}{{\partial {z_1}}}=\frac{{\partial {f_2}}}{{\partial {z_1}}}=0$ when $z_2=z_3=z_4=z_5=0$. And for any $\{j, k\}\subset \{2,3,4,5\}$, we have $z_1^az_j$ doesn't appear in $f_1$ for any non-negative integer $a$ or  $z_1^bz_k$ doesn't appear in $f_2$ for any non-negative integer $b$. It follows that $\frac{\partial f_1}{\partial z_j}=0$ or $\frac{\partial f_2}{\partial z_k}=0$ when $z_2=z_3=z_4=z_5=0$. Similarly we have  $\frac{\partial f_1}{\partial z_k}=0$ or $\frac{\partial f_2}{\partial z_j}=0$ when $z_2=z_3=z_4=z_5=0$. Thus we have 
$$
\frac{\partial f_1}{\partial z_j}\frac{\partial f_2}{\partial z_k}-\frac{\partial f_1}{\partial z_k}\frac{\partial f_2}{\partial z_j}=0,\;\,\forall \;\{j, k\}\subset \{2,3,4,5\},  z_2=z_3=z_4=z_5=0,
$$ 
which implies $(\frac{{\partial {f_1}}}{{\partial {z_1}}},\cdots,\frac{{\partial {f_1}}}{{\partial {z_5}}})$ and $(\frac{{\partial {f_2}}}{{\partial {z_1}}},\cdots,\frac{{\partial {f_2}}}{{\partial {z_5}}})$ are linear dependent. Thus $V(z_2, z_3, z_4, z_5)$, which has dimension one, is contained in the singular locus of $X$. This contradicts with $X$ has an isolated singularity.
\\
\\(2)~We may assume that $l=1$ and $i,j=1,2$. Assume  on the contrary  that neither of $(2a)$ and $(2b)$ occurs, then $z_1^a z_2^b$ does not appear in $f_1$, for any non-negative integers $a,b$ and $z_k z_1^a z_2^b$ does not appear in $f_1$. And for any $k \in \{3,4,5\}$ and for any non-negative integers $a,b$, we have $f_1=0$ and $(\frac{{\partial {f_1}}}{{\partial {z_1}}},\cdots,\frac{{\partial {f_1}}}{{\partial {z_5}}})=0$ when $z_3=z_4=z_5=0$. Thus $V(f_2, z_3,z_4,z_5)$, which has dimension at least one, is contained in the singular locus of $X$. This contradicts with $X$ has an isolated singularity.
\\
\\(3)~We may assume that $i,j=1, 2$. Assume  on the contrary  that neither of $(3a), (3b)$ and $(3c)$ occurs, then one of the following two cases occurs:
\\\textbf{subcase (a)}~$z_1^az_2^b$ does not appear in $f_q$ for any $q=1,2$ and any non-negative integer $a,b$, and there exist $l \in \{1,2\}, s,p \in\{3,4,5\} (s\ne p)$ such that both $z_sz_1^{a_1}z_2^{b_1}$ and $z_pz_1^{a_2}z_2^{b_2}$ do not appear in $f_l$ for any non-negative integer  $a_1,a_2,b_1,b_2$.
\\\textbf{subcase (b)}~$z_1^az_2^b$ does not appear in $f_q$ for any $q=1,2$ and any non-negative integer $a,b$, and there exists $k \in \{3,4,5\}$ such that $z_kz_1^a z_2^b$  does not appear in $f_q$ for any $q=1,2$ and any non-negative integer $ a,b$. 
\\If \textbf{subcase (a)} occurs, without lose of generality, we may assume that $l=1$ and $s,p=3,4$, then $f_1=f_2=\frac{{\partial {f_1}}}{{\partial {z_1}}}=\cdots=\frac{{\partial {f_1}}}{{\partial {z_4}}}=0$ when $z_3=z_4=z_5=0$. Thus $V(z_3,z_4,z_5,\frac{{\partial {f_1}}}{{\partial {z_5}}})$, which has dimension at least one, is contained in the singular locus of $X$. This contradicts with $X$ has an isolated singularity.
\\If \textbf{subcase (b)} occurs, without lose of generality, we may assume $k=3$. Then when $z_3=z_4=z_5=0$, we  have  $f_1=f_2=0$,  $(\frac{{\partial {f_1}}}{{\partial {z_1}}},\cdots,\frac{{\partial {f_1}}}{{\partial {z_5}}})=(0,0,0,\frac{{\partial {f_1}}}{{\partial {z_4}}},\frac{{\partial {f_1}}}{{\partial {z_5}}})$ and $(\frac{{\partial {f_2}}}{{\partial {z_1}}},\cdots,\frac{{\partial {f_2}}}{{\partial {z_5}}})=(0,0,0,\frac{{\partial {f_2}}}{{\partial {z_4}}},\frac{{\partial {f_2}}}{{\partial {z_5}}})$. Thus $V(z_3,z_4,z_5,\frac{{\partial {f_1}}}{{\partial {z_4}}}\frac{{\partial {f_2}}}{{\partial {z_5}}} - \frac{{\partial {f_2}}}{{\partial {z_4}}}\frac{{\partial {f_1}}}{{\partial {z_5}}})$, which has dimension at least one, is contained in the singular locus of $X$. This contradicts with $X$ has an isolated singularity.
\\
\\(4)~We may assume that $i,j,k=1,2,3$ and $l=1$. Assume  on the contrary that  neither of $(4a)$ and  $(4b)$  occurs, thus $z_1^a z_2^b z_3^c$ does not appear in $f_1$ for any non-negative integers $a,b,c$ and there exists $p \in\{4,5\}$ such that $z_p z_1^a z_2^b z_3^c$ does not appear in $f_1$ for any non-negative integers $a,b,c$. Without lose of generality, we may assume that $p=4$. Then  $f_1=\frac{{\partial {f_1}}}{{\partial {z_1}}}=\cdots=\frac{{\partial {f_1}}}{{\partial {z_4}}}=0$ when $z_4=z_5=0$. Thus $V(z_4,z_5,f_2,\frac{{\partial {f_1}}}{{\partial {z_5}}})$, which has dimension at least one, is contained in the singular locus of $X$. This contradicts with $X$ has an isolated singularity.
\\
\\(5)~We may assume that $i,j,k=1,2,3$.  Assume  on the contrary that  $z_1^a z_2^b z_3^c$ does not appear in $f_l$ for any $l=1,2$ and any non-negative integers $a,b,c$. Then when $z_4=z_5=0$ we have  $f_1=f_2=0$ and $(\frac{{\partial {f_1}}}{{\partial {z_1}}},\cdots,\frac{{\partial {f_1}}}{{\partial {z_5}}})=(0,0,0,\frac{{\partial {f_1}}}{{\partial {z_4}}},\frac{{\partial {f_1}}}{{\partial {z_5}}})$ and $(\frac{{\partial {f_2}}}{{\partial {z_1}}},\cdots,\frac{{\partial {f_2}}}{{\partial {z_5}}})=(0,0,0,\frac{{\partial {f_2}}}{{\partial {z_4}}},\frac{{\partial {f_2}}}{{\partial {z_5}}})$. Thus $V(z_4,z_5,\frac{{\partial {f_1}}}{{\partial {z_4}}}\frac{{\partial {f_2}}}{{\partial {z_5}}} - \frac{{\partial {f_2}}}{{\partial {z_4}}}\frac{{\partial {f_1}}}{{\partial {z_5}}})$, which has dimension at least two, is contained in the singular locus of $X$. This contradicts with $X$ has an isolated singularity.
\end{proof}

We define $N(a)=\{ka \mid k$ is non-negative integer$\}$, $N(a,b)=\{ka+sb \mid k,s$ are~non-negative~integers$\}$ and $N(a,b,c)=\{ka+sb+tc \mid k,s,t$ are~non-negative integers$\}$.
\begin{cor}\label{c2}
Let $X=V(f_1,f_2)$ be a weighted homogeneous ICIS of type
\\ $(w_1,\cdots,w_5;d_1,d_2)$, then we have:
\\(1)~for any $i \in \{1,2,3,4,5\}$, we have $d_1 \in N(w_i)$ or $d_2 \in N(w_i)$ or there exist $j ,k\in  \{1,2,3,4,5\}\setminus\{i\},\, j \ne k$ such that $d_1-w_j \in N(w_i)$ and $d_2-w_k \in N(w_i)$.
\\(2)~for any $\{i,j\} \subset \{1,2,3,4,5\}$ and any $l \in \{1,2\}$, we have $d_l \in N(w_i,w_j)$ or there exists $k\in \{1,2,3,4,5\}\setminus\{i,j\}$ such that $d_l-w_k\in N(w_i,w_j)$.
\\(3)~for any $\{i,j\} \subset \{1,2,3,4,5\}$, if $d_1,d_2 \notin N(w_i,w_j)$, then  there exist $\{p_1,p_2\},\{s_1,s_2\} \subset \{1,2,3,4,5\}\setminus \{i,j\}$ and $\{p_1,p_2\} \ne \{s_1,s_2\}$ such that $d_1-w_{p_1}$, $d_1-w_{p_2} \in N(w_i,w_j)$  and $d_2-w_{s_1}$, $d_2-w_{s_2}\in N(w_i,w_j)$.
\\(4)~for any $\{i,j,k\} \subset \{1,2,3,4,5\}$  and any $l \in \{1,2\}$, let $\{p,s\}=\{1,2,3,4,5\}\setminus \{i,j,k\}$,  if $d_l \notin N(w_i,w_j,w_k)$, then we have  $d_l-w_p, d_l-w_s \in N(w_i,w_j,w_k)$.
\\(5)~for any $\{i,j,k\} \subset \{1,2,3,4,5\}$, we have $d_1 \in N(w_i,w_j,w_k)$ or $d_2 \in N(w_i,w_j,w_k)$.
\end{cor}

\begin{theorem}
Let $X=V(f_1,f_2)$ be a weighted homogeneous ICIS of type $(w_1,\cdots,w_5;1,d)$,\\with $d \ge 1$. Then $(f_1,f_2)$ has the same type as one of the following weight homogeneous singularities in the list up to permutation of coordinates.  
\end{theorem}

\begin{rem}
We also list the Milnor number  $\mu$ and the vector basis of the miniversal deformation of the singularities in the list. In order to save space, we only list the set of maximum elements (i.e. \textbf{mini})  of the vector basis of the corresponding singularity. That is, $\{[a,0] \mid a\in \textbf{m},\; \exists \; [b,0] \in \textbf{mini}~s.t.~b \ge  a\} \cup \{[0,a] \mid  a\in \textbf{m},\; \exists \; [0,b] \in \textbf{mini}~s.t.~b \ge  a\}$ form a basis of miniversal deformation where $\textbf{m}=$ ideal $(z_1,\cdots,z_5)$.  A monomial $z_1^{a_1}z_2^{a_2}z_3^{a_3}z_4^{a_4}z_5^{a_5}\le z_1^{b_1}z_2^{b_2}z_3^{b_3}z_4^{b_4}z_5^{b_5} $ if and only if $a_i\leq b_i$ for $i=1,\cdots,5$. For example, if $\textbf{mini}=\{[z_1^3,0],[z_1z_2^2,0],[0,z_3z_4]\}$, then $$\{[0,1],[1,0],[z_1,0], [z_1^2,0], [z_1^3,0], [z_2,0], [z_2^2,0], [z_1z_2,0], [z_1z_2^2,0], [0,z_3], [0,z_4], [0,z_3z_4]\}$$ form a vector basis of miniversal deformation.
\end{rem}
\leavevmode \\
\begin{flalign*}
\begin{split}
(1)\left\{ \begin{array}{l} z_1^{2}+z_2^{2}+z_3^{2}+z_4^{2}+z_5^{n}=0\\z_1^{2}+2z_2^{2}+3z_3^{2}+4z_4^{2}+5z_5^{n}=0\end{array}\right.
\end{split}&
\end{flalign*}
\\$n\ge2$
\\$(w_1,w_2,w_3,w_4,w_5;1,d)=(\frac{1}{2},\frac{1}{2},\frac{1}{2},\frac{1}{2},\frac{1}{n};1,1)$
\\$\mu=-7 + 8 n$
\begin{flalign*}
\begin{split}
mini=\{[z_5^{-2 + 2 n},0],[z_4^{2}z_5^{-2 + n},0],[z_3z_5^{-2 + n},0],[z_2z_5^{-2 + n},0],[z_1z_5^{-2 + n},0],[0,z_5^{-2 + n}]\}
\end{split}&
\end{flalign*}
\\(2)$\left\{ \begin{array}{l} z_1^{2}+z_2^{2}+z_3^{2}+z_4^{3}+z_5^{3}=0\\z_1^{2}+2z_2^{2}+3z_3^{2}+4z_4^{3}+5z_5^{3}=0\end{array}\right.$
\\$(w_1,w_2,w_3,w_4,w_5;1,d)=(\frac{1}{2},\frac{1}{2},\frac{1}{2},\frac{1}{3},\frac{1}{3};1,1)$
\\$\mu=32$
\begin{flalign*}
\begin{split}
mini=\{[z_4z_5^{4},0],[z_4^{4}z_5,0],[z_3z_4z_5,0],[z_2z_4z_5,0],[z_1z_4z_5,0],[0,z_4z_5]\}
\end{split}&
\end{flalign*}
\\(3)$\left\{ \begin{array}{l} z_1^{2}+z_2^{2}+z_3^{2}+z_4^{3}+z_5^{4}=0\\z_1^{2}+2z_2^{2}+3z_3^{2}+4z_4^{3}+5z_5^{4}=0\end{array}\right.$
\\$(w_1,w_2,w_3,w_4,w_5;1,d)=(\frac{1}{2},\frac{1}{2},\frac{1}{2},\frac{1}{3},\frac{1}{4};1,1)$
\\$\mu=47$
\begin{flalign*}
\begin{split}
mini=\{[z_4z_5^{6},0],[z_4^{4}z_5^{2},0],[z_3z_4z_5^{2},0],[z_2z_4z_5^{2},0],[z_1z_4z_5^{2},0],[0,z_4z_5^{2}]\}
\end{split}&
\end{flalign*}
\\(4)$\left\{ \begin{array}{l} z_1^{2}+z_2^{2}+z_3^{2}+z_4^{3}+z_5^{5}=0\\z_1^{2}+2z_2^{2}+3z_3^{2}+4z_4^{3}+5z_5^{5}=0\end{array}\right.$
\\$(w_1,w_2,w_3,w_4,w_5;1,d)=(\frac{1}{2},\frac{1}{2},\frac{1}{2},\frac{1}{3},\frac{1}{5};1,1)$
\\$\mu=62$
\begin{flalign*}
\begin{split}
mini=\{[z_4z_5^{8},0],[z_4^{4}z_5^{3},0],[z_3z_4z_5^{3},0],[z_2z_4z_5^{3},0],[z_1z_4z_5^{3},0],[0,z_4z_5^{3}]\}
\end{split}&
\end{flalign*}
\\(5)$\left\{ \begin{array}{l} z_1z_3+z_2^{2}+z_4^{2}=0\\z_1^{2}+z_3^{4}+z_5^{3}+z_2^{2}z_3=0\end{array}\right.$
\\$(w_1,w_2,w_3,w_4,w_5;1,d)=(\frac{2}{3},\frac{1}{2},\frac{1}{3},\frac{1}{2},\frac{4}{9};1,\frac{4}{3})$
\\$\mu=45$
\begin{flalign*}
\begin{split}
mini=\{[z_5^{2},0],[z_3^{2}z_5,0],[0,z_3^{2}z_4^{2}z_5],[0,z_3^{3}z_4z_5],[0,z_2z_4z_5],[0,z_3^{6}z_5],[0,z_2z_3^{3}z_5]\}
\end{split}&
\end{flalign*}
\\(6)$\left\{ \begin{array}{l} z_1z_3+z_2^{2}+z_5^{3}=0\\z_1^{2}+z_4^{3}+z_3^{2}z_5+z_2z_5^{2}=0\end{array}\right.$
\\$(w_1,w_2,w_3,w_4,w_5;1,d)=(\frac{7}{12},\frac{1}{2},\frac{5}{12},\frac{7}{18},\frac{1}{3};1,\frac{7}{6})$
\\$\mu=50$
\begin{flalign*}
\begin{split}
mini=\{[z_4z_5^{5},0],[z_3z_4z_5^{2},0],[z_4^{2}z_5,0],[z_3^{2}z_4z_5,0],[0,z_4z_5^{5}],[0,z_3z_4z_5^{3}],[0,z_2z_4z_5],[0,z_2z_3z_4]\}
\end{split}&
\end{flalign*}
\\(7)$\left\{ \begin{array}{l} z_1z_3+z_4z_5=0\\z_1^{2}+z_2^{2}+z_3^{n}+z_4^{2}+2z_5^{n}=0\end{array}\right.$
\\$n\ge3$
\\$(w_1,w_2,w_3,w_4,w_5;1,d)=(\frac{n}{2 + n},\frac{n}{2 + n},\frac{2}{2 + n},\frac{n}{2 + n},\frac{2}{2 + n};1,\frac{2 n}{2 + n})$
\\$\mu=(1 + n)^2$
\begin{flalign*}
\begin{split}
mini=\{[z_5^{n},0],[z_3z_5^{-2 + n},0],[z_2,0],[0,z_5^{-2 + 2 n}],[0,z_3z_5^{-4 + 2 n}],[0,z_3^{-2 + n}z_5^{-2 + n}],[0,z_3^{-1 + n}]\}
\end{split}&
\end{flalign*}
\\(8)$\left\{ \begin{array}{l} z_1z_3+z_4z_5=0\\z_1^{2}+z_2^{2}+z_3z_4^{2}+z_3^{3 n}+2z_5^{2 n}=0\end{array}\right.$
\\$n\ge2$
\\$(w_1,w_2,w_3,w_4,w_5;1,d)=(\frac{3 n}{2 + 3 n},\frac{3 n}{2 + 3 n},\frac{2}{2 + 3 n},\frac{-1 + 3 n}{2 + 3 n},\frac{3}{2 + 3 n};1,\frac{6 n}{2 + 3 n})$
\\$\mu=3 + 7 n + 6 n^2$
\begin{flalign*}
\begin{split}
mini=\{[z_3,0],[z_2,0],[0,z_3z_5^{-1 + 2 n}],[0,z_3^{-2 + 3 n}z_5^{-2 + 2 n}],[0,z_3^{-2 + 3 n}z_4z_5],[0,z_3^{-5 + 6 n}z_5],[0,z_4^{2}],\\ [0,z_3^{-1 + 3 n}z_4],[0,z_3^{-2 + 6 n}]\}
\end{split}&
\end{flalign*}
\\(9)$\left\{ \begin{array}{l} z_1z_3+z_4z_5=0\\z_1^{2}+z_2^{2}+z_3z_4^{2}+z_3^{1 + 3 n}+2z_3z_5^{2 n}=0\end{array}\right.$
\\$n\ge1$
\\$(w_1,w_2,w_3,w_4,w_5;1,d)=(\frac{1 + 3 n}{3 (1 + n)},\frac{1 + 3 n}{3 (1 + n)},\frac{2}{3 (1 + n)},\frac{n}{1 + n},\frac{1}{1 + n};1,\frac{2 (1 + 3 n)}{3 (1 + n)})$
\\$\mu=6 + 11 n + 6 n^2$
\begin{flalign*}
\begin{split}
mini=\{[z_3,0],[z_2,0],[0,z_5^{2 n}],[0,z_3^{3}z_5^{-1 + 2 n}],[0,z_3^{3 n}z_5^{-2 + 2 n}],[0,z_3^{-1 + 3 n}z_4z_5],[0,z_3^{-3 + 6 n}z_5],[0,z_4^{2}],\\ [0,z_3^{3 n}z_4],[0,z_3^{6 n}]\}
\end{split}&
\end{flalign*}
\\(10)$\left\{ \begin{array}{l} z_1z_3+z_4z_5=0\\z_1^{2}+z_2^{2}+z_3z_4^{3}+z_3^{16}+z_5^{4}=0\end{array}\right.$
\\$(w_1,w_2,w_3,w_4,w_5;1,d)=(\frac{8}{9},\frac{8}{9},\frac{1}{9},\frac{5}{9},\frac{4}{9};1,\frac{16}{9})$
\\$\mu=127$
\begin{flalign*}
\begin{split}
mini=\{[z_3,0],[z_2,0],[0,z_3^{6}z_5^{3}],[0,z_3^{14}z_5^{2}],[0,z_3^{14}z_4z_5],[0,z_3^{22}z_5],[0,z_4^{3}],[0,z_3^{15}z_4^{2}],[0,z_3^{30}]\}
\end{split}&
\end{flalign*}
\\(11)$\left\{ \begin{array}{l} z_1z_3+z_4z_5=0\\z_1^{2}+z_2^{2}+z_3z_4^{3}+z_3^{40}+z_5^{5}=0\end{array}\right.$
\\$(w_1,w_2,w_3,w_4,w_5;1,d)=(\frac{20}{21},\frac{20}{21},\frac{1}{21},\frac{13}{21},\frac{8}{21};1,\frac{40}{21})$
\\$\mu=358$
\begin{flalign*}
\begin{split}
mini=\{[z_3,0],[z_2,0],[0,z_3^{38}z_5^{4}],[0,z_3^{38}z_4z_5],[0,z_4^{3}],[0,z_3^{39}z_4^{2}],[0,z_3^{78}]\}
\end{split}&
\end{flalign*}
\\(12)$\left\{ \begin{array}{l} z_1z_3+z_4z_5=0\\z_1^{2}+z_2^{2}+z_3z_4^{3}+z_3^{10}+z_3z_5^{3}=0\end{array}\right.$
\\$(w_1,w_2,w_3,w_4,w_5;1,d)=(\frac{5}{6},\frac{5}{6},\frac{1}{6},\frac{1}{2},\frac{1}{2};1,\frac{5}{3})$
\\$\mu=73$
\begin{flalign*}
\begin{split}
mini=\{[z_3,0],[z_2,0],[0,z_5^{3}],[0,z_3^{6}z_5^{2}],[0,z_3^{8}z_4z_5],[0,z_3^{12}z_5],[0,z_4^{3}],[0,z_3^{6}z_4^{2}],[0,z_3^{12}z_4],[0,z_3^{18}]\}
\end{split}&
\end{flalign*}
\\(13)$\left\{ \begin{array}{l} z_1z_3+z_4z_5=0\\z_1^{2}+z_2^{2}+z_3z_4^{3}+z_3^{19}+z_3z_5^{4}=0\end{array}\right.$
\\$(w_1,w_2,w_3,w_4,w_5;1,d)=(\frac{19}{21},\frac{19}{21},\frac{2}{21},\frac{4}{7},\frac{3}{7};1,\frac{38}{21})$
\\$\mu=155$
\begin{flalign*}
\begin{split}
mini=\{[z_3,0],[z_2,0],[0,z_5^{4}],[0,z_3^{9}z_5^{3}],[0,z_3^{18}z_5^{2}],[0,z_3^{17}z_4z_5],[0,z_3^{27}z_5],[0,z_4^{3}],[0,z_3^{18}z_4^{2}],[0,z_3^{36}]\}
\end{split}&
\end{flalign*}
\\(14)$\left\{ \begin{array}{l} z_1z_3+z_4z_5=0\\z_1^{2}+z_2^{2}+z_3z_4^{3}+z_3^{46}+z_3z_5^{5}=0\end{array}\right.$
\\$(w_1,w_2,w_3,w_4,w_5;1,d)=(\frac{23}{24},\frac{23}{24},\frac{1}{24},\frac{5}{8},\frac{3}{8};1,\frac{23}{12})$
\\$\mu=417$
\begin{flalign*}
\begin{split}
mini=\{[z_3,0],[z_2,0],[0,z_5^{5}],[0,z_3^{45}z_5^{4}],[0,z_3^{44}z_4z_5],[0,z_4^{3}],[0,z_3^{45}z_4^{2}],[0,z_3^{90}]\}
\end{split}&
\end{flalign*}
\\(15)$\left\{ \begin{array}{l} z_1z_3+z_4z_5=0\\z_1^{2}+z_2^{2}+z_4^{3}+z_3^{6}+z_5^{3}=0\end{array}\right.$
\\$(w_1,w_2,w_3,w_4,w_5;1,d)=(\frac{3}{4},\frac{3}{4},\frac{1}{4},\frac{1}{2},\frac{1}{2};1,\frac{3}{2})$
\\$\mu=39$
\begin{flalign*}
\begin{split}
mini=\{[z_3,0],[z_2,0],[0,z_3^{2}z_5^{2}],[0,z_3^{4}z_4z_5],[0,z_3^{6}z_5],[0,z_3^{2}z_4^{2}],[0,z_3^{6}z_4],[0,z_3^{10}]\}
\end{split}&
\end{flalign*}
\\(16)$\left\{ \begin{array}{l} z_1z_3+z_4z_5=0\\z_1^{2}+z_2^{2}+z_4^{3}+z_3^{12}+z_5^{4}=0\end{array}\right.$
\\$(w_1,w_2,w_3,w_4,w_5;1,d)=(\frac{6}{7},\frac{6}{7},\frac{1}{7},\frac{4}{7},\frac{3}{7};1,\frac{12}{7})$
\\$\mu=92$
\begin{flalign*}
\begin{split}
mini=\{[z_3,0],[z_2,0],[0,z_3^{4}z_5^{3}],[0,z_3^{10}z_5^{2}],[0,z_3^{10}z_4z_5],[0,z_3^{16}z_5],[0,z_3^{10}z_4^{2}],[0,z_3^{22}]\}
\end{split}&
\end{flalign*}
\\(17)$\left\{ \begin{array}{l} z_1z_3+z_4z_5=0\\z_1^{2}+z_2^{2}+z_4^{3}+z_3^{30}+z_5^{5}=0\end{array}\right.$
\\$(w_1,w_2,w_3,w_4,w_5;1,d)=(\frac{15}{16},\frac{15}{16},\frac{1}{16},\frac{5}{8},\frac{3}{8};1,\frac{15}{8})$
\\$\mu=265$
\begin{flalign*}
\begin{split}
mini=\{[z_3,0],[z_2,0],[0,z_3^{28}z_5^{4}],[0,z_3^{28}z_4z_5],[0,z_3^{28}z_4^{2}],[0,z_3^{58}]\}
\end{split}&
\end{flalign*}
\\(18)$\left\{ \begin{array}{l} z_1z_3+z_4z_5=0\\z_1^{2}+z_2^{2}+z_4^{3}+z_3^{8}+z_3z_5^{3}=0\end{array}\right.$
\\$(w_1,w_2,w_3,w_4,w_5;1,d)=(\frac{4}{5},\frac{4}{5},\frac{1}{5},\frac{8}{15},\frac{7}{15};1,\frac{8}{5})$
\\$\mu=56$
\begin{flalign*}
\begin{split}
mini=\{[z_3,0],[z_2,0],[0,z_5^{3}],[0,z_3^{4}z_5^{2}],[0,z_3^{6}z_4z_5],[0,z_3^{9}z_5],[0,z_3^{4}z_4^{2}],[0,z_3^{9}z_4],[0,z_3^{14}]\}
\end{split}&
\end{flalign*}
\\(19)$\left\{ \begin{array}{l} z_1z_3+z_4z_5=0\\z_1^{2}+z_2^{2}+z_4^{3}+z_3^{15}+z_3z_5^{4}=0\end{array}\right.$
\\$(w_1,w_2,w_3,w_4,w_5;1,d)=(\frac{15}{17},\frac{15}{17},\frac{2}{17},\frac{10}{17},\frac{7}{17};1,\frac{30}{17})$
\\$\mu=120$
\begin{flalign*}
\begin{split}
mini=\{[z_3,0],[z_2,0],[0,z_5^{4}],[0,z_3^{7}z_5^{3}],[0,z_3^{14}z_5^{2}],[0,z_3^{13}z_4z_5],[0,z_3^{21}z_5],[0,z_3^{13}z_4^{2}],[0,z_3^{28}]\}
\end{split}&
\end{flalign*}
\\(20)$\left\{ \begin{array}{l} z_1z_3+z_4z_5=0\\z_1^{2}+z_2^{2}+z_4^{3}+z_3^{36}+z_3z_5^{5}=0\end{array}\right.$
\\$(w_1,w_2,w_3,w_4,w_5;1,d)=(\frac{18}{19},\frac{18}{19},\frac{1}{19},\frac{12}{19},\frac{7}{19};1,\frac{36}{19})$
\\$\mu=324$
\begin{flalign*}
\begin{split}
mini=\{[z_3,0],[z_2,0],[0,z_5^{5}],[0,z_3^{35}z_5^{4}],[0,z_3^{34}z_4z_5],[0,z_3^{34}z_4^{2}],[0,z_3^{70}]\}
\end{split}&
\end{flalign*}
\\(21)$\left\{ \begin{array}{l} z_1z_3+z_4z_5=0\\z_1^{2}+z_2^{2}+z_3z_4^{2}+z_3^{3 n}z_5+2z_5^{1 + 2 n}=0\end{array}\right.$
\\$n\ge 1$
\\$(w_1,w_2,w_3,w_4,w_5;1,d)=(\frac{3 (1 + 2 n)}{7 + 6 n},\frac{3 (1 + 2 n)}{7 + 6 n},\frac{4}{7 + 6 n},\frac{1 + 6 n}{7 + 6 n},\frac{6}{7 + 6 n};1,\frac{6 (1 + 2 n)}{7 + 6 n})$
\\$\mu=8 + 13 n + 6 n^2$
\begin{flalign*}
\begin{split}
mini=\{[z_3,0],[z_2,0],[0,z_3z_5^{2 n}],[0,z_3^{-2 + 3 n}z_5^{-1 + 2 n}],[0,z_3^{-5 + 6 n}z_5^{2}],[0,z_3^{-1 + 3 n}z_4z_5],[0,z_3^{-2 + 6 n}z_5],\\ [0,z_4^{2}],[0,z_3^{1 + 3 n}z_4],[0,z_3^{1 + 6 n}]\}
\end{split}&
\end{flalign*}
\\(22)$\left\{ \begin{array}{l} z_1z_3+z_4z_5=0\\z_1^{2}+z_2^{2}+z_3z_4^{2}+z_3^{1 + 3 n}z_5+2z_3z_5^{1 + 2 n}=0\end{array}\right.$
\\$n\ge1$
\\$(w_1,w_2,w_3,w_4,w_5;1,d)=(\frac{5 + 6 n}{3 (3 + 2 n)},\frac{5 + 6 n}{3 (3 + 2 n)},\frac{4}{3 (3 + 2 n)},\frac{1 + 2 n}{3 + 2 n},\frac{2}{3 + 2 n};1,\frac{2 (5 + 6 n)}{3 (3 + 2 n)})$
\\$\mu=13 + 17 n + 6 n^2$
\begin{flalign*}
\begin{split}
mini=\{[z_3,0],[z_2,0],[0,z_5^{1 + 2 n}],[0,z_3^{3}z_5^{2 n}],[0,z_3^{3 n}z_5^{-1 + 2 n}],[0,z_3^{-3 + 6 n}z_5^{2}],[0,z_3^{3 n}z_4z_5],[0,z_3^{6 n}z_5],\\ [0,z_4^{2}],[0,z_3^{2 + 3 n}z_4],[0,z_3^{3 + 6 n}]\}
\end{split}&
\end{flalign*}
\\(23)$\left\{ \begin{array}{l} z_1z_3+z_4^{2}+z_2z_5^{n}=0\\z_1^{2}+z_2^{2}+z_3^{3}+z_5^{3 n}=0\end{array}\right.$
\\$1\le n \le3$
\\$(w_1,w_2,w_3,w_4,w_5;1,d)=(\frac{3}{5},\frac{3}{5},\frac{2}{5},\frac{1}{2},\frac{2}{5 n};1,\frac{6}{5})$
\\$\mu=-10 + 32 n$
\begin{flalign*}
\begin{split}
mini=\{[z_5^{-2 + 6 n},0],[z_3z_5^{-2 + 4 n},0],[z_2z_5^{-2 + n},0],[0,z_5^{-2 + 6 n}],[0,z_3z_5^{-2 + 4 n}],[0,z_3z_4z_5^{-2 + 3 n}],\\ [0,z_4^{2}z_5^{-2 + 2 n}],[0,z_3^{2}z_5^{-2 + 2 n}],[0,z_3z_4^{2}z_5^{-2 + n}]\}
\end{split}&
\end{flalign*}
\\(24)$\left\{ \begin{array}{l} z_1z_3+z_4^{2}+z_2z_5=0\\z_1^{2}+z_2^{2}+z_3^{4}+z_5^{4}=0\end{array}\right.$
\\$(w_1,w_2,w_3,w_4,w_5;1,d)=(\frac{2}{3},\frac{2}{3},\frac{1}{3},\frac{1}{2},\frac{1}{3};1,\frac{4}{3})$
\\$\mu=41$
\begin{flalign*}
\begin{split}
mini=\{[z_5^{6},0],[z_3z_5^{4},0],[0,z_5^{6}],[0,z_3z_5^{4}],[0,z_4^{2}z_5^{2}],[0,z_3^{2}z_4z_5^{2}],[0,z_3z_4^{2}],[0,z_3^{3}]\}
\end{split}&
\end{flalign*}
\\(25)$\left\{ \begin{array}{l} z_1z_3+z_4^{2}+z_2z_5=0\\z_1^{2}+z_2^{2}+z_3^{5}+z_5^{5}=0\end{array}\right.$
\\$(w_1,w_2,w_3,w_4,w_5;1,d)=(\frac{5}{7},\frac{5}{7},\frac{2}{7},\frac{1}{2},\frac{2}{7};1,\frac{10}{7})$
\\$\mu=66$
\begin{flalign*}
\begin{split}
mini=\{[z_5^{8},0],[z_3z_5^{6},0],[0,z_5^{8}],[0,z_3z_5^{6}],[0,z_3z_4^{2}z_5^{3}],[0,z_3^{3}z_4z_5^{3}],[0,z_3^{2}z_4^{2}],[0,z_3^{4}]\}
\end{split}&
\end{flalign*}
\\(26)$\left\{ \begin{array}{l} z_1z_3+z_4^{2}+z_2z_5^{n}=0\\2z_1^{2}+z_2^{2}+z_3^{2}z_5+z_5^{1 + 2 n}=0\end{array}\right.$
\\$n\ge2$
\\$(w_1,w_2,w_3,w_4,w_5;1,d)=(\frac{1 + 2 n}{1 + 4 n},\frac{1 + 2 n}{1 + 4 n},\frac{2 n}{1 + 4 n},\frac{1}{2},\frac{2}{1 + 4 n};1,\frac{2 (1 + 2 n)}{1 + 4 n})$
\\$\mu=6 + 16 n$
\begin{flalign*}
\begin{split}
mini=\{[z_5^{4 n},0],[z_3z_5^{2 n},0],[z_2z_5^{-1 + n},0],[0,z_5^{4 n}],[0,z_4z_5^{2 n}],[0,z_3z_5^{2 n}],[0,z_4^{2}z_5^{-1 + n}],[0,z_3z_4]\}
\end{split}&
\end{flalign*}
\\(27)$\left\{ \begin{array}{l} z_1z_3+z_4^{2}+z_2z_5^{2}=0\\z_1^{2}+z_2^{2}+z_3^{3}z_5+z_5^{7}=0\end{array}\right.$
\\$(w_1,w_2,w_3,w_4,w_5;1,d)=(\frac{7}{11},\frac{7}{11},\frac{4}{11},\frac{1}{2},\frac{2}{11};1,\frac{14}{11})$
\\$\mu=73$
\begin{flalign*}
\begin{split}
mini=\{[z_5^{12},0],[z_3z_5^{8},0],[z_2,0],[0,z_5^{12}],[0,z_3z_5^{8}],[0,z_3z_4z_5^{6}],[0,z_4^{2}z_5^{4}],[0,z_3^{2}z_5^{4}],[0,z_3z_4^{2}z_5],\\ [0,z_3^{2}z_4],[0,z_3^{3}]\}
\end{split}&
\end{flalign*}
\\(28)$\left\{ \begin{array}{l} z_1z_3+z_4^{2}+z_5^{n}=0\\z_1^{2}+z_2^{2}+z_3^{4}+z_3z_4^{2}=0\end{array}\right.$
\\$4\le n \le5$
\\$(w_1,w_2,w_3,w_4,w_5;1,d)=(\frac{2}{3},\frac{2}{3},\frac{1}{3},\frac{1}{2},\frac{1}{n};1,\frac{4}{3})$
\\$\mu=-13 + 18 n$
\begin{flalign*}
\begin{split}
mini=\{[z_3z_5^{-2 + n},0],[z_2z_5^{-2 + n},0],[0,z_5^{-2 + 3 n}],[0,z_4z_5^{-2 + 2 n}],[0,z_3^{2}z_5^{-2 + 2 n}],[0,z_4^{2}z_5^{-2 + n}],\\ [0,z_3^{2}z_4z_5^{-2 + n}],[0,z_3^{5}z_5^{-2 + n}]\}
\end{split}&
\end{flalign*}
\\(29)$\left\{ \begin{array}{l} z_1z_3+z_4^{2}+z_5^{3}=0\\z_1^{2}+z_2^{2}+z_3^{6}+z_4^{3}=0\end{array}\right.$
\\$(w_1,w_2,w_3,w_4,w_5;1,d)=(\frac{3}{4},\frac{3}{4},\frac{1}{4},\frac{1}{2},\frac{1}{3};1,\frac{3}{2})$
\\$\mu=68$
\begin{flalign*}
\begin{split}
mini=\{[z_3z_5,0],[z_2z_5,0],[0,z_3^{2}z_4z_5^{4}],[0,z_3^{4}z_5^{4}],[0,z_3^{4}z_4z_5^{2}],[0,z_3z_4^{2}z_5],[0,z_3^{10}z_5]\}
\end{split}&
\end{flalign*}
\\(30)$\left\{ \begin{array}{l} z_1z_3+z_4^{2}+z_5^{3}=0\\z_1^{2}+z_2^{2}+z_3^{7}+z_3z_5^{4}=0\end{array}\right.$
\\$(w_1,w_2,w_3,w_4,w_5;1,d)=(\frac{7}{9},\frac{7}{9},\frac{2}{9},\frac{1}{2},\frac{1}{3};1,\frac{14}{9})$
\\$\mu=82$
\begin{flalign*}
\begin{split}
mini=\{[z_3z_5,0],[z_2z_5,0],[0,z_5^{5}],[0,z_4z_5^{3}],[0,z_3^{6}z_5^{3}],[0,z_3^{6}z_4z_5^{2}],[0,z_3^{5}z_4^{2}z_5],[0,z_3^{12}z_5]\}
\end{split}&
\end{flalign*}
\\(31)$\left\{ \begin{array}{l} z_1z_3+z_4^{2}+z_5^{4}=0\\z_1^{2}+z_2^{2}+z_3^{2}z_4+z_5^{5}=0\end{array}\right.$
\\$(w_1,w_2,w_3,w_4,w_5;1,d)=(\frac{5}{8},\frac{5}{8},\frac{3}{8},\frac{1}{2},\frac{1}{4};1,\frac{5}{4})$
\\$\mu=47$
\begin{flalign*}
\begin{split}
mini=\{[z_3z_5^{2},0],[z_2z_5^{2},0],[0,z_5^{9}],[0,z_3z_5^{6}],[0,z_4z_5^{5}],[0,z_4^{2}z_5^{2}],[0,z_3z_4z_5^{2}],[0,z_3^{4}z_5^{2}]\}
\end{split}&
\end{flalign*}
\\(32)$\left\{ \begin{array}{l} z_1z_3+z_4^{2}+z_5^{1 + 2 n}=0\\z_1^{2}+z_2^{2}+z_3^{2}z_5+z_4z_5^{1 + n}=0\end{array}\right.$
\\$n\ge1$
\\$(w_1,w_2,w_3,w_4,w_5;1,d)=(\frac{3 + 4 n}{4 (1 + 2 n)},\frac{3 + 4 n}{4 (1 + 2 n)},\frac{1 + 4 n}{4 (1 + 2 n)},\frac{1}{2},\frac{1}{1 + 2 n};1,\frac{3 + 4 n}{2 (1 + 2 n)})$
\\$\mu=2 (5 + 8 n)$
\begin{flalign*}
\begin{split}
mini=\{[z_5^{1 + 4 n},0],[z_3z_5^{2 n},0],[z_2z_5^{-1 + 2 n},0],[0,z_5^{1 + 4 n}],[0,z_3z_5^{1 + 2 n}],[0,z_4^{2}z_5^{n}],[0,z_3z_4]\}
\end{split}&
\end{flalign*}
\\(33)$\left\{ \begin{array}{l} z_1z_3+z_4^{2}+z_5^{3}=0\\z_1^{2}+z_2^{2}+z_3^{4}z_5+z_3z_4z_5^{2}=0\end{array}\right.$
\\$(w_1,w_2,w_3,w_4,w_5;1,d)=(\frac{13}{18},\frac{13}{18},\frac{5}{18},\frac{1}{2},\frac{1}{3};1,\frac{13}{9})$
\\$\mu=57$
\begin{flalign*}
\begin{split}
mini=\{[z_3z_5,0],[z_2z_5,0],[0,z_5^{5}],[0,z_4z_5^{4}],[0,z_3^{3}z_5^{4}],[0,z_3^{6}z_5^{3}],[0,z_4^{2}z_5],[0,z_3^{2}z_4z_5],[0,z_3^{2}z_4^{2}],\\ [0,z_3^{5}z_4],[0,z_3^{8}]\}
\end{split}&
\end{flalign*}
\\(34)$\left\{ \begin{array}{l} z_1z_3+z_4^{2}z_5+z_5^{4}=0\\z_1^{2}+z_2^{2}+z_3^{2}z_5+z_4^{3}=0\end{array}\right.$
\\$(w_1,w_2,w_3,w_4,w_5;1,d)=(\frac{9}{16},\frac{9}{16},\frac{7}{16},\frac{3}{8},\frac{1}{4};1,\frac{9}{8})$
\\$\mu=53$
\begin{flalign*}
\begin{split}
mini=\{[z_5^{8},0],[z_4z_5^{5},0],[z_3z_5^{4},0],[z_2z_5^{3},0],[z_3z_4z_5,0],[z_4^{3},0],[z_2z_4,0],[0,z_5^{8}],[0,z_4z_5^{5}],[0,z_3z_5^{5}],\\ [0,z_3z_4z_5^{2}]\}
\end{split}&
\end{flalign*}
\\(35)$\left\{ \begin{array}{l} z_1z_3+z_2z_5=0\\z_1^{2}+z_2^{2}+z_4^{3}+z_3^{n}+z_5^{n}=0\end{array}\right.$
\\$3\le n \le5$
\\$(w_1,w_2,w_3,w_4,w_5;1,d)=(\frac{n}{2 + n},\frac{n}{2 + n},\frac{2}{2 + n},\frac{2 n}{3 (2 + n)},\frac{2}{2 + n};1,\frac{2 n}{2 + n})$
\\$\mu=1 + 4 n + 2 n^2$
\begin{flalign*}
\begin{split}
mini=\{[z_4z_5^{n},0],[z_3z_4z_5^{-2 + n},0],[z_4^{2},0],[0,z_4z_5^{-2 + 2 n}],[0,z_3z_4z_5^{-4 + 2 n}],[0,z_3^{-2 + n}z_4z_5^{-2 + n}],\\ [0,z_3^{-1 + n}z_4]\}
\end{split}&
\end{flalign*}
\\(36)$\left\{ \begin{array}{l} z_1z_3+z_2z_4+z_5^{n}=0\\z_1^{2}+z_2^{2}+z_4^{3}+z_3^{3}=0\end{array}\right.$
\\$3\le n \le4$
\\$(w_1,w_2,w_3,w_4,w_5;1,d)=(\frac{3}{5},\frac{3}{5},\frac{2}{5},\frac{2}{5},\frac{1}{n};1,\frac{6}{5})$
\\$\mu=2 (-7 + 9 n)$
\begin{flalign*}
\begin{split}
mini=\{[z_4^{2}z_5^{-2 + n},0],[z_3z_4z_5^{-2 + n},0],[0,z_3z_4z_5^{-2 + 2 n}],[0,z_4^{4}z_5^{-2 + n}],[0,z_3z_4^{2}z_5^{-2 + n}],[0,z_3^{2}z_5^{-2 + n}]\}
\end{split}&
\end{flalign*}
\\(37)$\left\{ \begin{array}{l} z_1z_3+z_2z_4=0\\z_1^{2}+z_2^{2}+z_4^{3}+z_3^{3}+z_5^{n}=0\end{array}\right.$
\\$4\le n \le5$
\\$(w_1,w_2,w_3,w_4,w_5;1,d)=(\frac{3}{5},\frac{3}{5},\frac{2}{5},\frac{2}{5},\frac{6}{5 n};1,\frac{6}{5})$
\\$\mu=-14 + 15 n$
\begin{flalign*}
\begin{split}
mini=\{[z_5^{-2 + 2 n},0],[z_4^{2}z_5^{-2 + n},0],[z_3z_4z_5^{-2 + n},0],[0,z_4^{4}z_5^{-2 + n}],[0,z_3z_4^{2}z_5^{-2 + n}],[0,z_3^{2}z_5^{-2 + n}]\}
\end{split}&
\end{flalign*}
\\(38)$\left\{ \begin{array}{l} z_1z_3+z_2z_4+z_5^{2}=0\\2z_1^{2}+z_2^{2}z_3+z_4^{3}+z_3^{3}z_4=0\end{array}\right.$
\\$(w_1,w_2,w_3,w_4,w_5;1,d)=(\frac{9}{13},\frac{7}{13},\frac{4}{13},\frac{6}{13},\frac{1}{2};1,\frac{18}{13})$
\\$\mu=47$
\begin{flalign*}
\begin{split}
mini=\{[z_3,0],[0,z_5^{3}],[0,z_4z_5^{2}],[0,z_3^{5}z_5^{2}],[0,z_3^{2}z_4z_5],[0,z_2z_5],[0,z_3^{2}z_4^{2}],[0,z_3^{5}z_4],[0,z_2z_3^{2}z_4],[0,z_3^{7}],\\ [0,z_2z_3^{5}],[0,z_2^{2}]\}
\end{split}&
\end{flalign*}
\\(39)$\left\{ \begin{array}{l} z_1z_3+z_2z_4=0\\2z_1^{2}+z_2^{2}z_3+z_4^{3}+z_3^{3}z_4+z_5^{3}=0\end{array}\right.$
\\$(w_1,w_2,w_3,w_4,w_5;1,d)=(\frac{9}{13},\frac{7}{13},\frac{4}{13},\frac{6}{13},\frac{6}{13};1,\frac{18}{13})$
\\$\mu=53$
\begin{flalign*}
\begin{split}
mini=\{[z_5^{2},0],[z_3z_5,0],[0,z_3z_4^{2}z_5],[0,z_3^{4}z_4z_5],[0,z_2z_3^{2}z_4z_5],[0,z_3^{7}z_5],[0,z_2z_3^{4}z_5],[0,z_2^{2}z_5]\}
\end{split}&
\end{flalign*}
\\(40)$\left\{ \begin{array}{l} z_1z_3+z_2z_4=0\\2z_1^{2}+z_2^{2}z_3+z_3z_4^{3}+z_3^{4}z_4+z_5^{3}=0\end{array}\right.$
\\$(w_1,w_2,w_3,w_4,w_5;1,d)=(\frac{11}{15},\frac{3}{5},\frac{4}{15},\frac{2}{5},\frac{22}{45};1,\frac{22}{15})$
\\$\mu=71$
\begin{flalign*}
\begin{split}
mini=\{[z_5^{2},0],[z_3z_5,0],[0,z_4^{3}z_5],[0,z_3^{3}z_4^{2}z_5],[0,z_3^{6}z_4z_5],[0,z_2z_3^{3}z_4z_5],[0,z_3^{9}z_5],[0,z_2z_3^{5}z_5],[0,z_2^{2}z_5]\}
\end{split}&
\end{flalign*}
\\(41)$\left\{ \begin{array}{l} z_2^{2}+z_3^{2}+z_4^{2}+z_5^{2}=0\\z_1^{2}+z_2^{3}+z_3^{3}+z_4^{3}=0\end{array}\right.$
\\$(w_1,w_2,w_3,w_4,w_5;1,d)=(\frac{3}{4},\frac{1}{2},\frac{1}{2},\frac{1}{2},\frac{1}{2};1,\frac{3}{2})$
\\$\mu=31$
\begin{flalign*}
\begin{split}
mini=\{[z_1,0],[0,z_2z_3z_4z_5^{2}],[0,z_4^{4}],[0,z_3z_4^{2}],[0,z_3^{2}]\}
\end{split}&
\end{flalign*}
\\(42)$\left\{ \begin{array}{l} z_2^{2}+z_3^{2}+z_4^{2}+z_5^{3}=0\\z_1^{2}+z_2^{3}+z_3^{3}+z_4^{3}=0\end{array}\right.$
\\$(w_1,w_2,w_3,w_4,w_5;1,d)=(\frac{3}{4},\frac{1}{2},\frac{1}{2},\frac{1}{2},\frac{1}{3};1,\frac{3}{2})$
\\$\mu=54$
\begin{flalign*}
\begin{split}
mini=\{[z_1z_5,0],[0,z_2z_3z_4z_5^{4}],[0,z_4^{4}z_5],[0,z_3z_4^{2}z_5],[0,z_3^{2}z_5]\}
\end{split}&
\end{flalign*}
\\(43)$\left\{ \begin{array}{l} z_2z_3+z_4z_5=0\\z_1^{2}+z_2^{3}+z_3^{n}+z_4^{3}+z_5^{n}=0\end{array}\right.$
\\$4\le n \le5$
\\$(w_1,w_2,w_3,w_4,w_5;1,d)=(\frac{3 n}{2 (3 + n)},\frac{n}{3 + n},\frac{3}{3 + n},\frac{n}{3 + n},\frac{3}{3 + n};1,\frac{3 n}{3 + n})$
\\$\mu=-2 + 5 n + 2 n^2$
\begin{flalign*}
\begin{split}
mini=\{[z_1,0],[0,z_4z_5^{-2 + 2 n}],[0,z_3z_5^{-3 + 2 n}],[0,z_2z_5^{-3 + 2 n}],[0,z_3z_4z_5^{-5 + 2 n}],[0,z_3^{2}z_5^{-5 + 2 n}],[0,z_4^{2}z_5^{-1 + n}],\\ [0,z_3^{-3 + n}z_4z_5^{-2 + n}],[0,z_2z_4z_5^{-2 + n}],[0,z_3^{-2 + n}z_5^{-2 + n}],[0,z_2^{2}z_5^{-2 + n}],[0,z_3^{-2 + n}z_4z_5],[0,z_3^{-1 + n}z_5]\}
\end{split}&
\end{flalign*}
\\(44)$\left\{
\begin{array}{l} z_2z_3+z_4z_5=0\\z_1^{2}+z_2^{3}+z_3^{5}+z_3z_4^{3}+z_3z_5^{3}=0\end{array}\right.$
\\$(w_1,w_2,w_3,w_4,w_5;1,d)=(\frac{15}{16},\frac{5}{8},\frac{3}{8},\frac{1}{2},\frac{1}{2};1,\frac{15}{8})$
\\$\mu=69$
\begin{flalign*}
\begin{split}
mini=\{[z_1,0],[0,z_5^{4}],[0,z_4z_5^{3}],[0,z_2z_5^{3}],[0,z_3^{2}z_4^{2}z_5^{2}],[0,z_3^{3}z_4z_5^{2}],[0,z_3^{4}z_5^{2}],[0,z_3^{3}z_4^{2}z_5],[0,z_3^{7}z_4z_5],\\ [0,z_2z_4z_5],[0,z_2z_4^{3}],[0,z_3^{4}z_4^{2}],[0,z_3^{8}],[0,z_2^{2}]\}
\end{split}&
\end{flalign*}
\\(45)$\left\{ \begin{array}{l} z_2z_3+z_4z_5=0\\z_1^{2}+z_2^{3}+z_3^{4}z_5+z_3z_4^{3}+z_5^{4}=0\end{array}\right.$
\\$(w_1,w_2,w_3,w_4,w_5;1,d)=(\frac{24}{25},\frac{16}{25},\frac{9}{25},\frac{13}{25},\frac{12}{25};1,\frac{48}{25})$
\\$\mu=76$
\begin{flalign*}
\begin{split}
mini=\{[z_1,0],[0,z_3z_4z_5^{3}],[0,z_3^{2}z_5^{3}],[0,z_2z_5^{3}],[0,z_3^{2}z_4^{2}z_5^{2}],[0,z_3^{5}z_4z_5^{2}],[0,z_3^{6}z_5^{2}],[0,z_3^{3}z_4^{2}z_5],\\ [0,z_2z_4z_5],[0,z_2z_4^{3}],[0,z_3^{5}z_4^{2}],[0,z_3^{9}z_4],[0,z_2^{2}]\}
\end{split}&
\end{flalign*}
\\(46)$\left\{ \begin{array}{l} z_2z_3+z_4z_5=0\\z_1^{2}+z_2^{2}z_4+z_3^{6}+z_3z_4^{3}+z_5^{4}=0\end{array}\right.$
\\$(w_1,w_2,w_3,w_4,w_5;1,d)=(\frac{18}{19},\frac{13}{19},\frac{6}{19},\frac{10}{19},\frac{9}{19};1,\frac{36}{19})$
\\$\mu=77$
\begin{flalign*}
\begin{split}
mini=\{[z_1,0],[0,z_3z_4z_5^{3}],[0,z_3^{2}z_5^{3}],[0,z_2z_5^{3}],[0,z_3^{4}z_4z_5^{2}],[0,z_3^{5}z_5^{2}],[0,z_2^{2}z_5^{2}],[0,z_3^{4}z_4^{2}z_5],\\ [0,z_3^{7}z_4z_5],[0,z_3^{8}z_5],[0,z_4^{3}],[0,z_3^{5}z_4^{2}],[0,z_3^{10}z_4],[0,z_2z_4]\}
\end{split}&
\end{flalign*}
\\(47)$\left\{ \begin{array}{l} z_2z_3+z_4z_5=0\\z_1^{2}+z_2^{2}z_4+z_3^{4}z_4+z_3z_4^{3}+z_3z_5^{3}+z_2^{2}z_5=0\end{array}\right.$
\\$(w_1,w_2,w_3,w_4,w_5;1,d)=(\frac{11}{12},\frac{2}{3},\frac{1}{3},\frac{1}{2},\frac{1}{2};1,\frac{11}{6})$
\\$\mu=67$
\begin{flalign*}
\begin{split}
mini=\{[z_1,0],[0,z_3z_5^{4}],[0,z_4z_5^{3}],[0,z_3^{3}z_5^{3}],[0,z_3^{3}z_4z_5^{2}],[0,z_3^{6}z_5^{2}],[0,z_2z_5^{2}],[0,z_3^{2}z_4^{2}z_5],\\ [0,z_3^{5}z_4z_5],[0,z_3^{9}z_5],[0,z_4^{3}],[0,z_3^{3}z_4^{2}],[0,z_3^{7}z_4],[0,z_2z_4],[0,z_2^{2}]\}
\end{split}&
\end{flalign*}
\\(48)$\left\{ \begin{array}{l} z_2z_3+z_4z_5=0\\z_1^{2}+z_2^{2}z_5+2z_3^{3}z_4+z_4^{3}+z_3z_5^{3}=0\end{array}\right.$
\\$(w_1,w_2,w_3,w_4,w_5;1,d)=(\frac{27}{32},\frac{5}{8},\frac{3}{8},\frac{9}{16},\frac{7}{16};1,\frac{27}{16})$
\\$\mu=49$
\begin{flalign*}
\begin{split}
mini=\{[z_1,0],[0,z_5^{5}],[0,z_3^{2}z_5^{4}],[0,z_4z_5^{3}],[0,z_3^{3}z_5^{3}],[0,z_2z_5^{3}],[0,z_3z_4z_5^{2}],[0,z_3^{4}z_5^{2}],[0,z_4^{2}z_5],\\ [0,z_3^{2}z_4z_5],[0,z_3^{7}z_5],[0,z_2^{2}z_4]\}
\end{split}&
\end{flalign*}
\\(49)$\left\{ \begin{array}{l} z_2z_3+z_4z_5=0\\z_1^{2}+z_2^{2}z_5+z_3^{5}z_4+z_4^{3}+z_5^{5}=0\end{array}\right.$
\\$(w_1,w_2,w_3,w_4,w_5;1,d)=(\frac{15}{16},\frac{3}{4},\frac{1}{4},\frac{5}{8},\frac{3}{8};1,\frac{15}{8})$
\\$\mu=89$
\begin{flalign*}
\begin{split}
mini=\{[z_1,0],[0,z_4z_5^{4}],[0,z_3^{4}z_5^{4}],[0,z_3^{2}z_4z_5^{3}],[0,z_3^{7}z_5^{3}],[0,z_2z_5^{3}],[0,z_3^{4}z_4z_5^{2}],[0,z_3^{9}z_5^{2}],[0,z_4^{2}z_5],\\ [0,z_3^{7}z_4z_5],[0,z_3^{13}z_5],[0,z_3^{9}z_4],[0,z_2^{2}z_4]\}
\end{split}&
\end{flalign*}
\\(50)$\left\{ \begin{array}{l} z_2z_3+z_4z_5=0\\z_1^{2}+z_2^{2}z_5+z_3^{8}z_5+z_4^{3}+z_3z_5^{5}=0\end{array}\right.$
\\$(w_1,w_2,w_3,w_4,w_5;1,d)=(\frac{39}{40},\frac{4}{5},\frac{1}{5},\frac{13}{20},\frac{7}{20};1,\frac{39}{20})$
\\$\mu=121$
\begin{flalign*}
\begin{split}
mini=\{[z_1,0],[0,z_5^{7}],[0,z_3z_4z_5^{4}],[0,z_3^{7}z_5^{4}],[0,z_3^{6}z_4z_5^{3}],[0,z_3^{14}z_5^{3}],[0,z_2z_5^{3}],[0,z_4^{2}z_5],[0,z_3^{7}z_4z_5],\\ [0,z_3^{12}z_4],[0,z_2^{2}z_4],[0,z_3^{16}]\}
\end{split}&
\end{flalign*}
\\(51)$\left\{ \begin{array}{l} z_2z_3+z_4z_5=0\\z_1^{2}+z_2^{2}z_5+z_3^{4}z_5+z_2z_4^{2}+z_5^{4}+z_3^{2}z_4^{2}=0\end{array}\right.$
\\$(w_1,w_2,w_3,w_4,w_5;1,d)=(\frac{8}{9},\frac{2}{3},\frac{1}{3},\frac{5}{9},\frac{4}{9};1,\frac{16}{9})$
\\$\mu=61$
\begin{flalign*}
\begin{split}
mini=\{[z_1,0],[0,z_5^{5}],[0,z_3^{2}z_5^{4}],[0,z_4z_5^{3}],[0,z_3^{4}z_5^{3}],[0,z_2z_5^{3}],[0,z_3^{2}z_4z_5^{2}],[0,z_3^{6}z_5^{2}],[0,z_3^{3}z_4z_5],\\ [0,z_2z_4z_5],[0,z_3^{7}z_5],[0,z_3z_4^{2}],[0,z_3^{6}z_4],[0,z_3^{10}],[0,z_2^{2}]\}
\end{split}&
\end{flalign*}
\\(52)$\left\{ \begin{array}{l} z_2z_3+z_4z_5=0\\z_1^{2}+z_2^{2}z_5+z_3^{6}z_5+z_2z_4^{2}+z_3z_5^{4}+z_3^{3}z_4^{2}=0\end{array}\right.$
\\$(w_1,w_2,w_3,w_4,w_5;1,d)=(\frac{23}{24},\frac{3}{4},\frac{1}{4},\frac{7}{12},\frac{5}{12};1,\frac{23}{12})$
\\$\mu=93$
\begin{flalign*}
\begin{split}
mini=\{[z_1,0],[0,z_3^{2}z_5^{5}],[0,z_4z_5^{4}],[0,z_3^{5}z_5^{4}],[0,z_3z_4z_5^{3}],[0,z_3^{7}z_5^{3}],[0,z_2z_5^{3}],[0,z_3^{4}z_4z_5^{2}],[0,z_3^{10}z_5^{2}],\\ [0,z_3^{6}z_4z_5],[0,z_2z_4z_5],[0,z_3^{11}z_5],[0,z_3^{2}z_4^{2}],[0,z_3^{9}z_4],[0,z_3^{15}],[0,z_2^{2}]\}
\end{split}&
\end{flalign*}
\\(53)$\left\{ \begin{array}{l} z_2z_3+z_4^{2}+z_5^{3}=0\\z_1^{2}+z_2^{2}z_5+z_3^{4}+z_3z_4z_5^{2}=0\end{array}\right.$
\\$(w_1,w_2,w_3,w_4,w_5;1,d)=(\frac{7}{9},\frac{11}{18},\frac{7}{18},\frac{1}{2},\frac{1}{3};1,\frac{14}{9})$
\\$\mu=63$
\begin{flalign*}
\begin{split}
mini=\{[z_5^{2},0],[z_1z_5,0],[0,z_5^{9}],[0,z_4z_5^{6}],[0,z_3z_5^{6}],[0,z_2z_5^{6}],[0,z_3^{2}z_5^{5}],[0,z_3z_4z_5^{4}],[0,z_4^{2}z_5^{3}],\\ [0,z_3^{3}z_5^{2}],[0,z_3^{2}z_4^{2}z_5],[0,z_3z_4^{3}],[0,z_2z_4]\}
\end{split}&
\end{flalign*}
\\(54)$\left\{ \begin{array}{l} z_2z_3+z_4^{2}+z_1z_5^{2}=0\\2z_1^{2}+z_2^{2}z_5+z_3^{3}+z_3z_5^{5}=0\end{array}\right.$
\\$(w_1,w_2,w_3,w_4,w_5;1,d)=(\frac{15}{23},\frac{13}{23},\frac{10}{23},\frac{1}{2},\frac{4}{23};1,\frac{30}{23})$
\\$\mu=79$
\begin{flalign*}
\begin{split}
mini=\{[z_5^{13},0],[z_1z_5,0],[0,z_5^{13}],[0,z_4z_5^{9}],[0,z_3z_5^{9}],[0,z_2z_5^{9}],[0,z_4^{2}z_5^{4}],[0,z_3z_4z_5^{4}],[0,z_3^{2}z_5^{4}],\\ [0,z_3z_4^{2}z_5],[0,z_4^{3}],[0,z_2z_4]\}
\end{split}&
\end{flalign*}
\\(55)$\left\{ \begin{array}{l} z_1z_2+z_3^{2}+z_4^{2}+z_5^{2}=0\\z_1z_3+2z_2^{5}+z_2z_4^{2}=0\end{array}\right.$
\\$(w_1,w_2,w_3,w_4,w_5;1,d)=(\frac{3}{4},\frac{1}{4},\frac{1}{2},\frac{1}{2},\frac{1}{2};1,\frac{5}{4})$
\\$\mu=31$
\begin{flalign*}
\begin{split}
mini=\{[z_3,0],[0,z_2^{3}z_5^{2}],[0,z_4z_5],[0,z_2^{4}z_5],[0,z_4^{2}],[0,z_2^{4}z_4],[0,z_2^{3}z_3],[0,z_2^{8}]\}
\end{split}&
\end{flalign*}
\\(56)$\left\{ \begin{array}{l} z_1z_2+z_3^{2}+z_4^{2}+z_5^{3}=0\\2z_1z_3+z_2^{5}+z_2z_4^{2}=0\end{array}\right.$
\\$(w_1,w_2,w_3,w_4,w_5;1,d)=(\frac{3}{4},\frac{1}{4},\frac{1}{2},\frac{1}{2},\frac{1}{3};1,\frac{5}{4})$
\\$\mu=56$
\begin{flalign*}
\begin{split}
mini=\{[z_3z_5,0],[0,z_2^{3}z_5^{4}],[0,z_4z_5^{2}],[0,z_2^{4}z_5^{2}],[0,z_4^{2}z_5],[0,z_2^{4}z_4z_5],[0,z_2^{3}z_3z_5],[0,z_2^{8}z_5]\}
\end{split}&
\end{flalign*}
\\(57)$\left\{ \begin{array}{l} z_1z_2+z_3^{2}+z_4^{2}+z_5^{4}=0\\z_1z_3+z_2^{3}+z_2z_5^{3}=0\end{array}\right.$
\\$(w_1,w_2,w_3,w_4,w_5;1,d)=(\frac{5}{8},\frac{3}{8},\frac{1}{2},\frac{1}{2},\frac{1}{4};1,\frac{9}{8})$
\\$\mu=43$
\begin{flalign*}
\begin{split}
mini=\{[z_5^{7},0],[z_3z_5^{3},0],[0,z_5^{8}],[0,z_4z_5^{4}],[0,z_3z_5^{4}],[0,z_2z_4^{2}z_5],[0,z_2z_3z_5],[0,z_2^{2}z_5]\}
\end{split}&
\end{flalign*}
\\(58)$\left\{ \begin{array}{l} z_1z_2+z_3^{2}+z_4^{2}+z_5^{3}=0\\z_1z_3+z_2^{2}z_4+z_4z_5^{2}=0\end{array}\right.$
\\$(w_1,w_2,w_3,w_4,w_5;1,d)=(\frac{2}{3},\frac{1}{3},\frac{1}{2},\frac{1}{2},\frac{1}{3};1,\frac{7}{6})$
\\$\mu=36$
\begin{flalign*}
\begin{split}
mini=\{[z_5^{5},0],[z_3z_5,0],[0,z_5^{6}],[0,z_2z_5^{4}],[0,z_3z_5^{3}],[0,z_2^{2}z_5^{2}],[0,z_2z_3z_5],[0,z_2^{3}z_5],\\ [0,z_4^{2}],[0,z_2z_4],[0,z_2^{3}z_3]\}
\end{split}&
\end{flalign*}
\\(59)$\left\{ \begin{array}{l} z_1z_2+z_3^{2}+z_4^{2}+z_5^{n}=0\\z_1z_5+2z_3^{2}+z_4^{2}+3z_2^{n}=0\end{array}\right.$
\\$n\ge3$
\\$(w_1,w_2,w_3,w_4,w_5;1,d)=(\frac{-1 + n}{n},\frac{1}{n},\frac{1}{2},\frac{1}{2},\frac{1}{n};1,1)$
\\$\mu=-3 + 4 n + n^2$
\begin{flalign*}
\begin{split}
mini=\{[z_5^{-2 + 2 n},0],[z_4z_5^{-2 + n},0],[z_3z_5^{-2 + n},0],[0,z_5^{-2 + 2 n}],[0,z_2^{-2 + n}z_5^{-1 + n}],[0,z_2^{-1 + n}]\}
\end{split}&
\end{flalign*}
\\(60)$\left\{ \begin{array}{l} z_1z_2+z_3^{2}+z_4^{2}+z_5^{n}=0\\z_1z_5+3z_2^{1 + 2 n}+2z_2z_3^{2}+z_2z_4^{2}=0\end{array}\right.$
\\$n\ge3$
\\$(w_1,w_2,w_3,w_4,w_5;1,d)=(\frac{-1 + 2 n}{2 n},\frac{1}{2 n},\frac{1}{2},\frac{1}{2},\frac{1}{n};1,\frac{1 + 2 n}{2 n})$
\\$\mu=5 + 9 n + 2 n^2$
\begin{flalign*}
\begin{split}
mini=\{[z_5^{-1 + n},0],[0,z_2z_5^{n}],[0,z_2^{-1 + 2 n}z_5^{-1 + n}],[0,z_2^{-3 + 4 n}z_5],[0,z_4^{2}],[0,z_3z_4],[0,z_2^{2 n}z_4],\\ [0,z_2^{2 n}z_3],[0,z_2^{4 n}]\}
\end{split}&
\end{flalign*}
\\(61)$\left\{ \begin{array}{l} z_1z_2+z_3^{2}+z_4^{2}=0\\z_1z_3+z_2^{n}+z_5^{3}=0\end{array}\right.$
\\$n\ge3$
\\$(w_1,w_2,w_3,w_4,w_5;1,d)=(\frac{-1 + 2 n}{2 (1 + n)},\frac{3}{2 (1 + n)},\frac{1}{2},\frac{1}{2},\frac{n}{2 (1 + n)};1,\frac{3 n}{2 (1 + n)})$
\\$\mu=1 + 8 n$
\begin{flalign*}
\begin{split}
mini=\{[z_5^{2},0],[z_3z_5,0],[0,z_2^{-2 + n}z_4^{2}z_5],[0,z_2^{-1 + n}z_3z_5]\}
\end{split}&
\end{flalign*}
\\(62)$\left\{ \begin{array}{l} z_1z_2+z_3^{2}+z_4^{2}=0\\z_1z_3+z_2^{n}+z_5^{4}=0\end{array}\right.$
\\$3\le n \le7$
\\$(w_1,w_2,w_3,w_4,w_5;1,d)=(\frac{-1 + 2 n}{2 (1 + n)},\frac{3}{2 (1 + n)},\frac{1}{2},\frac{1}{2},\frac{3 n}{8 (1 + n)};1,\frac{3 n}{2 (1 + n)})$
\\$\mu=1 + 12 n$
\begin{flalign*}
\begin{split}
mini=\{[z_5^{3},0],[z_3z_5^{2},0],[0,z_2^{-2 + n}z_4^{2}z_5^{2}],[0,z_2^{-1 + n}z_3z_5^{2}]\}
\end{split}&
\end{flalign*}
\\(63)$\left\{ \begin{array}{l} z_1z_2+z_3^{2}+z_4^{2}=0\\z_1z_3+z_2^{n}+z_5^{5}=0\end{array}\right.$
\\$3\le n \le4$
\\$(w_1,w_2,w_3,w_4,w_5;1,d)=(\frac{-1 + 2 n}{2 (1 + n)},\frac{3}{2 (1 + n)},\frac{1}{2},\frac{1}{2},\frac{3 n}{10 (1 + n)};1,\frac{3 n}{2 (1 + n)})$
\\$\mu=1 + 16 n$
\begin{flalign*}
\begin{split}
mini=\{[z_5^{4},0],[z_3z_5^{3},0],[0,z_2^{-2 + n}z_4^{2}z_5^{3}],[0,z_2^{-1 + n}z_3z_5^{3}]\}
\end{split}&
\end{flalign*}
\\(64)$\left\{ \begin{array}{l} z_1z_2+z_3^{2}+z_4^{2}=0\\z_1z_3+z_2^{3}+z_5^{n}=0\end{array}\right.$
\\$6\le n \le8$
\\$(w_1,w_2,w_3,w_4,w_5;1,d)=(\frac{5}{8},\frac{3}{8},\frac{1}{2},\frac{1}{2},\frac{9}{8 n};1,\frac{9}{8})$
\\$\mu=-11 + 12 n$
\begin{flalign*}
\begin{split}
mini=\{[z_5^{-2 + 2 n},0],[z_3z_5^{-2 + n},0],[0,z_2z_4^{2}z_5^{-2 + n}],[0,z_2z_3z_5^{-2 + n}],[0,z_2^{2}z_5^{-2 + n}]\}
\end{split}&
\end{flalign*}
\\(65)$\left\{ \begin{array}{l} z_1z_2+z_3^{2}+z_4^{2}=0\\z_1z_3+z_2^{n}z_3+z_5^{3}+z_1z_4=0\end{array}\right.$
\\$n\ge2$
\\$(w_1,w_2,w_3,w_4,w_5;1,d)=(\frac{n}{1 + n},\frac{1}{1 + n},\frac{1}{2},\frac{1}{2},\frac{1 + 3 n}{6 (1 + n)};1,\frac{1 + 3 n}{2 (1 + n)})$
\\$\mu=5 + 12 n$
\begin{flalign*}
\begin{split}
mini=\{[z_5^{2},0],[z_4z_5,0],[0,z_2^{-2 + n}z_4^{2}z_5],[0,z_2^{-1 + 2 n}z_4z_5],[0,z_2^{-1 + n}z_3z_5],[0,z_2^{2 n}z_5]\}
\end{split}&
\end{flalign*}
\\(66)$\left\{ \begin{array}{l} z_1z_2+z_3^{2}+z_4^{2}=0\\z_1z_3+z_2^{n}z_3+z_5^{4}+z_1z_4=0\end{array}\right.$
\\$2\le n \le4$
\\$(w_1,w_2,w_3,w_4,w_5;1,d)=(\frac{n}{1 + n},\frac{1}{1 + n},\frac{1}{2},\frac{1}{2},\frac{1 + 3 n}{8 (1 + n)};1,\frac{1 + 3 n}{2 (1 + n)})$
\\$\mu=7 + 18 n$
\begin{flalign*}
\begin{split}
mini=\{[z_5^{3},0],[z_4z_5^{2},0],[0,z_2^{-2 + n}z_4^{2}z_5^{2}],[0,z_2^{-1 + 2 n}z_4z_5^{2}],[0,z_2^{-1 + n}z_3z_5^{2}],[0,z_2^{2 n}z_5^{2}]\}
\end{split}&
\end{flalign*}
\\(67)$\left\{ \begin{array}{l} z_1z_2+z_3^{2}+z_4^{2}=0\\z_1z_3+z_2^{2}z_3+z_5^{5}+z_1z_4=0\end{array}\right.$
\\$(w_1,w_2,w_3,w_4,w_5;1,d)=(\frac{2}{3},\frac{1}{3},\frac{1}{2},\frac{1}{2},\frac{7}{30};1,\frac{7}{6})$
\\$\mu=57$
\begin{flalign*}
\begin{split}
mini=\{[z_5^{8},0],[z_4z_5^{3},0],[0,z_4^{2}z_5^{3}],[0,z_2^{3}z_4z_5^{3}],[0,z_2z_3z_5^{3}]\}
\end{split}&
\end{flalign*}
\\(68)$\left\{ \begin{array}{l} z_1z_2+z_3^{2}+z_4^{2}=0\\z_1z_3+z_2^{2}z_3+z_5^{6}+z_1z_4=0\end{array}\right.$
\\$(w_1,w_2,w_3,w_4,w_5;1,d)=(\frac{2}{3},\frac{1}{3},\frac{1}{2},\frac{1}{2},\frac{7}{36};1,\frac{7}{6})$
\\$\mu=71$
\begin{flalign*}
\begin{split}
mini=\{[z_5^{10},0],[z_4z_5^{4},0],[0,z_4^{2}z_5^{4}],[0,z_2^{3}z_4z_5^{4}],[0,z_2z_3z_5^{4}]\}
\end{split}&
\end{flalign*}
\\(69)$\left\{ \begin{array}{l} z_1z_2+z_3^{2}+z_4^{2}=0\\z_1z_3+z_2^{2}z_5+z_5^{n}=0\end{array}\right.$
\\$n\ge3$
\\$(w_1,w_2,w_3,w_4,w_5;1,d)=(\frac{1 + 3 n}{2 (-1 + 3 n)},\frac{3 (-1 + n)}{2 (-1 + 3 n)},\frac{1}{2},\frac{1}{2},\frac{3}{-1 + 3 n};1,\frac{3 n}{-1 + 3 n})$
\\$\mu=1 + 8 n$
\begin{flalign*}
\begin{split}
mini=\{[z_5^{-2 + 2 n},0],[z_3z_5^{-1 + n},0],[0,z_4^{2}z_5^{-1 + n}],[0,z_3z_5^{-1 + n}],[0,z_2z_5^{-1 + n}],[0,z_2z_4],[0,z_2z_3]\}
\end{split}&
\end{flalign*}
\\(70)$\left\{ \begin{array}{l} z_1z_2+z_3^{2}+z_4^{2}=0\\z_1z_3+z_2^{3}z_5+z_5^{n}=0\end{array}\right.$
\\$3\le n \le6$
\\$(w_1,w_2,w_3,w_4,w_5;1,d)=(\frac{1 + 5 n}{2 (-1 + 4 n)},\frac{3 (-1 + n)}{2 (-1 + 4 n)},\frac{1}{2},\frac{1}{2},\frac{9}{2 (-1 + 4 n)};1,\frac{9 n}{2 (-1 + 4 n)})$
\\$\mu=1 + 12 n$
\begin{flalign*}
\begin{split}
mini=\{[z_5^{-2 + 2 n},0],[z_3z_5^{-1 + n},0],[0,z_2z_4^{2}z_5^{-1 + n}],[0,z_2z_3z_5^{-1 + n}],[0,z_2^{2}z_5^{-1 + n}],[0,z_2^{2}z_4],[0,z_2^{2}z_3]\}
\end{split}&
\end{flalign*}
\\(71)$\left\{ \begin{array}{l} z_1z_2+z_3^{2}+z_4^{2}=0\\z_1z_3+z_2^{4}z_5+z_5^{n}=0\end{array}\right.$
\\$3\le n \le4$
\\$(w_1,w_2,w_3,w_4,w_5;1,d)=(\frac{1 + 7 n}{2 (-1 + 5 n)},\frac{3 (-1 + n)}{2 (-1 + 5 n)},\frac{1}{2},\frac{1}{2},\frac{6}{-1 + 5 n};1,\frac{6 n}{-1 + 5 n})$
\\$\mu=1 + 16 n$
\begin{flalign*}
\begin{split}
mini=\{[z_5^{-2 + 2 n},0],[z_3z_5^{-1 + n},0],[0,z_2^{2}z_4^{2}z_5^{-1 + n}],[0,z_2^{2}z_3z_5^{-1 + n}],[0,z_2^{3}z_5^{-1 + n}],[0,z_2^{3}z_4],[0,z_2^{3}z_3]\}
\end{split}&
\end{flalign*}
\\(72)$\left\{ \begin{array}{l} z_1z_2+z_3^{2}+z_4^{2}=0\\z_1z_3+z_2^{5}z_5+z_5^{n}=0\end{array}\right.$
\\$3\le n \le4$
\\$(w_1,w_2,w_3,w_4,w_5;1,d)=(\frac{1 + 9 n}{2 (-1 + 6 n)},\frac{3 (-1 + n)}{2 (-1 + 6 n)},\frac{1}{2},\frac{1}{2},\frac{15}{2 (-1 + 6 n)};1,\frac{15 n}{2 (-1 + 6 n)})$
\\$\mu=1 + 20 n$
\begin{flalign*}
\begin{split}
mini=\{[z_5^{-2 + 2 n},0],[z_3z_5^{-1 + n},0],[0,z_2^{3}z_4^{2}z_5^{-1 + n}],[0,z_2^{3}z_3z_5^{-1 + n}],[0,z_2^{4}z_5^{-1 + n}],[0,z_2^{4}z_4],[0,z_2^{4}z_3]\}
\end{split}&
\end{flalign*}
\\(73)$\left\{ \begin{array}{l} z_1z_2+z_3^{2}+z_4^{2}=0\\z_1z_3+z_2^{n}z_5+z_5^{3}=0\end{array}\right.$
\\$n\ge6$
\\$(w_1,w_2,w_3,w_4,w_5;1,d)=(\frac{-1 + 3 n}{2 + 3 n},\frac{3}{2 + 3 n},\frac{1}{2},\frac{1}{2},\frac{3 n}{2 (2 + 3 n)};1,\frac{9 n}{2 (2 + 3 n)})$
\\$\mu=1 + 12 n$
\begin{flalign*}
\begin{split}
mini=\{[z_5^{2},0],[z_3z_5,0],[0,z_2^{-2 + n}z_4^{2}z_5],[0,z_2^{-1 + n}z_3z_5],[0,z_2^{-2 + 2 n}z_4^{2}],[0,z_2^{-1 + 2 n}z_3]\}
\end{split}&
\end{flalign*}
\\(74)$\left\{ \begin{array}{l} z_1z_2+z_3^{2}+z_4^{2}=0\\z_1z_5+z_2^{2 n}+z_5^{n}+2z_2z_3^{2}+z_2z_4^{2}=0\end{array}\right.$
\\$n\ge3$
\\$(w_1,w_2,w_3,w_4,w_5;1,d)=(\frac{2 (-1 + n)}{-1 + 2 n},\frac{1}{-1 + 2 n},\frac{1}{2},\frac{1}{2},\frac{2}{-1 + 2 n};1,\frac{2 n}{-1 + 2 n})$
\\$\mu=1 + 7 n + 2 n^2$
\begin{flalign*}
\begin{split}
mini=\{[z_5^{-1 + n},0],[0,z_2^{3}z_5^{-1 + n}],[0,z_2^{-1 + 2 n}z_5^{-2 + n}],[0,z_2^{-5 + 4 n}z_5],[0,z_4^{2}],[0,z_3z_4],[0,z_2^{-1 + 2 n}z_4],\\ [0,z_2^{-1 + 2 n}z_3],[0,z_2^{-2 + 4 n}]\}
\end{split}&
\end{flalign*}
\\(75)$\left\{ \begin{array}{l} z_1z_2+z_3^{2}+z_4^{3}=0\\z_1z_4+z_2^{3}+z_5^{n}=0\end{array}\right.$
\\$3\le n \le5$
\\$(w_1,w_2,w_3,w_4,w_5;1,d)=(\frac{2}{3},\frac{1}{3},\frac{1}{2},\frac{1}{3},\frac{1}{n};1,1)$
\\$\mu=2 (-7 + 8 n)$
\begin{flalign*}
\begin{split}
mini=\{[z_4z_5^{-2 + 2 n},0],[z_4^{2}z_5^{-2 + n},0],[0,z_2z_4^{3}z_5^{-2 + n}],[0,z_2z_3z_5^{-2 + n}],[0,z_2^{2}z_5^{-2 + n}]\}
\end{split}&
\end{flalign*}
\\(76)$\left\{ \begin{array}{l} z_1z_2+z_3^{2}+z_4^{3}+z_5^{3}=0\\z_1z_4+z_2^{3}z_4+z_2z_3z_5=0\end{array}\right.$
\\$(w_1,w_2,w_3,w_4,w_5;1,d)=(\frac{3}{4},\frac{1}{4},\frac{1}{2},\frac{1}{3},\frac{1}{3};1,\frac{13}{12})$
\\$\mu=53$
\begin{flalign*}
\begin{split}
mini=\{[z_5^{2},0],[z_4^{2}z_5,0],[0,z_5^{4}],[0,z_2^{2}z_5^{3}],[0,z_2z_4z_5^{2}],[0,z_3z_5^{2}],[0,z_2^{5}z_5^{2}],[0,z_2z_4^{2}z_5],[0,z_2^{4}z_4z_5],\\ [0,z_4^{3}],[0,z_2^{5}z_4],[0,z_2^{8}]\}
\end{split}&
\end{flalign*}
\\(77)$\left\{ \begin{array}{l} z_1z_2+z_3^{2}+z_5^{n}=0\\z_1z_5+z_2^{n}+z_4^{3}=0\end{array}\right.$
\\$4\le n \le5$
\\$(w_1,w_2,w_3,w_4,w_5;1,d)=(\frac{-1 + n}{n},\frac{1}{n},\frac{1}{2},\frac{1}{3},\frac{1}{n};1,1)$
\\$\mu=-5 + 7 n + 2 n^2$
\begin{flalign*}
\begin{split}
mini=\{[z_4z_5^{-1 + n},0],[z_4^{2}z_5^{-2 + n},0],[0,z_2^{-2 + n}z_4z_5^{n}],[0,z_2^{-1 + n}z_4z_5^{-1 + n}],[0,z_2^{-2 + n}z_3z_4]\}
\end{split}&
\end{flalign*}
\\(78)$\left\{ \begin{array}{l} z_1z_2+z_3^{2}+z_4^{3}+z_5^{4}=0\\z_1z_5+z_2^{5}+z_2z_3z_4=0\end{array}\right.$
\\$(w_1,w_2,w_3,w_4,w_5;1,d)=(\frac{19}{24},\frac{5}{24},\frac{1}{2},\frac{1}{3},\frac{1}{4};1,\frac{25}{24})$
\\$\mu=71$
\begin{flalign*}
\begin{split}
mini=\{[z_4z_5^{3},0],[z_4^{2},0],[0,z_4z_5^{4}],[0,z_2^{3}z_5^{4}],[0,z_2^{2}z_4z_5^{3}],[0,z_2^{6}z_5^{3}],[0,z_2^{3}z_4z_5^{2}],[0,z_2^{7}z_5^{2}],[0,z_2z_4^{2}z_5],\\ [0,z_4^{3}],[0,z_3z_4^{2}],[0,z_2^{5}z_4],[0,z_2^{8}]\}
\end{split}&
\end{flalign*}
\\(79)$\left\{ \begin{array}{l} z_1z_2+z_3^{2}+z_4^{3}=0\\z_1z_3+z_2^{2}z_5+z_5^{4}=0\end{array}\right.$
\\$(w_1,w_2,w_3,w_4,w_5;1,d)=(\frac{13}{22},\frac{9}{22},\frac{1}{2},\frac{1}{3},\frac{3}{11};1,\frac{12}{11})$
\\$\mu=61$
\begin{flalign*}
\begin{split}
mini=\{[z_4z_5^{6},0],[z_3z_4z_5^{3},0],[0,z_4^{4}z_5^{3}],[0,z_3z_4z_5^{3}],[0,z_2z_4z_5^{3}],[0,z_2z_4^{2}],[0,z_2z_3z_4]\}
\end{split}&
\end{flalign*}
\\(80)$\left\{ \begin{array}{l} z_1z_2+z_3^{2}+z_4^{3}=0\\z_1z_4+z_2^{4}+z_5^{4}=0\end{array}\right.$
\\$(w_1,w_2,w_3,w_4,w_5;1,d)=(\frac{11}{15},\frac{4}{15},\frac{1}{2},\frac{1}{3},\frac{4}{15};1,\frac{16}{15})$
\\$\mu=65$
\begin{flalign*}
\begin{split}
mini=\{[z_4z_5^{6},0],[z_4^{2}z_5^{2},0],[0,z_2^{2}z_4^{3}z_5^{2}],[0,z_2^{2}z_3z_5^{2}],[0,z_2^{3}z_5^{2}]\}
\end{split}&
\end{flalign*}
\\(81)$\left\{ \begin{array}{l} z_1z_2+z_3^{2}+z_4^{3}=0\\z_1z_4+z_2^{2}z_3+z_5^{4}=0\end{array}\right.$
\\$(w_1,w_2,w_3,w_4,w_5;1,d)=(\frac{13}{18},\frac{5}{18},\frac{1}{2},\frac{1}{3},\frac{19}{72};1,\frac{19}{18})$
\\$\mu=62$
\begin{flalign*}
\begin{split}
mini=\{[z_4z_5^{6},0],[z_4^{2}z_5^{2},0],[0,z_4^{3}z_5^{2}],[0,z_2^{3}z_4^{2}z_5^{2}],[0,z_2z_3z_5^{2}]\}
\end{split}&
\end{flalign*}
\\(82)$\left\{ \begin{array}{l} z_1z_2+z_3^{2}+z_4^{2}z_5+z_5^{5}=0\\z_1z_4+z_2^{3}+z_2z_3z_5=0\end{array}\right.$
\\$(w_1,w_2,w_3,w_4,w_5;1,d)=(\frac{13}{20},\frac{7}{20},\frac{1}{2},\frac{2}{5},\frac{1}{5};1,\frac{21}{20})$
\\$\mu=61$
\begin{flalign*}
\begin{split}
mini=\{[z_5^{8},0],[z_4z_5^{4},0],[z_4^{2},0],[0,z_5^{10}],[0,z_2z_5^{8}],[0,z_4z_5^{6}],[0,z_3z_5^{5}],[0,z_2z_4z_5^{4}],[0,z_2^{2}z_5^{3}],\\ [0,z_4^{2}z_5^{2}],[0,z_2z_4^{2}]\}
\end{split}&
\end{flalign*}
\\(83)$\left\{ \begin{array}{l} z_1z_2+z_3^{2}+z_4^{2}z_5=0\\z_1z_3+z_2^{2}z_5+z_5^{4}+z_4^{3}=0\end{array}\right.$
\\$(w_1,w_2,w_3,w_4,w_5;1,d)=(\frac{13}{22},\frac{9}{22},\frac{1}{2},\frac{4}{11},\frac{3}{11};1,\frac{12}{11})$
\\$\mu=54$
\begin{flalign*}
\begin{split}
mini=\{[z_5^{7},0],[z_4z_5^{6},0],[z_3z_5^{4},0],[z_3z_4z_5,0],[z_4^{2},0],[0,z_5^{6}],[0,z_4z_5^{5}],[0,z_3z_5^{4}],[0,z_2z_5^{4}],\\ [0,z_2z_4z_5^{2}],[0,z_3z_4z_5],[0,z_2z_3z_5],[0,z_2z_3z_4]\}
\end{split}&
\end{flalign*}
\\(84)$\left\{ \begin{array}{l} z_1z_2+z_3^{2}+z_4^{2}z_5=0\\z_1z_4+z_2^{3}+z_5^{n}=0\end{array}\right.$
\\$3\le n \le7$
\\$(w_1,w_2,w_3,w_4,w_5;1,d)=(\frac{3 + 5 n}{3 + 8 n},\frac{3 n}{3 + 8 n},\frac{1}{2},\frac{-3 + 4 n}{3 + 8 n},\frac{9}{3 + 8 n};1,\frac{9 n}{3 + 8 n})$
\\$\mu=-2 + 12 n$
\begin{flalign*}
\begin{split}
mini=\{[z_5^{-1 + 2 n},0],[z_4z_5^{-1 + n},0],[0,z_2z_4^{2}z_5^{-1 + n}],[0,z_2^{2}z_5^{-1 + n}],[0,z_2z_3z_5^{-2 + n}]\}
\end{split}&
\end{flalign*}
\\(85)$\left\{ \begin{array}{l} z_1z_2+z_3^{2}+z_4^{2}z_5=0\\z_1z_4+z_2^{4}+z_5^{4}=0\end{array}\right.$
\\$(w_1,w_2,w_3,w_4,w_5;1,d)=(\frac{8}{11},\frac{3}{11},\frac{1}{2},\frac{4}{11},\frac{3}{11};1,\frac{12}{11})$
\\$\mu=61$
\begin{flalign*}
\begin{split}
mini=\{[z_5^{7},0],[z_4z_5^{3},0],[0,z_2^{2}z_4^{2}z_5^{3}],[0,z_2^{3}z_5^{3}],[0,z_2^{2}z_3z_5^{2}]\}
\end{split}&
\end{flalign*}
\\(86)$\left\{ \begin{array}{l} z_1z_2+z_3^{2}+z_4^{2}z_5=0\\z_1z_4+z_2^{5}+z_5^{4}=0\end{array}\right.$
\\$(w_1,w_2,w_3,w_4,w_5;1,d)=(\frac{41}{53},\frac{12}{53},\frac{1}{2},\frac{19}{53},\frac{15}{53};1,\frac{60}{53})$
\\$\mu=76$
\begin{flalign*}
\begin{split}
mini=\{[z_4z_5^{3},0],[0,z_2^{3}z_4^{2}z_5^{3}],[0,z_2^{4}z_4z_5^{3}],[0,z_2^{3}z_3z_5^{2}]\}
\end{split}&
\end{flalign*}
\\(87)$\left\{ \begin{array}{l} z_1z_2+z_3^{2}+z_4^{2}z_5=0\\z_1z_4+z_2^{2}z_3+z_5^{4}=0\end{array}\right.$
\\$(w_1,w_2,w_3,w_4,w_5;1,d)=(\frac{37}{52},\frac{15}{52},\frac{1}{2},\frac{19}{52},\frac{7}{26};1,\frac{14}{13})$
\\$\mu=57$
\begin{flalign*}
\begin{split}
mini=\{[z_5^{7},0],[z_4z_5^{3},0],[0,z_4^{2}z_5^{3}],[0,z_2^{3}z_4z_5^{3}],[0,z_2z_3z_5^{3}],[0,z_2z_4^{2}]\}
\end{split}&
\end{flalign*}
\\(88)$\left\{ \begin{array}{l} z_1z_2+z_3^{2}+z_4^{2}z_5=0\\z_1z_4+z_2^{2}z_3+z_5^{5}=0\end{array}\right.$
\\$(w_1,w_2,w_3,w_4,w_5;1,d)=(\frac{45}{64},\frac{19}{64},\frac{1}{2},\frac{25}{64},\frac{7}{32};1,\frac{35}{32})$
\\$\mu=71$
\begin{flalign*}
\begin{split}
mini=\{[z_5^{9},0],[z_4z_5^{4},0],[0,z_4^{2}z_5^{4}],[0,z_2^{3}z_4z_5^{4}],[0,z_2z_3z_5^{4}],[0,z_2z_4^{2}]\}
\end{split}&
\end{flalign*}
\\(89)$\left\{ \begin{array}{l} z_1z_2+z_3^{2}+z_4^{2}z_5=0\\z_1z_5+z_2^{6}+z_4^{3}+z_5^{4}=0\end{array}\right.$
\\$(w_1,w_2,w_3,w_4,w_5;1,d)=(\frac{9}{11},\frac{2}{11},\frac{1}{2},\frac{4}{11},\frac{3}{11};1,\frac{12}{11})$
\\$\mu=74$
\begin{flalign*}
\begin{split}
mini=\{[z_5^{4},0],[z_4z_5^{2},0],[z_4^{2},0],[0,z_2^{3}z_4z_5^{3}],[0,z_2^{5}z_5^{3}],[0,z_2^{4}z_4z_5^{2}],[0,z_2^{10}z_5],[0,z_2^{4}z_3z_4],\\ [0,z_2^{5}z_4],[0,z_2^{11}]\}
\end{split}&
\end{flalign*}
\\(90)$\left\{ \begin{array}{l} z_1z_2+z_3^{2}+z_5^{3}=0\\z_1z_3+z_2^{3}+z_4^{3}=0\end{array}\right.$
\\$(w_1,w_2,w_3,w_4,w_5;1,d)=(\frac{5}{8},\frac{3}{8},\frac{1}{2},\frac{3}{8},\frac{1}{3};1,\frac{9}{8})$
\\$\mu=46$
\begin{flalign*}
\begin{split}
mini=\{[z_4^{4}z_5,0],[z_3z_4z_5,0],[0,z_2z_4z_5^{4}],[0,z_2z_3z_4z_5],[0,z_2^{2}z_4z_5]\}
\end{split}&
\end{flalign*}
\\(91)$\left\{ \begin{array}{l} z_1z_2+z_3^{2}+z_5^{3}=0\\z_1z_3+z_2^{4}+z_4^{3}=0\end{array}\right.$
\\$(w_1,w_2,w_3,w_4,w_5;1,d)=(\frac{7}{10},\frac{3}{10},\frac{1}{2},\frac{2}{5},\frac{1}{3};1,\frac{6}{5})$
\\$\mu=60$
\begin{flalign*}
\begin{split}
mini=\{[z_4^{2}z_5,0],[z_3z_4z_5,0],[0,z_2^{2}z_4z_5^{4}],[0,z_2^{3}z_3z_4z_5]\}
\end{split}&
\end{flalign*}
\\(92)$\left\{ \begin{array}{l} z_1z_2+z_3^{2}+z_5^{3}=0\\z_1z_3+z_2^{3}z_4+z_4^{3}=0\end{array}\right.$
\\$(w_1,w_2,w_3,w_4,w_5;1,d)=(\frac{8}{11},\frac{3}{11},\frac{1}{2},\frac{9}{22},\frac{1}{3};1,\frac{27}{22})$
\\$\mu=67$
\begin{flalign*}
\begin{split}
mini=\{[z_4^{2}z_5,0],[z_3z_4z_5,0],[0,z_2z_4z_5^{4}],[0,z_2^{4}z_5^{4}],[0,z_2^{2}z_3z_4z_5],[0,z_2^{5}z_3z_5]\}
\end{split}&
\end{flalign*}
\\(93)$\left\{ \begin{array}{l} z_1z_2+z_3^{2}+z_5^{3}=0\\z_1z_3+z_2^{2}z_5+z_4^{3}=0\end{array}\right.$
\\$(w_1,w_2,w_3,w_4,w_5;1,d)=(\frac{11}{18},\frac{7}{18},\frac{1}{2},\frac{10}{27},\frac{1}{3};1,\frac{10}{9})$
\\$\mu=44$
\begin{flalign*}
\begin{split}
mini=\{[z_3z_4z_5^{2},0],[z_4^{2}z_5,0],[0,z_4z_5^{5}],[0,z_2z_3z_4z_5^{2}]\}
\end{split}&
\end{flalign*}
\\(94)$\left\{ \begin{array}{l} z_1z_2+z_3^{2}+z_5^{4}=0\\z_1z_3+z_2^{2}z_5+z_4^{3}=0\end{array}\right.$
\\$(w_1,w_2,w_3,w_4,w_5;1,d)=(\frac{7}{12},\frac{5}{12},\frac{1}{2},\frac{13}{36},\frac{1}{4};1,\frac{13}{12})$
\\$\mu=59$
\begin{flalign*}
\begin{split}
mini=\{[z_3z_4z_5^{3},0],[z_4^{2}z_5^{2},0],[0,z_4z_5^{7}],[0,z_2z_3z_4z_5^{3}]\}
\end{split}&
\end{flalign*}
\\(95)$\left\{ \begin{array}{l} z_1z_2+z_3^{2}+z_5^{3}=0\\z_1z_3+z_2^{3}z_5+z_4^{3}=0\end{array}\right.$
\\$(w_1,w_2,w_3,w_4,w_5;1,d)=(\frac{17}{24},\frac{7}{24},\frac{1}{2},\frac{29}{72},\frac{1}{3};1,\frac{29}{24})$
\\$\mu=62$
\begin{flalign*}
\begin{split}
mini=\{[z_4^{2}z_5,0],[z_3z_4z_5,0],[0,z_2z_4z_5^{5}],[0,z_2^{2}z_4z_5^{2}],[0,z_2^{2}z_3z_4z_5],[0,z_2^{3}z_3z_4],[0,z_2^{6}z_4]\}
\end{split}&
\end{flalign*}
\\(96)$\left\{ \begin{array}{l} z_1z_2+z_3^{2}+z_5^{3}=0\\z_1z_4+z_2^{6}+z_4^{3}+z_2z_3^{2}=0\end{array}\right.$
\\$(w_1,w_2,w_3,w_4,w_5;1,d)=(\frac{4}{5},\frac{1}{5},\frac{1}{2},\frac{2}{5},\frac{1}{3};1,\frac{6}{5})$
\\$\mu=73$
\begin{flalign*}
\begin{split}
mini=\{[z_4^{2}z_5,0],[0,z_2^{4}z_5^{4}],[0,z_3z_5^{2}],[0,z_2^{5}z_5^{2}],[0,z_2z_4^{2}z_5],[0,z_2^{5}z_4z_5],[0,z_2^{5}z_3z_5],[0,z_2^{10}z_5]\}
\end{split}&
\end{flalign*}
\\(97)$\left\{ \begin{array}{l} z_1z_2+z_3^{2}+z_5^{3}=0\\z_1z_4+z_2^{3}z_3+z_4^{3}+z_3z_5^{2}=0\end{array}\right.$
\\$(w_1,w_2,w_3,w_4,w_5;1,d)=(\frac{7}{9},\frac{2}{9},\frac{1}{2},\frac{7}{18},\frac{1}{3};1,\frac{7}{6})$
\\$\mu=63$
\begin{flalign*}
\begin{split}
mini=\{[z_5^{2},0],[z_4^{2}z_5,0],[0,z_5^{4}],[0,z_2^{3}z_5^{3}],[0,z_2^{3}z_4z_5^{2}],[0,z_2^{4}z_5^{2}],[0,z_4^{2}z_5],[0,z_2^{6}z_4z_5],[0,z_2^{7}z_5],\\ [0,z_2^{2}z_4^{2}],[0,z_2^{4}z_3],[0,z_2^{10}]\}
\end{split}&
\end{flalign*}
\\(98)$\left\{ \begin{array}{l} z_1z_2+z_3^{2}+z_5^{3}=0\\z_1z_5+z_2^{n}+z_4^{3}=0\end{array}\right.$
\\$3\le n \le14$
\\$(w_1,w_2,w_3,w_4,w_5;1,d)=(\frac{-1 + 3 n}{3 (1 + n)},\frac{4}{3 (1 + n)},\frac{1}{2},\frac{4 n}{9 (1 + n)},\frac{1}{3};1,\frac{4 n}{3 (1 + n)})$
\\$\mu=4 + 10 n$
\begin{flalign*}
\begin{split}
mini=\{[z_4z_5^{2},0],[z_4^{2}z_5,0],[0,z_2^{-2 + n}z_4z_5^{3}],[0,z_2^{-1 + n}z_4z_5^{2}],[0,z_2^{-2 + n}z_3z_4]\}
\end{split}&
\end{flalign*}
\\(99)$\left\{ \begin{array}{l} z_1z_2+z_3^{2}+z_5^{4}=0\\z_1z_5+z_2^{n}+z_4^{3}=0\end{array}\right.$
\\$4\le n \le8$
\\$(w_1,w_2,w_3,w_4,w_5;1,d)=(\frac{-1 + 4 n}{4 (1 + n)},\frac{5}{4 (1 + n)},\frac{1}{2},\frac{5 n}{12 (1 + n)},\frac{1}{4};1,\frac{5 n}{4 (1 + n)})$
\\$\mu=7 + 12 n$
\begin{flalign*}
\begin{split}
mini=\{[z_4z_5^{3},0],[z_4^{2}z_5^{2},0],[0,z_2^{-2 + n}z_4z_5^{4}],[0,z_2^{-1 + n}z_4z_5^{3}],[0,z_2^{-2 + n}z_3z_4]\}
\end{split}&
\end{flalign*}
\\(100)$\left\{ \begin{array}{l} z_1z_2+z_3^{2}+z_5^{5}=0\\z_1z_5+z_2^{6}+z_4^{3}=0\end{array}\right.$
\\$(w_1,w_2,w_3,w_4,w_5;1,d)=(\frac{29}{35},\frac{6}{35},\frac{1}{2},\frac{12}{35},\frac{1}{5};1,\frac{36}{35})$
\\$\mu=94$
\begin{flalign*}
\begin{split}
mini=\{[z_4z_5^{4},0],[z_4^{2}z_5^{3},0],[0,z_2^{4}z_4z_5^{5}],[0,z_2^{5}z_4z_5^{4}],[0,z_2^{4}z_3z_4]\}
\end{split}&
\end{flalign*}
\\(101)$\left\{ \begin{array}{l} z_1z_2+z_3^{2}+z_5^{3}=0\\z_1z_5+z_2^{n}z_3+z_4^{3}=0\end{array}\right.$
\\$2\le n \le8$
\\$(w_1,w_2,w_3,w_4,w_5;1,d)=(\frac{1 + 6 n}{6 (1 + n)},\frac{5}{6 (1 + n)},\frac{1}{2},\frac{3 + 8 n}{18 (1 + n)},\frac{1}{3};1,\frac{3 + 8 n}{6 (1 + n)})$
\\$\mu=2 (5 + 8 n)$
\begin{flalign*}
\begin{split}
mini=\{[z_4z_5^{2},0],[z_4^{2}z_5,0],[0,z_2^{-2 + n}z_4z_5^{3}],[0,z_2^{-1 + 2 n}z_4z_5^{2}],[0,z_2^{2 n}z_4z_5],[0,z_2^{-1 + n}z_3z_4]\}
\end{split}&
\end{flalign*}
\\(102)$\left\{ \begin{array}{l} z_1z_2+z_3^{2}+z_5^{4}=0\\z_1z_5+z_2^{3}z_3+z_4^{3}=0\end{array}\right.$
\\$(w_1,w_2,w_3,w_4,w_5;1,d)=(\frac{13}{16},\frac{3}{16},\frac{1}{2},\frac{17}{48},\frac{1}{4};1,\frac{17}{16})$
\\$\mu=75$
\begin{flalign*}
\begin{split}
mini=\{[z_4z_5^{3},0],[z_4^{2}z_5^{2},0],[0,z_2z_4z_5^{4}],[0,z_2^{5}z_4z_5^{3}],[0,z_2^{6}z_4z_5^{2}],[0,z_2^{2}z_3z_4]\}
\end{split}&
\end{flalign*}
\\(103)$\left\{ \begin{array}{l} z_1z_2+z_3^{2}+z_5^{4}=0\\z_1z_5+z_2^{4}z_3+z_4^{3}=0\end{array}\right.$
\\$(w_1,w_2,w_3,w_4,w_5;1,d)=(\frac{17}{20},\frac{3}{20},\frac{1}{2},\frac{11}{30},\frac{1}{4};1,\frac{11}{10})$
\\$\mu=95$
\begin{flalign*}
\begin{split}
mini=\{[z_4z_5^{3},0],[z_4^{2}z_5^{2},0],[0,z_2^{2}z_4z_5^{4}],[0,z_2^{7}z_4z_5^{3}],[0,z_2^{8}z_4z_5^{2}],[0,z_2^{3}z_3z_4]\}
\end{split}&
\end{flalign*}
\\(104)$\left\{ \begin{array}{l} z_1z_2+z_3^{2}+z_5^{5}=0\\z_1z_5+z_2^{3}z_3+z_4^{3}=0\end{array}\right.$
\\$(w_1,w_2,w_3,w_4,w_5;1,d)=(\frac{33}{40},\frac{7}{40},\frac{1}{2},\frac{41}{120},\frac{1}{5};1,\frac{41}{40})$
\\$\mu=92$
\begin{flalign*}
\begin{split}
mini=\{[z_4z_5^{4},0],[z_4^{2}z_5^{3},0],[0,z_2z_4z_5^{5}],[0,z_2^{5}z_4z_5^{4}],[0,z_2^{6}z_4z_5^{3}],[0,z_2^{2}z_3z_4]\}
\end{split}&
\end{flalign*}
\\(105)$\left\{ \begin{array}{l} z_1z_2+z_3^{2}+z_5^{3}=0\\z_1z_5+z_2^{n}z_4+z_4^{3}=0\end{array}\right.$
\\$2\le n \le9$
\\$(w_1,w_2,w_3,w_4,w_5;1,d)=(\frac{-2 + 9 n}{3 (2 + 3 n)},\frac{8}{3 (2 + 3 n)},\frac{1}{2},\frac{4 n}{3 (2 + 3 n)},\frac{1}{3};1,\frac{4 n}{2 + 3 n})$
\\$\mu=4 + 15 n$
\begin{flalign*}
\begin{split}
mini=\{[z_4z_5^{2},0],[z_4^{2}z_5,0],[0,z_2^{-2 + n}z_4z_5^{3}],[0,z_2^{-2 + 2 n}z_5^{3}],[0,z_2^{-1 + n}z_4z_5^{2}],[0,z_2^{-1 + 2 n}z_5^{2}],\\ [0,z_2^{-2 + n}z_3z_4],[0,z_2^{-2 + 2 n}z_3]\}
\end{split}&
\end{flalign*}
\\(106)$\left\{ \begin{array}{l} z_1z_2+z_3^{2}+z_5^{4}=0\\z_1z_5+z_2^{3}z_4+z_4^{3}=0\end{array}\right.$
\\$(w_1,w_2,w_3,w_4,w_5;1,d)=(\frac{17}{22},\frac{5}{22},\frac{1}{2},\frac{15}{44},\frac{1}{4};1,\frac{45}{44})$
\\$\mu=61$
\begin{flalign*}
\begin{split}
mini=\{[z_4z_5^{3},0],[z_4^{2}z_5^{2},0],[0,z_2z_4z_5^{4}],[0,z_2^{4}z_5^{4}],[0,z_2^{2}z_4z_5^{3}],[0,z_2^{5}z_5^{3}],[0,z_2z_3z_4],[0,z_2^{4}z_3]\}
\end{split}&
\end{flalign*}
\\(107)$\left\{ \begin{array}{l} z_1z_2+z_3^{2}+z_5^{4}=0\\z_1z_5+z_2^{5}z_4+z_4^{3}=0\end{array}\right.$
\\$(w_1,w_2,w_3,w_4,w_5;1,d)=(\frac{29}{34},\frac{5}{34},\frac{1}{2},\frac{25}{68},\frac{1}{4};1,\frac{75}{68})$
\\$\mu=97$
\begin{flalign*}
\begin{split}
mini=\{[z_4z_5^{3},0],[z_4^{2}z_5^{2},0],[0,z_2^{3}z_4z_5^{4}],[0,z_2^{8}z_5^{4}],[0,z_2^{4}z_4z_5^{3}],[0,z_2^{9}z_5^{3}],[0,z_2^{3}z_3z_4],[0,z_2^{8}z_3]\}
\end{split}&
\end{flalign*}
\\(108)$\left\{ \begin{array}{l} z_1z_2+z_3^{2}+z_5^{4}=0\\z_1z_5+z_2^{4}z_5+z_4^{3}+z_2z_3z_4=0\end{array}\right.$
\\$(w_1,w_2,w_3,w_4,w_5;1,d)=(\frac{4}{5},\frac{1}{5},\frac{1}{2},\frac{7}{20},\frac{1}{4};1,\frac{21}{20})$
\\$\mu=70$
\begin{flalign*}
\begin{split}
mini=\{[z_4z_5^{3},0],[z_4^{2}z_5^{2},0],[0,z_2^{3}z_5^{4}],[0,z_2^{2}z_4z_5^{3}],[0,z_2^{6}z_5^{3}],[0,z_2^{3}z_4z_5^{2}],[0,z_2^{7}z_5],[0,z_2^{2}z_4^{2}],\\ [0,z_3z_4],[0,z_2^{5}z_4],[0,z_2^{10}]\}
\end{split}&
\end{flalign*}
\\(109)$\left\{ \begin{array}{l} z_1z_2+z_3^{2}+z_4z_5^{2}=0\\z_1z_5+z_2^{n}+z_4^{3}=0\end{array}\right.$
\\$3\le n \le11$
\\$(w_1,w_2,w_3,w_4,w_5;1,d)=(\frac{-3 + 7 n}{6 + 7 n},\frac{9}{6 + 7 n},\frac{1}{2},\frac{3 n}{6 + 7 n},\frac{3 + 2 n}{6 + 7 n};1,\frac{9 n}{6 + 7 n})$
\\$\mu=1 + 11 n$
\begin{flalign*}
\begin{split}
mini=\{[z_5^{4},0],[z_4^{2}z_5,0],[0,z_2^{-3 + n}z_5^{5}],[0,z_2^{-2 + n}z_5^{4}],[0,z_2^{-2 + n}z_4z_5^{2}],[0,z_2^{-1 + n}z_4z_5],[0,z_2^{-2 + n}z_3z_4]\}
\end{split}&
\end{flalign*}
\\(110)$\left\{ \begin{array}{l} z_1z_2+z_3^{2}+z_4z_5^{3}=0\\z_1z_5+z_2^{n}+z_4^{3}=0\end{array}\right.$
\\$5\le n \le7$
\\$(w_1,w_2,w_3,w_4,w_5;1,d)=(\frac{-3 + 10 n}{9 + 10 n},\frac{12}{9 + 10 n},\frac{1}{2},\frac{4 n}{9 + 10 n},\frac{3 + 2 n}{9 + 10 n};1,\frac{12 n}{9 + 10 n})$
\\$\mu=4 + 14 n$
\begin{flalign*}
\begin{split}
mini=\{[z_5^{6},0],[z_4^{2}z_5^{2},0],[0,z_2^{-3 + n}z_5^{7}],[0,z_2^{-2 + n}z_5^{6}],[0,z_2^{-2 + n}z_4z_5^{3}],[0,z_2^{-1 + n}z_4z_5^{2}],[0,z_2^{-2 + n}z_3z_4]\}
\end{split}&
\end{flalign*}
\\(111)$\left\{ \begin{array}{l} z_1z_2+z_3^{2}+z_4z_5^{2}=0\\z_1z_5+z_2^{n}z_3+z_4^{3}=0\end{array}\right.$
\\$2\le n \le6$
\\$(w_1,w_2,w_3,w_4,w_5;1,d)=(\frac{1 + 14 n}{2 (6 + 7 n)},\frac{11}{2 (6 + 7 n)},\frac{1}{2},\frac{1 + 3 n}{6 + 7 n},\frac{5 + 4 n}{2 (6 + 7 n)};1,\frac{3 (1 + 3 n)}{6 + 7 n})$
\\$\mu=7 + 18 n$
\begin{flalign*}
\begin{split}
mini=\{[z_5^{4},0],[z_4^{2}z_5,0],[0,z_2^{-3 + n}z_5^{5}],[0,z_2^{-2 + 2 n}z_5^{4}],[0,z_2^{-1 + 2 n}z_5^{3}],[0,z_2^{-2 + n}z_4z_5^{2}],[0,z_2^{-1 + 2 n}z_4z_5],\\ [0,z_2^{-1 + n}z_3z_4],[0,z_2^{2 n}z_4]\}
\end{split}&
\end{flalign*}
\\(112)$\left\{ \begin{array}{l} z_1z_2+z_3^{2}+z_4z_5^{3}=0\\z_1z_5+z_2^{3}z_3+z_4^{3}=0\end{array}\right.$
\\$(w_1,w_2,w_3,w_4,w_5;1,d)=(\frac{32}{39},\frac{7}{39},\frac{1}{2},\frac{9}{26},\frac{17}{78};1,\frac{27}{26})$
\\$\mu=85$
\begin{flalign*}
\begin{split}
mini=\{[z_5^{6},0],[z_4^{2}z_5^{2},0],[0,z_5^{7}],[0,z_2^{4}z_5^{6}],[0,z_2^{5}z_5^{5}],[0,z_2z_4z_5^{3}],[0,z_2^{5}z_4z_5^{2}],[0,z_2^{6}z_4z_5],[0,z_2^{2}z_3z_4]\}
\end{split}&
\end{flalign*}
\\(113)$\left\{ \begin{array}{l} z_1z_2+z_4z_5=0\\z_1z_4+z_3^{2}+z_2^{n_2}+z_5^{n_5}=0\end{array}\right.$
\\$n_2 \ge 3~and~2\le n_5 \le n_2$
\\$(w_1,w_2,w_3,w_4,w_5;1,d)=(\frac{n_2 - n_5 + n_2 n_5}{n_2 + n_5 + n_2 n_5},\frac{2 n_5}{n_2 + n_5 + n_2 n_5},\frac{n_2 n_5}{n_2 + n_5 + n_2 n_5},\frac{-n_2 + n_5 + n_2 n_5}{n_2 + n_5 + n_2 n_5},\frac{2 n_2}{n_2 + n_5 + n_2 n_5};1,\frac{2 n_2 n_5}{n_2 + n_5 + n_2 n_5})$
\\$\mu=(1 + n_2) (1 + n_5)$
\begin{flalign*}
\begin{split}
mini=\{[z_5^{-1 + n_5},0],[z_3,0],[0,z_2^{-1 + n_2}z_5^{-1 + n_5}],[0,z_2^{-1 + n_2}z_4]\}
\end{split}&
\end{flalign*}
\\(114)$\left\{ \begin{array}{l} z_1z_2+z_4z_5=0\\z_1z_4+z_3^{2}+z_2^{n_2}+z_1z_5^{n_5}=0\end{array}\right.$
\\$n_2 \ge 3~and~2\le 2 n_5 \le n_2$
\\$(w_1,w_2,w_3,w_4,w_5;1,d)=(\frac{n_2 - n_5 + n_2 n_5}{(1 + n_2) (1 + n_5)},\frac{1 + 2 n_5}{(1 + n_2) (1 + n_5)},\frac{n_2 + 2 n_2 n_5}{2 (1 + n_2) (1 + n_5)},\frac{n_5}{1 + n_5},\frac{1}{1 + n_5};1,\frac{n_2 + 2 n_2 n_5}{(1 + n_2) (1 + n_5)})$
\\$\mu=1 + 2 n_2 (1 + n_5)$
\begin{flalign*}
\begin{split}
mini=\{[z_5^{n_5},0],[z_3,0],[0,z_2^{-2 + n_2}z_5^{2 n_5}],[0,z_2^{-1 + n_2}z_5^{n_5}],[0,z_2^{-2 + n_2}z_4]\}
\end{split}&
\end{flalign*}
\\(115)$\left\{ \begin{array}{l} z_1z_2+z_4z_5=0\\z_1z_4+z_3^{2}+z_2^{n_2}z_4+z_5^{n_5}=0\end{array}\right.$
\\$n_2 \ge 2~and~2\le n_5 \le 2 n_2$
\\$(w_1,w_2,w_3,w_4,w_5;1,d)=(\frac{n_2}{1 + n_2},\frac{1}{1 + n_2},\frac{n_5 + 2 n_2 n_5}{2 (1 + n_2) (1 + n_5)},\frac{-n_2 + n_5 + n_2 n_5}{(1 + n_2) (1 + n_5)},\frac{1 + 2 n_2}{(1 + n_2) (1 + n_5)};1,\frac{n_5 + 2 n_2 n_5}{(1 + n_2) (1 + n_5)})$
\\$\mu=1 + 2 (1 + n_2) n_5$
\begin{flalign*}
\begin{split}
mini=\{[z_5^{-1 + n_5},0],[z_3,0],[0,z_2^{n_2}z_5^{-1 + n_5}],[0,z_2^{2 n_2}z_5^{-2 + n_5}],[0,z_2^{-1 + n_2}z_4]\}
\end{split}&
\end{flalign*}
\\(116)$\left\{ \begin{array}{l} z_1z_2+z_4z_5=0\\z_1z_4+2z_3^{2}+3z_2^{n_2}z_4+4z_1z_5^{n_5}+5z_2^{n_2}z_5^{n_5}=0\end{array}\right.$
\\$n_2 \ge 2~and~1\le n_5 \le n_2$
\\$(w_1,w_2,w_3,w_4,w_5;1,d)=(\frac{n_2}{1 + n_2},\frac{1}{1 + n_2},\frac{n_2 + n_5 + 2 n_2 n_5}{2 (1 + n_2) (1 + n_5)},\frac{n_5}{1 + n_5},\frac{1}{1 + n_5};1,\frac{n_2 + n_5 + 2 n_2 n_5}{(1 + n_2) (1 + n_5)})$
\\$\mu=(1 + 2 n_2) (1 + 2 n_5)$
\begin{flalign*}
\begin{split}
mini=\{[z_5^{n_5},0],[z_3,0],[0,z_2^{-2 + n_2}z_5^{2 n_5}],[0,z_2^{-1 + n_2}z_5^{-1 + 2 n_5}],[0,z_2^{-1 + 2 n_2}z_5^{n_5}],[0,z_2^{2 n_2}z_5^{-1 + n_5}],\\ [0,z_2^{3 n_2}z_5^{-2 + n_5}],[0,z_2^{-1 + n_2}z_4]\}
\end{split}&
\end{flalign*}
\\(117)$\left\{ \begin{array}{l} z_1z_2+z_4^{2}+z_5^{2}=0\\z_1z_4+z_3^{2}+z_2^{n}=0\end{array}\right.$
\\$n\ge3$
\\$(w_1,w_2,w_3,w_4,w_5;1,d)=(\frac{-1 + 2 n}{2 (1 + n)},\frac{3}{2 (1 + n)},\frac{3 n}{4 (1 + n)},\frac{1}{2},\frac{1}{2};1,\frac{3 n}{2 (1 + n)})$
\\$\mu=1 + 4 n$
\begin{flalign*}
\begin{split}
mini=\{[z_4,0],[z_3,0],[0,z_2^{-2 + n}z_5^{2}],[0,z_2^{-1 + n}z_4]\}
\end{split}&
\end{flalign*}
\\(118)$\left\{ \begin{array}{l} z_1z_2+z_4^{2}+z_5^{3}=0\\z_1z_4+z_3^{2}+z_2^{n}=0\end{array}\right.$
\\$n\ge3$
\\$(w_1,w_2,w_3,w_4,w_5;1,d)=(\frac{-1 + 2 n}{2 (1 + n)},\frac{3}{2 (1 + n)},\frac{3 n}{4 (1 + n)},\frac{1}{2},\frac{1}{3};1,\frac{3 n}{2 (1 + n)})$
\\$\mu=3 + 7 n$
\begin{flalign*}
\begin{split}
mini=\{[z_4z_5,0],[z_3z_5,0],[0,z_2^{-2 + n}z_5^{4}],[0,z_2^{-1 + n}z_4z_5]\}
\end{split}&
\end{flalign*}
\\(119)$\left\{ \begin{array}{l} z_1z_2+z_4^{2}+z_5^{4}=0\\z_1z_4+z_3^{2}+z_2^{n}=0\end{array}\right.$
\\$n\ge3$
\\$(w_1,w_2,w_3,w_4,w_5;1,d)=(\frac{-1 + 2 n}{2 (1 + n)},\frac{3}{2 (1 + n)},\frac{3 n}{4 (1 + n)},\frac{1}{2},\frac{1}{4};1,\frac{3 n}{2 (1 + n)})$
\\$\mu=5 + 10 n$
\begin{flalign*}
\begin{split}
mini=\{[z_4z_5^{2},0],[z_3z_5^{2},0],[0,z_2^{-2 + n}z_5^{6}],[0,z_2^{-1 + n}z_4z_5^{2}]\}
\end{split}&
\end{flalign*}
\\(120)$\left\{ \begin{array}{l} z_1z_2+z_4^{2}+z_5^{5}=0\\z_1z_4+z_3^{2}+z_2^{n}=0\end{array}\right.$
\\$3\le n \le13$
\\$(w_1,w_2,w_3,w_4,w_5;1,d)=(\frac{-1 + 2 n}{2 (1 + n)},\frac{3}{2 (1 + n)},\frac{3 n}{4 (1 + n)},\frac{1}{2},\frac{1}{5};1,\frac{3 n}{2 (1 + n)})$
\\$\mu=7 + 13 n$
\begin{flalign*}
\begin{split}
mini=\{[z_4z_5^{3},0],[z_3z_5^{3},0],[0,z_2^{-2 + n}z_5^{8}],[0,z_2^{-1 + n}z_4z_5^{3}]\}
\end{split}&
\end{flalign*}
\\(121)$\left\{ \begin{array}{l} z_1z_2+z_4^{2}+z_5^{6}=0\\z_1z_4+z_3^{2}+z_2^{n}=0\end{array}\right.$
\\$3\le n \le7$
\\$(w_1,w_2,w_3,w_4,w_5;1,d)=(\frac{-1 + 2 n}{2 (1 + n)},\frac{3}{2 (1 + n)},\frac{3 n}{4 (1 + n)},\frac{1}{2},\frac{1}{6};1,\frac{3 n}{2 (1 + n)})$
\\$\mu=9 + 16 n$
\begin{flalign*}
\begin{split}
mini=\{[z_4z_5^{4},0],[z_3z_5^{4},0],[0,z_2^{-2 + n}z_5^{10}],[0,z_2^{-1 + n}z_4z_5^{4}]\}
\end{split}&
\end{flalign*}
\\(122)$\left\{ \begin{array}{l} z_1z_2+z_4^{2}+z_5^{7}=0\\z_1z_4+z_3^{2}+z_2^{n}=0\end{array}\right.$
\\$3\le n \le5$
\\$(w_1,w_2,w_3,w_4,w_5;1,d)=(\frac{-1 + 2 n}{2 (1 + n)},\frac{3}{2 (1 + n)},\frac{3 n}{4 (1 + n)},\frac{1}{2},\frac{1}{7};1,\frac{3 n}{2 (1 + n)})$
\\$\mu=11 + 19 n$
\begin{flalign*}
\begin{split}
mini=\{[z_4z_5^{5},0],[z_3z_5^{5},0],[0,z_2^{-2 + n}z_5^{12}],[0,z_2^{-1 + n}z_4z_5^{5}]\}
\end{split}&
\end{flalign*}
\\(123)$\left\{ \begin{array}{l} z_1z_2+z_4^{2}+z_5^{8}=0\\z_1z_4+z_3^{2}+z_2^{n}=0\end{array}\right.$
\\$3\le n \le4$
\\$(w_1,w_2,w_3,w_4,w_5;1,d)=(\frac{-1 + 2 n}{2 (1 + n)},\frac{3}{2 (1 + n)},\frac{3 n}{4 (1 + n)},\frac{1}{2},\frac{1}{8};1,\frac{3 n}{2 (1 + n)})$
\\$\mu=13 + 22 n$
\begin{flalign*}
\begin{split}
mini=\{[z_4z_5^{6},0],[z_3z_5^{6},0],[0,z_2^{-2 + n}z_5^{14}],[0,z_2^{-1 + n}z_4z_5^{6}]\}
\end{split}&
\end{flalign*}
\\(124)$\left\{ \begin{array}{l} z_1z_2+z_4^{2}+z_5^{9}=0\\z_1z_4+z_3^{2}+z_2^{n}=0\end{array}\right.$
\\$3\le n \le4$
\\$(w_1,w_2,w_3,w_4,w_5;1,d)=(\frac{-1 + 2 n}{2 (1 + n)},\frac{3}{2 (1 + n)},\frac{3 n}{4 (1 + n)},\frac{1}{2},\frac{1}{9};1,\frac{3 n}{2 (1 + n)})$
\\$\mu=5 (3 + 5 n)$
\begin{flalign*}
\begin{split}
mini=\{[z_4z_5^{7},0],[z_3z_5^{7},0],[0,z_2^{-2 + n}z_5^{16}],[0,z_2^{-1 + n}z_4z_5^{7}]\}
\end{split}&
\end{flalign*}
\\(125)$\left\{ \begin{array}{l} z_1z_2+z_4^{2}+z_5^{n}=0\\z_1z_4+z_3^{2}+z_2^{3}=0\end{array}\right.$
\\$10\le n \le15$
\\$(w_1,w_2,w_3,w_4,w_5;1,d)=(\frac{5}{8},\frac{3}{8},\frac{9}{16},\frac{1}{2},\frac{1}{n};1,\frac{9}{8})$
\\$\mu=-9 + 11 n$
\begin{flalign*}
\begin{split}
mini=\{[z_4z_5^{-2 + n},0],[z_3z_5^{-2 + n},0],[0,z_2z_5^{-2 + 2 n}],[0,z_2^{2}z_4z_5^{-2 + n}]\}
\end{split}&
\end{flalign*}
\\(126)$\left\{ \begin{array}{l} z_1z_2+z_4^{2}+z_5^{4}=0\\z_1z_4+z_3^{2}+z_2^{n}z_4+z_2^{n}z_5^{2}=0\end{array}\right.$
\\$n\ge2$
\\$(w_1,w_2,w_3,w_4,w_5;1,d)=(\frac{n}{1 + n},\frac{1}{1 + n},\frac{1 + 3 n}{4 (1 + n)},\frac{1}{2},\frac{1}{4};1,\frac{1 + 3 n}{2 (1 + n)})$
\\$\mu=5 (2 + 3 n)$
\begin{flalign*}
\begin{split}
mini=\{[z_4z_5^{2},0],[z_3z_5^{2},0],[0,z_2^{-2 + n}z_5^{6}],[0,z_2^{-1 + n}z_5^{4}],[0,z_2^{-1 + 2 n}z_5^{3}],[0,z_2^{-1 + n}z_4z_5^{2}],[0,z_2^{2 n}z_5^{2}],[0,z_2^{3 n}]\}
\end{split}&
\end{flalign*}
\\(127)$\left\{ \begin{array}{l} z_1z_2+z_4^{2}+z_5^{6}=0\\z_1z_4+z_3^{2}+z_2^{n}z_4+z_2^{n}z_5^{3}=0\end{array}\right.$
\\$2\le n \le4$
\\$(w_1,w_2,w_3,w_4,w_5;1,d)=(\frac{n}{1 + n},\frac{1}{1 + n},\frac{1 + 3 n}{4 (1 + n)},\frac{1}{2},\frac{1}{6};1,\frac{1 + 3 n}{2 (1 + n)})$
\\$\mu=17 + 24 n$
\begin{flalign*}
\begin{split}
mini=\{[z_4z_5^{4},0],[z_3z_5^{4},0],[0,z_2^{-2 + n}z_5^{10}],[0,z_2^{-1 + n}z_5^{7}],[0,z_2^{-1 + 2 n}z_5^{5}],[0,z_2^{-1 + n}z_4z_5^{4}],[0,z_2^{2 n}z_5^{4}],\\ [0,z_2^{3 n}z_5]\}
\end{split}&
\end{flalign*}
\\(128)$\left\{ \begin{array}{l} z_1z_2+z_4^{2}+z_5^{8}=0\\z_1z_4+z_3^{2}+z_2^{2}z_4+z_2^{2}z_5^{4}=0\end{array}\right.$
\\$(w_1,w_2,w_3,w_4,w_5;1,d)=(\frac{2}{3},\frac{1}{3},\frac{7}{12},\frac{1}{2},\frac{1}{8};1,\frac{7}{6})$
\\$\mu=90$
\begin{flalign*}
\begin{split}
mini=\{[z_5^{14},0],[z_4z_5^{6},0],[z_3z_5^{6},0],[0,z_5^{18}],[0,z_2z_5^{14}],[0,z_4z_5^{10}],[0,z_2z_4z_5^{6}],[0,z_2^{4}z_5^{2}]\}
\end{split}&
\end{flalign*}
\\(129)$\left\{ \begin{array}{l} z_1z_2+z_4^{2}+z_5^{10}=0\\z_1z_4+z_3^{2}+z_2^{2}z_4+z_2^{2}z_5^{5}=0\end{array}\right.$
\\$(w_1,w_2,w_3,w_4,w_5;1,d)=(\frac{2}{3},\frac{1}{3},\frac{7}{12},\frac{1}{2},\frac{1}{10};1,\frac{7}{6})$
\\$\mu=115$
\begin{flalign*}
\begin{split}
mini=\{[z_5^{18},0],[z_4z_5^{8},0],[z_3z_5^{8},0],[0,z_5^{23}],[0,z_2z_5^{18}],[0,z_4z_5^{13}],[0,z_2z_4z_5^{8}],[0,z_2^{4}z_5^{3}]\}
\end{split}&
\end{flalign*}
\\(130)$\left\{ \begin{array}{l} z_1z_2+z_4^{2}+z_5^{n}=0\\z_1z_4+z_3^{2}+z_2^{2}z_5=0\end{array}\right.$
\\$n\ge2$
\\$(w_1,w_2,w_3,w_4,w_5;1,d)=(\frac{2 + 3 n}{6 n},\frac{-2 + 3 n}{6 n},\frac{1 + 3 n}{6 n},\frac{1}{2},\frac{1}{n};1,\frac{1 + 3 n}{3 n})$
\\$\mu=-1 + 8 n$
\begin{flalign*}
\begin{split}
mini=\{[z_4z_5^{-1 + n},0],[z_3z_5^{-2 + n},0],[0,z_5^{-1 + 2 n}],[0,z_2z_4z_5^{-1 + n}]\}
\end{split}&
\end{flalign*}
\\(131)$\left\{ \begin{array}{l} z_1z_2+z_4^{2}+z_5^{n}=0\\z_1z_4+z_3^{2}+z_2^{3}z_5=0\end{array}\right.$
\\$2\le n \le13$
\\$(w_1,w_2,w_3,w_4,w_5;1,d)=(\frac{2 + 5 n}{8 n},\frac{-2 + 3 n}{8 n},\frac{2 + 9 n}{16 n},\frac{1}{2},\frac{1}{n};1,\frac{2 + 9 n}{8 n})$
\\$\mu=-1 + 11 n$
\begin{flalign*}
\begin{split}
mini=\{[z_4z_5^{-1 + n},0],[z_3z_5^{-2 + n},0],[0,z_2z_5^{-1 + 2 n}],[0,z_2^{2}z_4z_5^{-1 + n}]\}
\end{split}&
\end{flalign*}
\\(132)$\left\{ \begin{array}{l} z_1z_2+z_4^{2}+z_5^{n}=0\\z_1z_4+z_3^{2}+z_2^{4}z_5=0\end{array}\right.$
\\$2\le n \le8$
\\$(w_1,w_2,w_3,w_4,w_5;1,d)=(\frac{2 + 7 n}{10 n},\frac{-2 + 3 n}{10 n},\frac{1 + 6 n}{10 n},\frac{1}{2},\frac{1}{n};1,\frac{1 + 6 n}{5 n})$
\\$\mu=-1 + 14 n$
\begin{flalign*}
\begin{split}
mini=\{[z_4z_5^{-1 + n},0],[z_3z_5^{-2 + n},0],[0,z_2^{2}z_5^{-1 + 2 n}],[0,z_2^{3}z_4z_5^{-1 + n}]\}
\end{split}&
\end{flalign*}
\\(133)$\left\{ \begin{array}{l} z_1z_2+z_4^{2}+z_5^{n}=0\\z_1z_4+z_3^{2}+z_2^{5}z_5=0\end{array}\right.$
\\$2\le n \le7$
\\$(w_1,w_2,w_3,w_4,w_5;1,d)=(\frac{2 + 9 n}{12 n},\frac{-2 + 3 n}{12 n},\frac{2 + 15 n}{24 n},\frac{1}{2},\frac{1}{n};1,\frac{2 + 15 n}{12 n})$
\\$\mu=-1 + 17 n$
\begin{flalign*}
\begin{split}
mini=\{[z_4z_5^{-1 + n},0],[z_3z_5^{-2 + n},0],[0,z_2^{3}z_5^{-1 + 2 n}],[0,z_2^{4}z_4z_5^{-1 + n}]\}
\end{split}&
\end{flalign*}
\\(134)$\left\{ \begin{array}{l} z_1z_2+z_4^{2}+z_5^{n}=0\\z_1z_4+z_3^{2}+z_2^{6}z_5=0\end{array}\right.$
\\$2\le n \le6$
\\$(w_1,w_2,w_3,w_4,w_5;1,d)=(\frac{2 + 11 n}{14 n},\frac{-2 + 3 n}{14 n},\frac{1 + 9 n}{14 n},\frac{1}{2},\frac{1}{n};1,\frac{1 + 9 n}{7 n})$
\\$\mu=-1 + 20 n$
\begin{flalign*}
\begin{split}
mini=\{[z_4z_5^{-1 + n},0],[z_3z_5^{-2 + n},0],[0,z_2^{4}z_5^{-1 + 2 n}],[0,z_2^{5}z_4z_5^{-1 + n}]\}
\end{split}&
\end{flalign*}
\\(135)$\left\{ \begin{array}{l} z_1z_2+z_4^{2}+z_5^{n}=0\\z_1z_4+z_3^{2}+z_2^{7}z_5=0\end{array}\right.$
\\$2\le n \le5$
\\$(w_1,w_2,w_3,w_4,w_5;1,d)=(\frac{2 + 13 n}{16 n},\frac{-2 + 3 n}{16 n},\frac{2 + 21 n}{32 n},\frac{1}{2},\frac{1}{n};1,\frac{2 + 21 n}{16 n})$
\\$\mu=-1 + 23 n$
\begin{flalign*}
\begin{split}
mini=\{[z_4z_5^{-1 + n},0],[z_3z_5^{-2 + n},0],[0,z_2^{5}z_5^{-1 + 2 n}],[0,z_2^{6}z_4z_5^{-1 + n}]\}
\end{split}&
\end{flalign*}
\\(136)$\left\{ \begin{array}{l} z_1z_2+z_4^{2}+z_5^{n}=0\\z_1z_4+z_3^{2}+z_2^{8}z_5=0\end{array}\right.$
\\$2\le n \le5$
\\$(w_1,w_2,w_3,w_4,w_5;1,d)=(\frac{2 + 15 n}{18 n},\frac{-2 + 3 n}{18 n},\frac{1 + 12 n}{18 n},\frac{1}{2},\frac{1}{n};1,\frac{1 + 12 n}{9 n})$
\\$\mu=-1 + 26 n$
\begin{flalign*}
\begin{split}
mini=\{[z_5^{4},0],[z_3z_5^{3},0],[0,z_2^{-5 + n}z_4^{2}z_5^{3}],[0,z_2^{-4 + n}z_3z_5^{3}]\}
\end{split}&
\end{flalign*}
\\(137)$\left\{ \begin{array}{l} z_1z_2+z_4^{2}+z_5^{n}=0\\z_1z_4+z_3^{2}+z_2^{9}z_5=0\end{array}\right.$
\\$2\le n \le5$
\\$(w_1,w_2,w_3,w_4,w_5;1,d)=(\frac{2 + 17 n}{20 n},\frac{-2 + 3 n}{20 n},\frac{2 + 27 n}{40 n},\frac{1}{2},\frac{1}{n};1,\frac{2 + 27 n}{20 n})$
\\$\mu=-1 + 29 n$
\begin{flalign*}
\begin{split}
mini=\{[z_5^{-8 + 2 n},0],[z_3z_5^{-4 + n},0],[0,z_2^{3}z_4^{2}z_5^{-4 + n}],[0,z_2^{3}z_3z_5^{-4 + n}],[0,z_2^{4}z_5^{-4 + n}],[0,z_2^{4}z_4],[0,z_2^{4}z_3]\}
\end{split}&
\end{flalign*}
\\(138)$\left\{ \begin{array}{l} z_1z_2+z_4^{2}+z_5^{n}=0\\z_1z_4+z_3^{2}+z_2^{10}z_5=0\end{array}\right.$
\\$2\le n \le5$
\\$(w_1,w_2,w_3,w_4,w_5;1,d)=(\frac{2 + 19 n}{22 n},\frac{-2 + 3 n}{22 n},\frac{1 + 15 n}{22 n},\frac{1}{2},\frac{1}{n};1,\frac{1 + 15 n}{11 n})$
\\$\mu=-1 + 32 n$
\begin{flalign*}
\begin{split}
mini=\{[z_4z_5^{4},0],[z_3z_5^{4},0],[0,z_2^{-6 + n}z_5^{10}],[0,z_2^{-5 + n}z_5^{7}],[0,z_2^{-9 + 2 n}z_5^{5}],[0,z_2^{-5 + n}z_4z_5^{4}],[0,z_2^{-8 + 2 n}z_5^{4}],\\ [0,z_2^{-12 + 3 n}z_5]\}
\end{split}&
\end{flalign*}
\\(139)$\left\{ \begin{array}{l} z_1z_2+z_4^{2}+z_5^{n}=0\\z_1z_4+z_3^{2}+z_2^{11}z_5=0\end{array}\right.$
\\$2\le n \le5$
\\$(w_1,w_2,w_3,w_4,w_5;1,d)=(\frac{2 + 21 n}{24 n},\frac{-2 + 3 n}{24 n},\frac{2 + 33 n}{48 n},\frac{1}{2},\frac{1}{n};1,\frac{2 + 33 n}{24 n})$
\\$\mu=-1 + 35 n$
\begin{flalign*}
\begin{split}
mini=\{[z_4z_5^{-2 + n},0],[z_3z_5^{-2 + n},0],[0,z_2^{9}z_5^{-1 + 2 n}],[0,z_2^{10}z_5^{-1 + n}],[0,z_2^{10}z_4z_5^{-2 + n}],[0,z_2^{11}z_4],[0,z_2^{22}]\}
\end{split}&
\end{flalign*}
\\(140)$\left\{ \begin{array}{l} z_1z_2+z_4^{2}+z_5^{2}=0\\z_1z_4+z_3^{2}+z_2^{n}z_5=0\end{array}\right.$
\\$n\ge12$
\\$(w_1,w_2,w_3,w_4,w_5;1,d)=(\frac{n}{1 + n},\frac{1}{1 + n},\frac{1 + 3 n}{4 (1 + n)},\frac{1}{2},\frac{1}{2};1,\frac{1 + 3 n}{2 (1 + n)})$
\\$\mu=3 + 6 n$
\begin{flalign*}
\begin{split}
mini=\{[z_4,0],[z_3,0],[0,z_2^{-2 + n}z_5^{2}],[0,z_2^{-1 + n}z_5],[0,z_2^{-1 + 2 n}z_4],[0,z_2^{2 n}]\}
\end{split}&
\end{flalign*}
\\(141)$\left\{ \begin{array}{l} z_1z_2+z_4^{2}+z_5^{3}=0\\z_1z_4+z_3^{2}+z_2^{n}z_5=0\end{array}\right.$
\\$n\ge12$
\\$(w_1,w_2,w_3,w_4,w_5;1,d)=(\frac{-1 + 6 n}{6 (1 + n)},\frac{7}{6 (1 + n)},\frac{2 + 9 n}{12 (1 + n)},\frac{1}{2},\frac{1}{3};1,\frac{2 + 9 n}{6 (1 + n)})$
\\$\mu=5 + 9 n$
\begin{flalign*}
\begin{split}
mini=\{[z_4z_5,0],[z_3z_5,0],[0,z_2^{-2 + n}z_5^{4}],[0,z_2^{-1 + n}z_5^{2}],[0,z_2^{-1 + n}z_4z_5],[0,z_2^{-1 + 2 n}z_4],[0,z_2^{2 n}]\}
\end{split}&
\end{flalign*}
\\(142)$\left\{ \begin{array}{l} z_1z_2+z_4^{2}+z_5^{4}=0\\z_1z_4+z_3^{2}+z_2^{n}z_5=0\end{array}\right.$
\\$n\ge12$
\\$(w_1,w_2,w_3,w_4,w_5;1,d)=(\frac{-1 + 4 n}{4 (1 + n)},\frac{5}{4 (1 + n)},\frac{1 + 6 n}{8 (1 + n)},\frac{1}{2},\frac{1}{4};1,\frac{1 + 6 n}{4 (1 + n)})$
\\$\mu=7 + 12 n$
\begin{flalign*}
\begin{split}
mini=\{[z_4z_5^{2},0],[z_3z_5^{2},0],[0,z_2^{-2 + n}z_5^{6}],[0,z_2^{-1 + n}z_5^{3}],[0,z_2^{-1 + n}z_4z_5^{2}],[0,z_2^{-1 + 2 n}z_4],[0,z_2^{2 n}]\}
\end{split}&
\end{flalign*}
\\(143)$\left\{ \begin{array}{l} z_1z_2+z_4^{2}+z_5^{n_5}=0\\z_1z_5+z_3^{2}+z_2^{n_2}=0\end{array}\right.$
\\$n_2 \ge 3~and~3\le n_5 \le n_2$
\\$(w_1,w_2,w_3,w_4,w_5;1,d)=(\frac{-1 + n_2 n_5}{(1 + n_2) n_5},\frac{1 + n_5}{(1 + n_2) n_5},\frac{n_2 + n_2 n_5}{2 (1 + n_2) n_5},\frac{1}{2},\frac{1}{n_5};1,\frac{n_2 + n_2 n_5}{(1 + n_2) n_5})$
\\$\mu=-3 + 2 n_5 + n_2 (2 + n_5)$
\begin{flalign*}
\begin{split}
mini=\{[z_5^{-1 + n_5},0],[z_3z_5^{-2 + n_5},0],[0,z_2^{-2 + n_2}z_5^{n_5}],[0,z_2^{-1 + n_2}z_5^{-1 + n_5}],[0,z_2^{-2 + n_2}z_4]\}
\end{split}&
\end{flalign*}
\\(144)$\left\{ \begin{array}{l} z_1z_2+z_4^{2}+z_5^{n_5}=0\\z_1z_5+z_3^{2}+z_2^{n_2}z_4=0\end{array}\right.$
\\$n_2 \ge 2~and~3\le n_5 \le 2 n_2$
\\$(w_1,w_2,w_3,w_4,w_5;1,d)=(\frac{-2 + n_5 + 2 n_2 n_5}{2 (1 + n_2) n_5},\frac{2 + n_5}{2 (1 + n_2) n_5},\frac{2 n_2 + n_5 + 2 n_2 n_5}{4 (1 + n_2) n_5},\frac{1}{2},\frac{1}{n_5};1,\frac{2 n_2 + n_5 + 2 n_2 n_5}{2 (1 + n_2) n_5})$
\\$\mu=3 (-1 + n_5) + 2 n_2 (1 + n_5)$
\begin{flalign*}
\begin{split}
mini=\{[z_5^{-1 + n_5},0],[z_3z_5^{-2 + n_5},0],[0,z_2^{-2 + n_2}z_5^{n_5}],[0,z_2^{-1 + 2 n_2}z_5^{-1 + n_5}],[0,z_2^{2 n_2}z_5^{-2 + n_5}],[0,z_2^{-1 + n_2}z_4]\}
\end{split}&
\end{flalign*}
\\(145)$\left\{ \begin{array}{l} z_1z_2+z_4^{2}+z_3z_5^{2}=0\\z_1z_4+z_3^{2}+z_2^{6}+z_2z_5^{6}=0\end{array}\right.$
\\$(w_1,w_2,w_3,w_4,w_5;1,d)=(\frac{11}{14},\frac{3}{14},\frac{9}{14},\frac{1}{2},\frac{5}{28};1,\frac{9}{7})$
\\$\mu=97$
\begin{flalign*}
\begin{split}
mini=\{[z_5^{6},0],[z_4z_5^{2},0],[z_3,0],[0,z_5^{14}],[0,z_2^{2}z_5^{10}],[0,z_2^{4}z_5^{8}],[0,z_4z_5^{6}],[0,z_2^{3}z_4z_5^{4}],[0,z_2^{5}z_5^{4}],\\ [0,z_2^{4}z_4z_5^{2}],[0,z_2^{5}z_4]\}
\end{split}&
\end{flalign*}
\\(146)$\left\{ \begin{array}{l} z_1z_2+z_4^{2}+z_3z_5^{3}=0\\z_1z_4+z_3^{2}+z_2^{4}+z_5^{9}=0\end{array}\right.$
\\$(w_1,w_2,w_3,w_4,w_5;1,d)=(\frac{7}{10},\frac{3}{10},\frac{3}{5},\frac{1}{2},\frac{2}{15};1,\frac{6}{5})$
\\$\mu=94$
\begin{flalign*}
\begin{split}
mini=\{[z_5^{16},0],[z_4z_5^{4},0],[z_3z_5,0],[0,z_5^{16}],[0,z_2z_5^{13}],[0,z_4z_5^{10}],[0,z_2z_4z_5^{7}],[0,z_2^{2}z_5^{7}],[0,z_2^{2}z_4z_5^{4}],\\ [0,z_2^{3}z_5^{4}],[0,z_2^{3}z_4z_5]\}
\end{split}&
\end{flalign*}
\\(147)$\left\{ \begin{array}{l} z_1z_2+z_4^{2}+z_3z_5^{2}=0\\z_1z_4+z_3^{2}+z_2^{2}z_4+z_2z_5^{4}=0\end{array}\right.$
\\$(w_1,w_2,w_3,w_4,w_5;1,d)=(\frac{2}{3},\frac{1}{3},\frac{7}{12},\frac{1}{2},\frac{5}{24};1,\frac{7}{6})$
\\$\mu=50$
\begin{flalign*}
\begin{split}
mini=\{[z_5^{6},0],[z_4z_5^{2},0],[z_3,0],[0,z_5^{10}],[0,z_4z_5^{6}],[0,z_2^{2}z_5^{6}],[0,z_2z_4z_5^{2}],[0,z_2^{3}z_5^{2}],[0,z_2^{4}]\}
\end{split}&
\end{flalign*}
\\(148)$\left\{ \begin{array}{l} z_1z_2+z_4^{2}+z_3z_5^{3}=0\\z_1z_4+z_3^{2}+z_2^{2}z_4+z_2z_5^{6}=0\end{array}\right.$
\\$(w_1,w_2,w_3,w_4,w_5;1,d)=(\frac{2}{3},\frac{1}{3},\frac{7}{12},\frac{1}{2},\frac{5}{36};1,\frac{7}{6})$
\\$\mu=80$
\begin{flalign*}
\begin{split}
mini=\{[z_5^{10},0],[z_4z_5^{4},0],[z_3z_5,0],[0,z_5^{16}],[0,z_4z_5^{10}],[0,z_2^{2}z_5^{10}],[0,z_2z_4z_5^{4}],[0,z_2^{3}z_5^{4}],[0,z_2^{4}z_5]\}
\end{split}&
\end{flalign*}
\\(149)$\left\{ \begin{array}{l} z_1z_2+z_4^{2}+z_3z_5^{4}=0\\z_1z_4+z_3^{2}+z_2^{2}z_4+z_2z_5^{8}=0\end{array}\right.$
\\$(w_1,w_2,w_3,w_4,w_5;1,d)=(\frac{2}{3},\frac{1}{3},\frac{7}{12},\frac{1}{2},\frac{5}{48};1,\frac{7}{6})$
\\$\mu=110$
\begin{flalign*}
\begin{split}
mini=\{[z_5^{14},0],[z_4z_5^{6},0],[z_3z_5^{2},0],[0,z_5^{22}],[0,z_4z_5^{14}],[0,z_2^{2}z_5^{14}],[0,z_2z_4z_5^{6}],[0,z_2^{3}z_5^{6}],[0,z_2^{4}z_5^{2}]\}
\end{split}&
\end{flalign*}
\\(150)$\left\{ \begin{array}{l} z_1z_2+z_4^{2}+z_3z_5^{n_5}=0\\z_1z_5+z_3^{2}+z_2^{n_2}=0\end{array}\right.$
\\$n_2 \ge 4~and~4\le 2 n_5 \le n_2$
\\$(w_1,w_2,w_3,w_4,w_5;1,d)=(\frac{-2 + n_2 + 2 n_2 n_5}{n_2 + 2 n_5 + 2 n_2 n_5},\frac{2 (1 + n_5)}{n_2 + 2 n_5 + 2 n_2 n_5},\frac{n_2 (1 + n_5)}{n_2 + 2 n_5 + 2 n_2 n_5},\frac{1}{2},\frac{2 + n_2}{n_2 + 2 n_5 + 2 n_2 n_5};1,\frac{2 n_2 (1 + n_5)}{n_2 + 2 n_5 + 2 n_2 n_5})$
\\$\mu=-3 + 2 n_5 + n_2 (3 + 2 n_5)$
\begin{flalign*}
\begin{split}
mini=\{[z_5^{2 n_5},0],[z_3z_5^{-1 + n_5},0],[0,z_2^{-3 + n_2}z_5^{1 + 2 n_5}],[0,z_2^{-2 + n_2}z_5^{2 n_5}],[0,z_2^{-1 + n_2}z_5^{-1 + n_5}],[0,z_2^{-2 + n_2}z_4]\}
\end{split}&
\end{flalign*}
\\(151)$\left\{ \begin{array}{l} z_1z_2+z_4^{2}+z_3z_5^{n_5}=0\\z_1z_5+z_3^{2}+z_2^{n_2}z_4=0\end{array}\right.$
\\$n_2 \ge 2~and~2\le n_5 \le n_2$
\\$(w_1,w_2,w_3,w_4,w_5;1,d)=(\frac{-3 + 2 n_2 + 2 n_5 + 4 n_2 n_5}{2 (n_2 + 2 n_5 + 2 n_2 n_5)},\frac{3 + 2 n_5}{2 (n_2 + 2 n_5 + 2 n_2 n_5)},\frac{2 n_2 + n_5 + 2 n_2 n_5}{2 (n_2 + 2 n_5 + 2 n_2 n_5)},\frac{1}{2},\frac{3 + 2 n_2}{2 (n_2 + 2 n_5 + 2 n_2 n_5)};\\ 1,\frac{2 n_2 + n_5 + 2 n_2 n_5}{n_2 + 2 n_5 + 2 n_2 n_5})$
\\$\mu=-3 + 4 n_5 + 4 n_2 (1 + n_5)$
\begin{flalign*}
\begin{split}
mini=\{[z_5^{2 n_5},0],[z_3z_5^{-1 + n_5},0],[0,z_2^{-3 + n_2}z_5^{1 + 2 n_5}],[0,z_2^{-2 + 2 n_2}z_5^{2 n_5}],[0,z_2^{-1 + 2 n_2}z_5^{-1 + 2 n_5}],\\ [0,z_2^{2 n_2}z_5^{-2 + n_5}],[0,z_2^{-1 + n_2}z_4]\}
\end{split}&
\end{flalign*}
\\(152)$\left\{
\begin{array}{l} z_1z_2+z_4^{3}+z_5^{3}=0\\z_1z_4+z_3^{2}+z_2^{n}=0\end{array}\right.$
\\$n\ge3$
\\$(w_1,w_2,w_3,w_4,w_5;1,d)=(\frac{-1 + 3 n}{3 (1 + n)},\frac{4}{3 (1 + n)},\frac{2 n}{3 (1 + n)},\frac{1}{3},\frac{1}{3};1,\frac{4 n}{3 (1 + n)})$
\\$\mu=7 + 9 n$
\begin{flalign*}
\begin{split}
mini=\{[z_4^{2}z_5,0],[z_3z_4z_5,0],[0,z_2^{-2 + n}z_5^{4}],[0,z_2^{-1 + n}z_4^{2}z_5]\}
\end{split}&
\end{flalign*}
\\(153)$\left\{ \begin{array}{l} z_1z_2+z_4^{3}+z_5^{4}=0\\z_1z_4+z_3^{2}+z_2^{n}=0\end{array}\right.$
\\$3\le n \le6$
\\$(w_1,w_2,w_3,w_4,w_5;1,d)=(\frac{-1 + 3 n}{3 (1 + n)},\frac{4}{3 (1 + n)},\frac{2 n}{3 (1 + n)},\frac{1}{3},\frac{1}{4};1,\frac{4 n}{3 (1 + n)})$
\\$\mu=11 + 13 n$
\begin{flalign*}
\begin{split}
mini=\{[z_4^{2}z_5^{2},0],[z_3z_4z_5^{2},0],[0,z_2^{-2 + n}z_5^{6}],[0,z_2^{-1 + n}z_4^{2}z_5^{2}]\}
\end{split}&
\end{flalign*}
\\(154)$\left\{ \begin{array}{l} z_1z_2+z_4^{3}+z_5^{3}=0\\z_1z_4+z_3^{2}+z_2^{n}z_5=0\end{array}\right.$
\\$n\ge2$
\\$(w_1,w_2,w_3,w_4,w_5;1,d)=(\frac{n}{1 + n},\frac{1}{1 + n},\frac{1 + 4 n}{6 (1 + n)},\frac{1}{3},\frac{1}{3};1,\frac{1 + 4 n}{3 (1 + n)})$
\\$\mu=2 (5 + 6 n)$
\begin{flalign*}
\begin{split}
mini=\{[z_4^{2}z_5,0],[z_3z_4z_5,0],[0,z_2^{-2 + n}z_5^{4}],[0,z_2^{-1 + n}z_5^{2}],[0,z_2^{-1 + n}z_4^{2}z_5],[0,z_2^{-1 + 2 n}z_4^{2}],[0,z_2^{2 n}z_4]\}
\end{split}&
\end{flalign*}
\\(155)$\left\{ \begin{array}{l} z_1z_2+z_4^{3}+z_5^{4}=0\\z_1z_4+z_3^{2}+z_2^{n}z_5=0\end{array}\right.$
\\$3\le n \le5$
\\$(w_1,w_2,w_3,w_4,w_5;1,d)=(\frac{-1 + 12 n}{12 (1 + n)},\frac{13}{12 (1 + n)},\frac{3 + 16 n}{24 (1 + n)},\frac{1}{3},\frac{1}{4};1,\frac{3 + 16 n}{12 (1 + n)})$
\\$\mu=2 (7 + 8 n)$
\begin{flalign*}
\begin{split}
mini=\{[z_4^{2}z_5^{2},0],[z_3z_4z_5^{2},0],[0,z_2^{-2 + n}z_5^{6}],[0,z_2^{-1 + n}z_5^{3}],[0,z_2^{-1 + n}z_4^{2}z_5^{2}],[0,z_2^{-1 + 2 n}z_4^{2}],[0,z_2^{2 n}z_4]\}
\end{split}&
\end{flalign*}
\\(156)$\left\{ \begin{array}{l} z_1z_2+z_4^{3}+z_5^{5}=0\\z_1z_4+z_3^{2}+z_2^{3}z_5=0\end{array}\right.$
\\$(w_1,w_2,w_3,w_4,w_5;1,d)=(\frac{43}{60},\frac{17}{60},\frac{21}{40},\frac{1}{3},\frac{1}{5};1,\frac{21}{20})$
\\$\mu=78$
\begin{flalign*}
\begin{split}
mini=\{[z_4^{2}z_5^{4},0],[z_3z_4z_5^{3},0],[0,z_2z_5^{9}],[0,z_2^{2}z_4^{2}z_5^{4}]\}
\end{split}&
\end{flalign*}
\\(157)$\left\{ \begin{array}{l} z_1z_2+z_4^{3}+z_5^{4}=0\\z_1z_5+z_3^{2}+z_2^{n}=0\end{array}\right.$
\\$4\le n \le13$
\\$(w_1,w_2,w_3,w_4,w_5;1,d)=(\frac{-1 + 4 n}{4 (1 + n)},\frac{5}{4 (1 + n)},\frac{5 n}{8 (1 + n)},\frac{1}{3},\frac{1}{4};1,\frac{5 n}{4 (1 + n)})$
\\$\mu=11 (1 + n)$
\begin{flalign*}
\begin{split}
mini=\{[z_4z_5^{3},0],[z_3z_4z_5^{2},0],[0,z_2^{-2 + n}z_4z_5^{4}],[0,z_2^{-1 + n}z_4z_5^{3}],[0,z_2^{-2 + n}z_4^{2}]\}
\end{split}&
\end{flalign*}
\\(158)$\left\{ \begin{array}{l} z_1z_2+z_4^{3}+z_5^{5}=0\\z_1z_5+z_3^{2}+z_2^{n}=0\end{array}\right.$
\\$5\le n \le7$
\\$(w_1,w_2,w_3,w_4,w_5;1,d)=(\frac{-1 + 5 n}{5 (1 + n)},\frac{6}{5 (1 + n)},\frac{3 n}{5 (1 + n)},\frac{1}{3},\frac{1}{5};1,\frac{6 n}{5 (1 + n)})$
\\$\mu=15 + 13 n$
\begin{flalign*}
\begin{split}
mini=\{[z_4z_5^{4},0],[z_3z_4z_5^{3},0],[0,z_2^{-2 + n}z_4z_5^{5}],[0,z_2^{-1 + n}z_4z_5^{4}],[0,z_2^{-2 + n}z_4^{2}]\}
\end{split}&
\end{flalign*}
\\(159)$\left\{ \begin{array}{l} z_1z_2+z_4^{3}+z_5^{4}=0\\z_1z_5+z_3^{2}+z_2^{n}z_4=0\end{array}\right.$
\\$3\le n \le9$
\\$(w_1,w_2,w_3,w_4,w_5;1,d)=(\frac{1 + 12 n}{12 (1 + n)},\frac{11}{12 (1 + n)},\frac{4 + 15 n}{24 (1 + n)},\frac{1}{3},\frac{1}{4};1,\frac{4 + 15 n}{12 (1 + n)})$
\\$\mu=15 (1 + n)$
\begin{flalign*}
\begin{split}
mini=\{[z_4z_5^{3},0],[z_3z_4z_5^{2},0],[0,z_2^{-2 + n}z_4z_5^{4}],[0,z_2^{-1 + n}z_4z_5^{3}],[0,z_2^{-1 + 2 n}z_5^{3}],[0,z_2^{2 n}z_5^{2}],[0,z_2^{-1 + n}z_4^{2}]\}
\end{split}&
\end{flalign*}
\\(160)$\left\{ \begin{array}{l} z_1z_2+z_4^{3}+z_5^{5}=0\\z_1z_5+z_3^{2}+z_2^{4}z_4=0\end{array}\right.$
\\$(w_1,w_2,w_3,w_4,w_5;1,d)=(\frac{62}{75},\frac{13}{75},\frac{77}{150},\frac{1}{3},\frac{1}{5};1,\frac{77}{75})$
\\$\mu=92$
\begin{flalign*}
\begin{split}
mini=\{[z_4z_5^{4},0],[z_3z_4z_5^{3},0],[0,z_2^{2}z_4z_5^{5}],[0,z_2^{3}z_4z_5^{4}],[0,z_2^{7}z_5^{4}],[0,z_2^{8}z_5^{3}],[0,z_2^{3}z_4^{2}]\}
\end{split}&
\end{flalign*}
\\(161)$\left\{ \begin{array}{l} z_1z_2+z_4^{3}+z_5^{5}=0\\z_1z_5+z_3^{2}+z_2^{5}z_4=0\end{array}\right.$
\\$(w_1,w_2,w_3,w_4,w_5;1,d)=(\frac{77}{90},\frac{13}{90},\frac{19}{36},\frac{1}{3},\frac{1}{5};1,\frac{19}{18})$
\\$\mu=110$
\begin{flalign*}
\begin{split}
mini=\{[z_4z_5^{4},0],[z_3z_4z_5^{3},0],[0,z_2^{3}z_4z_5^{5}],[0,z_2^{9}z_5^{4}],[0,z_2^{4}z_4z_5^{3}],[0,z_2^{10}z_5^{3}],[0,z_2^{4}z_4^{2}]\}
\end{split}&
\end{flalign*}
\\(162)$\left\{ \begin{array}{l} z_1z_2+z_4^{3}+z_4z_5^{3}=0\\z_1z_5+z_3^{2}+z_2^{n}=0\end{array}\right.$
\\$5\le n \le9$
\\$(w_1,w_2,w_3,w_4,w_5;1,d)=(\frac{-2 + 9 n}{9 (1 + n)},\frac{11}{9 (1 + n)},\frac{11 n}{18 (1 + n)},\frac{1}{3},\frac{2}{9};1,\frac{11 n}{9 (1 + n)})$
\\$\mu=13 + 12 n$
\begin{flalign*}
\begin{split}
mini=\{[z_5^{5},0],[z_3z_5^{4},0],[z_4z_5^{2},0],[z_3z_4z_5,0],[0,z_2^{-2 + n}z_5^{6}],[0,z_2^{-1 + n}z_5^{5}],[0,z_2^{-2 + n}z_4z_5^{3}],\\ [0,z_2^{-1 + n}z_4z_5^{2}],[0,z_2^{-2 + n}z_4^{2}]\}
\end{split}&
\end{flalign*}
\\(163)$\left\{ \begin{array}{l} z_1z_2+z_4^{3}+z_3z_5^{2}=0\\z_1z_5+z_3^{2}+z_2^{n}=0\end{array}\right.$
\\$5\le n \le9$
\\$(w_1,w_2,w_3,w_4,w_5;1,d)=(\frac{-2 + 5 n}{4 + 5 n},\frac{6}{4 + 5 n},\frac{3 n}{4 + 5 n},\frac{1}{3},\frac{2 + n}{4 + 5 n};1,\frac{6 n}{4 + 5 n})$
\\$\mu=3 + 13 n$
\begin{flalign*}
\begin{split}
mini=\{[z_4z_5^{4},0],[z_3z_4z_5,0],[0,z_2^{-3 + n}z_4z_5^{5}],[0,z_2^{-2 + n}z_4z_5^{4}],[0,z_2^{-1 + n}z_4z_5],[0,z_2^{-2 + n}z_4^{2}]\}
\end{split}&
\end{flalign*}
\\(164)$\left\{ \begin{array}{l} z_1z_2+z_4^{3}+z_3z_5^{2}=0\\z_1z_5+z_3^{2}+z_2^{n}z_4=0\end{array}\right.$
\\$3\le n \le6$
\\$(w_1,w_2,w_3,w_4,w_5;1,d)=(\frac{-1 + 15 n}{3 (4 + 5 n)},\frac{13}{3 (4 + 5 n)},\frac{2 + 9 n}{3 (4 + 5 n)},\frac{1}{3},\frac{5 + 3 n}{3 (4 + 5 n)};1,\frac{2 (2 + 9 n)}{3 (4 + 5 n)})$
\\$\mu=7 + 18 n$
\begin{flalign*}
\begin{split}
mini=\{[z_4z_5^{4},0],[z_3z_4z_5,0],[0,z_2^{-3 + n}z_4z_5^{5}],[0,z_2^{-2 + n}z_4z_5^{4}],[0,z_2^{-2 + 2 n}z_5^{4}],[0,z_2^{-1 + 2 n}z_5^{3}],\\ [0,z_2^{-1 + n}z_4z_5],[0,z_2^{-1 + n}z_4^{2}],[0,z_2^{2 n}]\}
\end{split}&
\end{flalign*}
\\(165)$\left\{ \begin{array}{l} z_1z_2+z_4^{2}z_5+z_5^{4}=0\\z_1z_4+z_3^{2}+z_2^{n}=0\end{array}\right.$
\\$3\le n \le9$
\\$(w_1,w_2,w_3,w_4,w_5;1,d)=(\frac{-3 + 8 n}{8 (1 + n)},\frac{11}{8 (1 + n)},\frac{11 n}{16 (1 + n)},\frac{3}{8},\frac{1}{4};1,\frac{11 n}{8 (1 + n)})$
\\$\mu=9 + 12 n$
\begin{flalign*}
\begin{split}
mini=\{[z_4z_5^{3},0],[z_3z_5^{3},0],[z_3z_4,0],[0,z_2^{-2 + n}z_5^{6}],[0,z_2^{-1 + n}z_4z_5^{3}],[0,z_2^{-2 + n}z_4^{2}]\}
\end{split}&
\end{flalign*}
\\(166)$\left\{ \begin{array}{l} z_1z_2+z_4^{2}z_5+z_5^{5}=0\\z_1z_4+z_3^{2}+z_2^{n}=0\end{array}\right.$
\\$3\le n \le5$
\\$(w_1,w_2,w_3,w_4,w_5;1,d)=(\frac{-2 + 5 n}{5 (1 + n)},\frac{7}{5 (1 + n)},\frac{7 n}{10 (1 + n)},\frac{2}{5},\frac{1}{5};1,\frac{7 n}{5 (1 + n)})$
\\$\mu=11 + 15 n$
\begin{flalign*}
\begin{split}
mini=\{[z_4z_5^{4},0],[z_3z_5^{4},0],[z_3z_4,0],[0,z_2^{-2 + n}z_5^{8}],[0,z_2^{-1 + n}z_4z_5^{4}],[0,z_2^{-2 + n}z_4^{2}]\}
\end{split}&
\end{flalign*}
\\(167)$\left\{ \begin{array}{l} z_1z_2+z_4^{2}z_5+z_5^{6}=0\\z_1z_4+z_3^{2}+z_2^{n}=0\end{array}\right.$
\\$3\le n \le4$
\\$(w_1,w_2,w_3,w_4,w_5;1,d)=(\frac{-5 + 12 n}{12 (1 + n)},\frac{17}{12 (1 + n)},\frac{17 n}{24 (1 + n)},\frac{5}{12},\frac{1}{6};1,\frac{17 n}{12 (1 + n)})$
\\$\mu=13 + 18 n$
\begin{flalign*}
\begin{split}
mini=\{[z_4z_5^{5},0],[z_3z_5^{5},0],[z_3z_4,0],[0,z_2^{-2 + n}z_5^{10}],[0,z_2^{-1 + n}z_4z_5^{5}],[0,z_2^{-2 + n}z_4^{2}]\}
\end{split}&
\end{flalign*}
\\(168)$\left\{ \begin{array}{l} z_1z_2+z_4^{2}z_5+z_5^{n}=0\\z_1z_4+z_3^{2}+z_2^{3}=0\end{array}\right.$
\\$7\le n \le10$
\\$(w_1,w_2,w_3,w_4,w_5;1,d)=(\frac{1 + 5 n}{8 n},\frac{-1 + 3 n}{8 n},\frac{3 (-1 + 3 n)}{16 n},\frac{-1 + n}{2 n},\frac{1}{n};1,\frac{3 (-1 + 3 n)}{8 n})$
\\$\mu=1 + 11 n$
\begin{flalign*}
\begin{split}
mini=\{[z_4z_5^{-1 + n},0],[z_3z_5^{-1 + n},0],[z_3z_4,0],[0,z_2z_5^{-2 + 2 n}],[0,z_2^{2}z_4z_5^{-1 + n}],[0,z_2z_4^{2}]\}
\end{split}&
\end{flalign*}
\\(169)$\left\{ \begin{array}{l} z_1z_2+z_4^{2}z_5+z_5^{5}=0\\z_1z_4+z_3^{2}+z_2^{2}z_4+z_2^{2}z_5^{2}=0\end{array}\right.$
\\$(w_1,w_2,w_3,w_4,w_5;1,d)=(\frac{2}{3},\frac{1}{3},\frac{8}{15},\frac{2}{5},\frac{1}{5};1,\frac{16}{15})$
\\$\mu=59$
\begin{flalign*}
\begin{split}
mini=\{[z_5^{8},0],[z_4z_5^{4},0],[z_3z_5^{4},0],[z_3z_4,0],[0,z_5^{10}],[0,z_2z_5^{8}],[0,z_4z_5^{6}],[0,z_2z_4z_5^{4}],[0,z_4^{2}z_5^{2}],\\ [0,z_2z_4^{2}z_5],[0,z_2^{2}z_5]\}
\end{split}&
\end{flalign*}
\\(170)$\left\{ \begin{array}{l} z_1z_2+z_4^{2}z_5+z_5^{5}=0\\z_1z_4+z_3^{2}+z_2^{3}z_4+z_2^{3}z_5^{2}=0\end{array}\right.$
\\$(w_1,w_2,w_3,w_4,w_5;1,d)=(\frac{3}{4},\frac{1}{4},\frac{23}{40},\frac{2}{5},\frac{1}{5};1,\frac{23}{20})$
\\$\mu=80$
\begin{flalign*}
\begin{split}
mini=\{[z_5^{8},0],[z_4z_5^{4},0],[z_3z_5^{4},0],[z_3z_4,0],[0,z_2z_5^{10}],[0,z_2^{2}z_5^{8}],[0,z_2z_4z_5^{6}],[0,z_2^{2}z_4z_5^{4}],[0,z_2^{5}z_5],\\ [0,z_2^{2}z_4^{2}],[0,z_2^{6}]\}
\end{split}&
\end{flalign*}
\\(171)$\left\{ \begin{array}{l} z_1z_2+z_4^{2}z_5+z_5^{7}=0\\z_1z_4+z_3^{2}+z_2^{2}z_4+z_2^{2}z_5^{3}=0\end{array}\right.$
\\$(w_1,w_2,w_3,w_4,w_5;1,d)=(\frac{2}{3},\frac{1}{3},\frac{23}{42},\frac{3}{7},\frac{1}{7};1,\frac{23}{21})$
\\$\mu=84$
\begin{flalign*}
\begin{split}
mini=\{[z_5^{12},0],[z_4z_5^{6},0],[z_3z_5^{6},0],[z_3z_4,0],[0,z_5^{15}],[0,z_2z_5^{12}],[0,z_4z_5^{9}],[0,z_2z_4z_5^{6}],[0,z_4^{2}z_5^{3}],\\ [0,z_2z_4^{2}z_5^{2}],[0,z_2^{2}z_5^{2}]\}
\end{split}&
\end{flalign*}
\\(172)$\left\{ \begin{array}{l} z_1z_2+z_4^{2}z_5+z_5^{5}=0\\z_1z_4+z_3^{2}+z_2^{3}z_5+z_2z_4^{2}=0\end{array}\right.$
\\$(w_1,w_2,w_3,w_4,w_5;1,d)=(\frac{7}{10},\frac{3}{10},\frac{11}{20},\frac{2}{5},\frac{1}{5};1,\frac{11}{10})$
\\$\mu=66$
\begin{flalign*}
\begin{split}
mini=\{[z_5^{10},0],[z_4z_5^{4},0],[z_3z_5^{4},0],[z_3z_4,0],[0,z_5^{10}],[0,z_4z_5^{7}],[0,z_2z_5^{7}],[0,z_4^{2}z_5^{4}],[0,z_2z_4z_5^{4}],\\ [0,z_2^{2}z_5^{4}],[0,z_2^{4}]\}
\end{split}&
\end{flalign*}
\\(173)$\left\{ \begin{array}{l} z_1z_2+z_4^{2}z_5+z_5^{4}=0\\z_1z_5+z_3^{2}+z_2^{9}+z_4^{3}=0\end{array}\right.$
\\$(w_1,w_2,w_3,w_4,w_5;1,d)=(\frac{7}{8},\frac{1}{8},\frac{9}{16},\frac{3}{8},\frac{1}{4};1,\frac{9}{8})$
\\$\mu=93$
\begin{flalign*}
\begin{split}
mini=\{[z_3z_5^{3},0],[z_4z_5,0],[z_4^{3},0],[z_3z_4,0],[0,z_2^{8}z_5^{4}],[0,z_2^{7}z_4z_5^{2}],[0,z_2^{8}z_4z_5],[0,z_2^{17}]\}
\end{split}&
\end{flalign*}
\\(174)$\left\{ \begin{array}{l} z_1z_2+z_4^{2}z_5+z_3z_5^{2}=0\\z_1z_4+z_3^{2}+z_2^{4}+z_5^{5}=0\end{array}\right.$
\\$(w_1,w_2,w_3,w_4,w_5;1,d)=(\frac{13}{18},\frac{5}{18},\frac{5}{9},\frac{7}{18},\frac{2}{9};1,\frac{10}{9})$
\\$\mu=64$
\begin{flalign*}
\begin{split}
mini=\{[z_5^{9},0],[z_4z_5^{3},0],[z_3z_5,0],[z_3z_4,0],[0,z_5^{9}],[0,z_2z_5^{7}],[0,z_4z_5^{6}],[0,z_2z_4z_5^{4}],[0,z_2^{2}z_5^{4}],\\ [0,z_2^{2}z_4z_5^{3}],[0,z_2^{3}z_5^{2}],[0,z_2^{3}z_4z_5],[0,z_2^{2}z_4^{2}]\}
\end{split}&
\end{flalign*}
\\(175)$\left\{ \begin{array}{l} z_1z_2+z_4^{2}z_5+z_3z_5^{2}=0\\z_1z_4+z_3^{2}+z_2^{2}z_4+z_2z_5^{3}=0\end{array}\right.$
\\$(w_1,w_2,w_3,w_4,w_5;1,d)=(\frac{2}{3},\frac{1}{3},\frac{11}{21},\frac{8}{21},\frac{5}{21};1,\frac{22}{21})$
\\$\mu=49$
\begin{flalign*}
\begin{split}
mini=\{[z_4z_5^{6},0],[z_3z_5,0],[z_3z_4,0],[0,z_5^{8}],[0,z_4z_5^{5}],[0,z_2z_5^{5}],[0,z_4^{2}z_5^{2}],[0,z_2z_4z_5^{2}],[0,z_2^{2}z_5^{2}],\\ [0,z_2z_4^{2}z_5]\}
\end{split}&
\end{flalign*}
\\(176)$\left\{ \begin{array}{l} z_1z_2+z_4^{2}z_5+z_3z_5^{3}=0\\z_1z_4+z_3^{2}+z_2^{2}z_4+z_2z_5^{5}=0\end{array}\right.$
\\$(w_1,w_2,w_3,w_4,w_5;1,d)=(\frac{2}{3},\frac{1}{3},\frac{6}{11},\frac{14}{33},\frac{5}{33};1,\frac{12}{11})$
\\$\mu=79$
\begin{flalign*}
\begin{split}
mini=\{[z_4z_5^{10},0],[z_3z_5^{2},0],[z_3z_4,0],[0,z_5^{14}],[0,z_4z_5^{9}],[0,z_2z_5^{9}],[0,z_4^{2}z_5^{4}],[0,z_2z_4z_5^{4}],[0,z_2^{2}z_5^{4}],\\ [0,z_2z_4^{2}z_5^{2}]\}
\end{split}&
\end{flalign*}
\\(177)$\left\{ \begin{array}{l} z_1z_2+z_3z_4+z_5^{2}=0\\z_1z_3+z_2^{n}+z_4^{3}=0\end{array}\right.$
\\$n\ge3$
\\$(w_1,w_2,w_3,w_4,w_5;1,d)=(\frac{-3 + 4 n}{3 + 4 n},\frac{6}{3 + 4 n},\frac{3 + 2 n}{3 + 4 n},\frac{2 n}{3 + 4 n},\frac{1}{2};1,\frac{6 n}{3 + 4 n})$
\\$\mu=-2 + 8 n$
\begin{flalign*}
\begin{split}
mini=\{[z_4^{2},0],[0,z_2^{-2 + n}z_4z_5^{2}],[0,z_2^{-1 + n}z_4^{2}],[0,z_2^{-2 + n}z_3]\}
\end{split}&
\end{flalign*}
\\(178)$\left\{ \begin{array}{l} z_1z_2+z_3z_4+z_5^{2}=0\\z_1z_3+z_2^{n}+z_4^{4}=0\end{array}\right.$
\\$4\le n \le11$
\\$(w_1,w_2,w_3,w_4,w_5;1,d)=(\frac{-4 + 5 n}{4 + 5 n},\frac{8}{4 + 5 n},\frac{4 + 3 n}{4 + 5 n},\frac{2 n}{4 + 5 n},\frac{1}{2};1,\frac{8 n}{4 + 5 n})$
\\$\mu=-3 + 11 n$
\begin{flalign*}
\begin{split}
mini=\{[z_4^{3},0],[0,z_2^{-2 + n}z_4^{2}z_5^{2}],[0,z_2^{-1 + n}z_4^{3}],[0,z_2^{-2 + n}z_3]\}
\end{split}&
\end{flalign*}
\\(179)$\left\{ \begin{array}{l} z_1z_2+z_3z_4+z_5^{2}=0\\z_1z_3+z_2^{n}+z_4^{5}=0\end{array}\right.$
\\$5\le n \le7$
\\$(w_1,w_2,w_3,w_4,w_5;1,d)=(\frac{-5 + 6 n}{5 + 6 n},\frac{10}{5 + 6 n},\frac{5 + 4 n}{5 + 6 n},\frac{2 n}{5 + 6 n},\frac{1}{2};1,\frac{10 n}{5 + 6 n})$
\\$\mu=-4 + 14 n$
\begin{flalign*}
\begin{split}
mini=\{[z_4^{4},0],[0,z_2^{-2 + n}z_4^{3}z_5^{2}],[0,z_2^{-1 + n}z_4^{4}],[0,z_2^{-2 + n}z_3]\}
\end{split}&
\end{flalign*}
\\(180)$\left\{ \begin{array}{l} z_1z_2+z_3z_4+z_5^{3}=0\\z_1z_3+z_2^{n}+z_4^{3}=0\end{array}\right.$
\\$3\le n \le5$
\\$(w_1,w_2,w_3,w_4,w_5;1,d)=(\frac{-3 + 4 n}{3 + 4 n},\frac{6}{3 + 4 n},\frac{3 + 2 n}{3 + 4 n},\frac{2 n}{3 + 4 n},\frac{1}{3};1,\frac{6 n}{3 + 4 n})$
\\$\mu=-2 + 14 n$
\begin{flalign*}
\begin{split}
mini=\{[z_4^{2}z_5,0],[0,z_2^{-2 + n}z_4z_5^{4}],[0,z_2^{-1 + n}z_4^{2}z_5],[0,z_2^{-2 + n}z_3z_5]\}
\end{split}&
\end{flalign*}
\\(181)$\left\{ \begin{array}{l} z_1z_2+z_3z_4+z_5^{4}=0\\z_1z_3+z_2^{3}+z_4^{3}=0\end{array}\right.$
\\$(w_1,w_2,w_3,w_4,w_5;1,d)=(\frac{3}{5},\frac{2}{5},\frac{3}{5},\frac{2}{5},\frac{1}{4};1,\frac{6}{5})$
\\$\mu=58$
\begin{flalign*}
\begin{split}
mini=\{[z_4^{2}z_5^{2},0],[0,z_2z_4z_5^{6}],[0,z_2^{2}z_4^{2}z_5^{2}],[0,z_2z_3z_5^{2}]\}
\end{split}&
\end{flalign*}
\\(182)$\left\{ \begin{array}{l} z_1z_2+z_3z_4=0\\z_1z_3+z_2^{n}+z_4^{n}+z_5^{2}=0\end{array}\right.$
\\$n\ge3$
\\$(w_1,w_2,w_3,w_4,w_5;1,d)=(\frac{n}{2 + n},\frac{2}{2 + n},\frac{n}{2 + n},\frac{2}{2 + n},\frac{n}{2 + n};1,\frac{2 n}{2 + n})$
\\$\mu=(1 + n)^2$
\begin{flalign*}
\begin{split}
mini=\{[z_5,0],[z_4^{n},0],[0,z_2^{-1 + n}z_4^{-1 + n}],[0,z_2^{-2 + n}z_3]\}
\end{split}&
\end{flalign*}
\\(183)$\left\{ \begin{array}{l} z_1z_2+z_3z_4=0\\z_1z_3+z_2^{n}+z_4^{3}+z_5^{3}=0\end{array}\right.$
\\$n\ge3$
\\$(w_1,w_2,w_3,w_4,w_5;1,d)=(\frac{-3 + 4 n}{3 + 4 n},\frac{6}{3 + 4 n},\frac{3 + 2 n}{3 + 4 n},\frac{2 n}{3 + 4 n},\frac{2 n}{3 + 4 n};1,\frac{6 n}{3 + 4 n})$
\\$\mu=7 + 8 n$
\begin{flalign*}
\begin{split}
mini=\{[z_5^{2},0],[z_4^{2}z_5,0],[0,z_2^{-1 + n}z_4^{2}z_5],[0,z_2^{-1 + n}z_3z_5]\}
\end{split}&
\end{flalign*}
\\(184)$\left\{ \begin{array}{l} z_1z_2+z_3z_4=0\\z_1z_3+z_2^{n}+z_4^{4}+z_5^{3}=0\end{array}\right.$
\\$4\le n \le11$
\\$(w_1,w_2,w_3,w_4,w_5;1,d)=(\frac{-4 + 5 n}{4 + 5 n},\frac{8}{4 + 5 n},\frac{4 + 3 n}{4 + 5 n},\frac{2 n}{4 + 5 n},\frac{8 n}{3 (4 + 5 n)};1,\frac{8 n}{4 + 5 n})$
\\$\mu=9 + 10 n$
\begin{flalign*}
\begin{split}
mini=\{[z_5^{2},0],[z_4^{3}z_5,0],[0,z_2^{-1 + n}z_4^{3}z_5],[0,z_2^{-1 + n}z_3z_5]\}
\end{split}&
\end{flalign*}
\\(185)$\left\{ \begin{array}{l} z_1z_2+z_3z_4=0\\z_1z_3+z_2^{n}+z_4^{5}+z_5^{3}=0\end{array}\right.$
\\$5\le n \le7$
\\$(w_1,w_2,w_3,w_4,w_5;1,d)=(\frac{-5 + 6 n}{5 + 6 n},\frac{10}{5 + 6 n},\frac{5 + 4 n}{5 + 6 n},\frac{2 n}{5 + 6 n},\frac{10 n}{3 (5 + 6 n)};1,\frac{10 n}{5 + 6 n})$
\\$\mu=11 + 12 n$
\begin{flalign*}
\begin{split}
mini=\{[z_5^{2},0],[z_4^{4}z_5,0],[0,z_2^{-1 + n}z_4^{4}z_5],[0,z_2^{-1 + n}z_3z_5]\}
\end{split}&
\end{flalign*}
\\(186)$\left\{ \begin{array}{l} z_1z_2+z_3z_4=0\\z_1z_3+z_2^{n}+z_4^{3}+z_5^{4}=0\end{array}%
\right.$
\\$(w_1,w_2,w_3,w_4,w_5;1,d)=(\frac{34}{35},\frac{1}{35},\frac{5}{7},\frac{2}{7},\frac{17}{35};1,\frac{51}{35})$
\\$\mu=463$
\begin{flalign*}
\begin{split}
mini=\{[z_5^{2},0],[0,z_2^{50}z_5^{2}],[0,z_4^{5}],[0,z_2^{50}z_4^{4}],[0,z_3^{2}],[0,z_2^{50}z_3],[0,z_2^{99}]\}
\end{split}&
\end{flalign*}

\begin{proof}
It is easy to check that each singularities defined by pairs of polynomials in the list above are thee dimensional isolated rational complete intersection singularities.
By Lemma \ref{l3} (1), we know that for any $i \in \{ 1,2,3,4,5 \}$, one of the following cases occurs:
  \\ ~ (1a)~$z_i^{n_i}$ appears in $f_1$ for some $n_i$,
  \\ ~ (1b)~$z_i^{n_i}$ appears in $f_2$ for some $n_i$,
  \\ ~ (1c)~there exist $j_i, k_i \in \{1,2,3,4,5\}\setminus \{i\}~and~j_i \ne k_i$ such that $z_i^{n_i}z_{j_i}$ appears in $f_1$ for some $n_i$ and $z_i^{m_i}z_{k_i}$ appears in $f_2$ for some $m_i$.

For each $i \in \{1,2,3,4,5\}$, if one of (1a), (1b) and (1c) occurs, then  there are $3^5=243$ cases. If (1a) or (1b) occurs, one  monomial which appear in $f_1$ and $f_2$ can be determined. And if (1c) occurs, then two monomials which appear in $f_1$ and $f_2$ can be determined. Now we consider the following two cases:
\\(I)~There exists $i \in \{1,2,3,4,5\}$ such that (1c) occurs. Therefore more than 6 monomials in $f_1$ and $f_2$ are determined. Thus we get more than 6 equations of $w_1,\cdots,w_5 ,d$ (for instance, if we have $z_1^{n_1}z_2$ appears in $f_2$, then we have $n_1w_1+w_2=d$). So $(w_1,\cdots,w_5,d)$ is uniquely determined  by solving these 6 linear equations. And we have checked that each weight type $(w_1,\cdots,w_5;1,d)$ obtained by this way, which satisfies the rational condition $w_1+\cdots+w_5>1+d$ and the conditions listed in Corollary \ref{c2}, is the same as one of the weight types of the singularities in the list up to permutation of coordinates. 

More explicitly, for example, we treat the case that $z_1^{n_1},\cdots, z_4^{n_4}, z_5^{n_5}z_4$ appear in $f_1$ and  $z_5^{m_5}z_3$ appears in $f_2$. Then we can get $w_1=\frac{1}{n_1},\cdots, w_4=\frac{1}{n_4}, w_5=\frac{n_4-1}{n_4n_5}, d=\frac{m_5(n_4-1)}{n_4n_5}+\frac{1}{n_3}$ by solving the 6 corresponding linear equations. Without lose of generality, we may assume that $w_1 \ge w_2$. Since we have $w_1+\cdots+w_5>1+d$ and $d \ge 1$, so we conclude that $(n_1,\cdots,n_5,m_5)$ can only be one of the following cases:
\\(1)~$(2,2,u,v,1,1), \;2 \le u\le v$
\\(2)~$(2,3,u,v,1,1), \;2 \le u\le v \le 5$
\\(3)~$(2,4,u,v,1,1), \;2 \le u\le v \le 3$
\\(4)~$(2,5,u,v,1,1), \;2 \le u\le v \le 3$
\\(5)~$(2,u,2,2,1,1), \;u \ge 6$
\\(6)~$(3,u,2,2,1,1), \;3 \le u \le 5$
\\(7)~$(2,2,2,2,u,u), \;u \ge 2$
\\(8)~$(2,2,2,u,2,1), \;u \ge 2$
\\(9)~$(2,3,2,2,2,2)$.
\\Then we only consider the infinite cases (1), (5) and (7). The other finite cases can be checked easily.
\\For infinite case (1), we have $$w_1=w_2=\frac{1}{2}, w_3=\frac{1}{u}, w_4=\frac{1}{v}, w_5=1-\frac{1}{v},d=1-\frac{1}{v}+\frac{1}{u},~2 \le u\le v$$  By Corollary \ref{c2} (2), we have $d \in N(w_1,w_2)$ or there exists $k\in \{3,4,5\}$ such that $d-w_k\in N(w_1,w_2)$, it follows that one of following cases occurs: 
\\(i).~$d \in N(\frac{1}{2},\frac{1}{2})$ 
\\(ii).~$d-\frac{1}{u} \in N(\frac{1}{2},\frac{1}{2})$
\\(iii). ~$d-\frac{1}{v} \in N(\frac{1}{2},\frac{1}{2})$ 
\\(iv). ~$d+\frac{1}{v}-1 \in N(\frac{1}{2},\frac{1}{2})$. 

If (i) $d \in N(\frac{1}{2},\frac{1}{2})$ occurs, since $v \ge u \ge 2$, so we have $1 \le d =1-\frac{1}{v}+\frac{1}{u} <\frac{3}{2}$, thus $d=1$ and $u=v$. 

 If (ii) $d-\frac{1}{u} \in N(\frac{1}{2},\frac{1}{2})$ occurs, since $v \ge 2$, we have $\frac{1}{2} \le d-\frac{1}{u}=1-\frac{1}{v} < 1$, thus  $1-\frac{1}{v}= \frac{1}{2}$. Hence we have $v=2$. Since $2 \le u \le v$, so we have  $u=2$.
 
  If (iv) $d+\frac{1}{v}-1 \in N(\frac{1}{2},\frac{1}{2})$ occurs, i.e. $\frac{1}{u}=d+\frac{1}{v}-1  \in N(\frac{1}{2},\frac{1}{2})$.  Since  $0<\frac{1}{u} \le \frac{1}{2}$, we have $\frac{1}{u}=\frac{1}{2}$.  Therefore we have $u=2$.

\leavevmode \\If (iii) $d-\frac{1}{v} \in N(\frac{1}{2},\frac{1}{2})$ occurs, since $2 \le u \le v$, we have  $0<d-\frac{1}{v}=1-\frac{2}{v}+\frac{1}{u} < \frac{3}{2}$. Thus we have $d-\frac{1}{v}=1-\frac{2}{v}+\frac{1}{u}=\frac{1}{2}$ or 1. If $d-\frac{1}{v}=\frac{1}{2}$, notice that $d \ge 1$ and $2 \le u \le v$, thus we have $v=u=2$.  If $d-\frac{1}{v}=1-\frac{2}{v}+\frac{1}{u}=1$, we have $2u=v$. 
\\In conclusion, we have $2 \le u=v$ or $2=u \le v$ or $4 \le 2u=v$. So $(w_1,..,w_5;1,d)$ is same as one of the following cases:
\\case (1)~$2 \le u=v \Rightarrow (w_1,..,w_5;1,d)=(\frac{1}{2},\frac{1}{2},\frac{1}{n},\frac{1}{n},\frac{n-1}{n};1,1),~n \ge 2$
\\case (2)~$2=u \le v \Rightarrow (w_1,..,w_5;1,d)=(\frac{1}{2},\frac{1}{2},\frac{1}{2},\frac{1}{n},\frac{1}{2};1,1),~n \ge 2$
\\case (3)~$2u=v \Rightarrow (w_1,..,w_5;1,d)=(\frac{1}{2},\frac{1}{2},\frac{1}{n},\frac{1}{2n},\frac{2n-1}{2n};1, \frac{2n+1}{2n}),~n \ge 2$.
\\And we can see that case (1); case (2) and case (3) corresponds to the 1st $(n=2)$, 59th; 1th; and 55th, 60th singularities in the list respectively.

For infinite case (5), we have 

$$w_1=w_3=w_4=w_5=\frac{1}{2}, w_2=\frac{1}{u}, d=1.$$

For infinite case (7),  we have 

$$w_1=w_2=w_3=w_4=\frac{1}{2}, w_5=\frac{1}{2u}, d=1.$$
 it is easy to seen that $(w_1,..,w_5;1,d)$ of infinite cases (5) and (7) is included in the weight types of the 1th singularities in the list up to permutation of coordinates.
\\
\\(II)~For each $i \in \{1,2,3,4,5\}$, (1a) or (1b) occurs. Then there are only 5 monomials in $f_1$ and $f_2$ can be determined. In order to determine $(w_1,\cdots, w_5;d)$,  we need at least one more monomial included in $f_1$ and $f_2$. Since 5 monomials in $f_1$ and $f_2$ are known, so it is easy to seen that one of following will occurs:
\\(a)~there exists $\{i_1, i_2\} \subset \{1,2,3,4,5\}$ such that (1a) occurs when $i=i_1, i_2$
\\(b)~there exists $\{j_1, j_2\} \subset \{1,2,3,4,5\}$ such that (1b) occurs when $i=j_1, j_2$. 
\\If (a) occurs, then by Lemma \ref{l3} (2), we have $z_{i_1}^az_{i_2}^b$ appears in $f_2$ for some non-negative integer $a,b$ or there exist $k \in \{1,2,3,4,5\} \setminus \{i_1,i_2\}$ such that $z_kz_{i_1}^az_{i_2}^b$ appears in $f_2$ for some non-negative integer $a,b$. Thus we have 6 monomials in $f_1$ and $f_2$ are determined now.
\\If (b) occurs, then by Lemma \ref{l3} (2), we have $z_{j_1}^az_{j_2}^b$ appears in $f_1$ for some non-negative integer $a,b$ or there exist $k \in \{1,2,3,4,5\} \setminus \{j_1,j_2\}$ such that $z_kz_{j_1}^az_{j_2}^b$ appears in $f_1$ for some non-negative integer $a,b$. Thus there are 6 monomials in $f_1$ and $f_2$ are determined now.
\\More explicitly, let us consider the example that $z_1^{n_1}, z_2^{n_2},z_3^{n_3}$ appear in $f_1$, and $z_4^{n_4}, z_5^{n_5}$ appear in $f_2$. Then by Lemma \ref{l3} (2) we have $z_4^az_5^b$ appears in $f_1$ for some non-negative integers $a,b$ or there exist $k \in \{1,2,3\}$ such that $z_kz_4^az_5^b$ appears in $f_1$ for some non-negative integer $a,b$. Thus there are 6 monomials in $f_1$ and $f_2$ are determined.  It follows that $(w_1,\cdots,w_5; d)$ is determined as above. And we have checked that each weight type $(w_1,\cdots,w_5;1,d)$ gotten by this way, which satisfies the rational condition $w_1+\cdots+w_5>1+d$ and the conditions listed in Corollary \ref{c2}, is the same as one of the weight types of the singularities in the above list up to permutation of coordinates.

\end{proof}

\providecommand{\bysame}{\leavevmode\hbox to3em{\hrulefill}\thinspace}
\providecommand{\MR}{\relax\ifhmode\unskip\space\fi MR }
\providecommand{\MRhref}[2]{%
  \href{http://www.ams.org/mathscinet-getitem?mr=#1}{#2}
}
\providecommand{\href}[2]{#2}

\end{document}